\documentclass[journal,draftcls,onecolumn,12pt]{IEEEtran}
\usepackage[utf8]{inputenc} 
\usepackage[T1]{fontenc}
\usepackage{url}
\usepackage{ifthen}
\usepackage{cite}
\usepackage{color}
\usepackage{graphicx}
\graphicspath{{./figs/}}
\graphicspath{{.}{./figs/}}
\usepackage[cmex10]{amsmath}
\usepackage{amsfonts,amssymb} 
\usepackage{amsthm}
\usepackage{array}
\usepackage{mdwmath}
\usepackage{mdwtab}
\usepackage{eqparbox}
\usepackage{stfloats}
\usepackage[keeplastbox]{flushend} 
\newtheorem{theorem}{Theorem}
\newtheorem{lemma}{Lemma}
\newtheorem{proposition}{Proposition}
\newtheorem{corollary}{Corollary}
\newtheorem*{remark*}{Remark}

\IEEEoverridecommandlockouts
\begin{document}

	\title{The Degraded Discrete-Time Poisson Wiretap Channel}	

	\author{Morteza~Soltani,~\IEEEmembership{Student~Member,~IEEE},~and~Zouheir~Rezki,~\IEEEmembership{Senior~Member,~IEEE}
	\thanks{M.~Soltani was with the Department of Electrical and Computer Engineering, University of Idaho, Moscow, Idaho, USA, e-mail: solt8821@vandals.uidaho.edu, and Z. Rezki is with the Department of Electrical and Computer Engineering, University of California Santa Cruz, CA, USA, e-mail:zrezki@ucsc.edu. This work has been supported by the King Abdullah University of Science and Technology (KAUST), under a competitive research grant (CRG) OSR-2016-CRG5-2958-01. \newline Parts of this paper has been presented at the 2019 IEEE International Symposium on Information Theory (ISIT'2019), Paris, France, July 2019.}
}
\maketitle
	
\begin{abstract}
	This paper addresses the degraded discrete-time Poisson wiretap channel (DT--PWC) in an optical wireless communication system based on intensity modulation and direct detection. Subject to nonnegativity, peak- and average-intensity as well as bandwidth constraints, we study the secrecy-capacity-achieving input distribution of this wiretap channel and prove it to be unique and discrete with a finite number of mass points; one of them located at the origin. Furthermore, we establish that every point on the boundary of the rate-equivocation region of this wiretap channel is also obtained by a unique and discrete input distribution with finitely many mass points. In general, the number of mass points of the optimal distributions are greater than two. This is in contrast with the degraded continuous-time PWC when the signaling bandwidth is not restricted and where the secrecy capacity and the entire boundary of the rate-equivocation region are achieved by binary distributions. Furthermore, we extend our analysis to the case where only an average-intensity constraint is active. For this case, we find that the secrecy capacity and the entire boundary of the rate-equivocation region are attained by discrete distributions with countably \textit{infinite} number of mass points, but with finitely many mass points in any bounded interval. Additionally, we shed light on the asymptotic behavior of the secrecy capacity in the regimes where the constraints either tend to zero (low-intensity) or tend to infinity (high-intensity). In the low-intensity regime, we observe that: 1) when only the the peak-intensity constraint is active, the secrecy capacity scales quadratically in the peak-intensity; 2) when both peak- and average-intensity constraints are active with their ratio held fixed, the secrecy capacity again scales quadratically in the peak-intensity constraint; 3) when both peak- and average-intensity constraints are active and the peak-intensity is held fixed while the average-intensity tends to zero, the secrecy capacity scales linearly in the average-intensity constraint; 4) when only the average-intensity constraint is active and the channel gains of the legitimate receiver and the eavesdropper are identical, the secrecy capacity scales linearly in the average-intensity; 5) finally, when only the average-intensity constraint is active and the channel gains are different, the secrecy capacity scales, to within a constant, like $(\alpha_B-\alpha_E)\mathcal{E}\log\log\frac{1}{\mathcal{E}}$, where $\mathcal{E}$ is the average-intensity constraint and $\alpha_B$, and $\alpha_E$ are the legitimate receiver's and the eavesdropper's channel gains, respectively. In the high-intensity regime, we establish that under peak- and/or average-intensity constraints, the secrecy capacity is always upper bounded by a constant. This implies that the in this regime, the secrecy capacity does \textit{not} scale with the constraints and converges to a constant. 
\end{abstract}

\section{Introduction}
In this section, we first briefly outline the background motivating the current channel model. Next comes a brief overview of the related literature survey. Then, the paper's contributions are outlined.

\subsection{Background}
Intensity modulation and direct detection (IM-DD) is the simplest and the most commonly used technique for optical wireless communications. In this scheme, the channel input modulates the intensity of the emitted light. Thus, the input signal is proportional to the light intensity and is nonnegative. The receiver is usually equipped with a photodetector which absorbs integer number of photons and generates a real valued output corrupted by noise. Based on the distribution of the corrupting noise there exist several models for the underlying optical wireless communication channels. Free space optical (FSO) channels~\cite{5238736,Uysal-Book}, optical channels with input-dependent Gaussian noise~\cite{6121996,Uysal-Book}, and Poisson optical channels \cite{Uysal-Book,21284,217161,4729780} are the most widely used models for optical wireless communications. Among these models, the most accurate one that can capture most of the optical channel impairments is the Poisson model. The studies conducting research on Poisson optical channels are mainly categorized in two mainstreams. The first category considers the continuous-time Poisson model where the input signals can admit arbitrarily waveforms and there are no bandwidth constraints on the transmission. The second category concerns the discrete-time Poisson channel and deals with the cases where stringent transmission bandwidths are assumed. 

\subsection{Summary of Prior Work}
For the discrete-time Poisson channel, Shamai~\cite{217161} studied the single-user channel capacity and showed that the capacity-achieving distribution under nonnegativity, peak- and average-intensity constraints is discrete with a finite number of mass points. This specific structure of the capacity-achieving input distribution is also observed for other optical intensity channels, such as the input-independent Gaussian noise (also known as the free-space optical intensity channel) and the optical intensity channel with an input-dependent Gaussian  noise~\cite{1435651}. Furthermore, Cheraghchi \textit{et al.} studied the structure of the capacity-achieving input distribution of the discrete-time Poisson channel with nonnegativity and average-intensity constraints~\cite{8632953}. In particular, the authors proved that the capacity-achieving input distribution is discrete with the following properties: 1) the intersection of the support set of the optimal input distribution with any bounded set is finite; 2) the optimal support set itself is an unbounded set. 
In \cite{5773060,4729780}, authors provided asymptotic analysis of the channel capacity in the regimes where the peak- and/or average-intensity constraints tend to zero (low-intensity regime) or to infinity (high-intensity regime). The work in \cite{5773060} focused on characterizing the channel capacity in the low-intensity regime of an average-intensity constrained or an average- and peak-intensity constrained inputs and found upper and lower bounds which in some cases coincide. Additionally, authors in \cite{4729780} investigated the high-intensity behavior of the channel capacity for a peak- and average-intensity constrained inputs and presented tight bounds, thus fully characterizing the channel capacity in the high-intensity regime. Finally, for the discrete-time Poisson channel with an average-intensity constraint, Martinez provided an upper bound on the channel capacity that can accurately capture the high-intensity behavior of the channel capacity~\cite{Martinez}.

While the capacity of the discrete-time Poisson channel is generally unknown in closed-form, the capacity of the continuous-time Poisson channel where the signaling bandwidth is not restricted is known in closed-form \cite{21284,1056262}. For the peak-intensity constrained or peak- and average-intensity constrained inputs the capacity of the continuous-time Poisson channel is achieved by a binary distribution with mass points located at the origin and at the peak-intensity constraint~\cite{1056262}, however, the channel capacity of the average-intensity constrained input is \textit{infinite} and the capacity-achieving input is unknown~\cite{1056262}.   

The broadcast nature of optical wireless signals imposes a security challenge, especially in the presence of unauthorized eavesdroppers. This problem has been conventionally addressed by cryptographic encryption~\cite{Shannon1949} without considering the imperfections introduced by the communication channels. Wyner~\cite{Wyn75}, on the other hand, proved the possibility of secure communications without relying on encryption by introducing the notion of a degraded wiretap channel. This result was later generalized by Csiszar and Korner by dropping the degradedness assumption of the wiretap channel~\cite{1055892}. 
	
The wiretap channels are studied with respect to the rate-equivocation region, which is defined as the set of all rate pairs for which the transmitter can communicate confidential messages reliably with a legitimate receiver at a certain secrecy level against an eavesdropper~\cite{bb_2011}. A wiretap channel is called degraded when given the observations of the legitimate user, the observations of the eavesdropper are independent of the secret messages. For this type of channels, Wyner established that there exists a single-letter characterization for the rate-equivocation region~\cite{Wyn75}.

Authors in~\cite{7164335} studied the degraded Gaussian wiretap channel under amplitude and variance constraints, and prove that the entire rate-equivocation region of this wiretap channel is attained by discrete input distributions with finitely many mass points. Furthermore, the authors observed that the secrecy-capacity-achieving input distribution may not be identical to the capacity-achieving counterpart in general, resulting in a tradeoff between the rate and its equivocation. It is worth mentioning that the results pertaining to the Gaussian wiretap channel with amplitude and variance constraints can be directly applied to characterize the optimal distributions exhausting the entire rate-equivocation region of the FSO wiretap channel with peak- and average-intensity constraints. Furthermore, Dytso \textit{et al.} establish that the secrecy-capacity-achieving distribution of the FSO wiretap channel with an average-intensity constraint admits a countably infinite support set~\cite{8613368}. The authors also provide conditions for when the support set is or is not bounded.

The work in \cite{8399890} considers the degraded optical wiretap channel with input-dependent Gaussian noise under peak- and average-intensity constraints and verified the optimality of distributions with a finitely many mass points for attaining the entire boundary of the rate-equivocation region. Besides, the authors provided asymptotic behavior of the secrecy capacity in the low- and high-intensity regimes. For this wiretap channel, authors observed that, in general, there is a tradeoff between the rate and its equivocation. Finally, \cite{6294444} examined the degraded continuous-time Poisson wiretap channel (CT--PWC) under only a peak-intensity constraint and gave a closed-form expression for the secrecy capacity. Particularly, the authors showed that \textit{binary} input distributions with mass points located at the origin and the peak-intensity constraint along with a very short duty cycle exhaust the entire rate-equivocation region. 
\\
\subsection{Contributions}
In this work, we consider a degraded \textit{discrete}-time PWC (DT--PWC) which consists of a transmitter, a legitimate user and an eavesdropper. In this setup, the input signals are restricted to have finite bandwidths. This fact distinguishes the DT--PWC from its continuous-time counterpart, where input signals can have infinite bandwidths. Using an IM-DD system, the photodetectors at the legitimate user and the eavesdropper count the number of received photons and output signals that follow Poisson distributions. Here, the objective is to have secure communication with the legitimate user over a discrete-time Poisson channel while keeping the eavesdropper ignorant of the transmitted messages as much as possible. 

We start by the secrecy capacity of the degraded DT--PWC and employ the functional optimization problems addressed in, for example~\cite{Smith71a,217161,7164335,8399890}, to derive the necessary and sufficient optimality equations, also known as Karush-Kuhn-Tucker (KKT) conditions, that must be satisfied by an optimal solution. Using these equations, we confirm that a unique distribution with a countably finite number of mass points achieves the secrecy capacity of the degraded DT--PWC when only peak-intensity or both peak- and average-intensity constraints are active. This is done by providing a contradiction argument. We start by assuming, on the contrary, that the support set of the optimal solutions contains an infinite number of elements. Then recalling the Identity and Bolzano-Weierstrass Theorems from complex analysis we conclude that: 1) when the legitimate user's and the eavesdropper's channel gains are not identical, a nonnegative constant must be lower bounded by a logarithmically increasing function in $x$ where $x\geq 0$, which is a contradiction; 2) when the channel gains are identical, the nonnegative constant must be upper bounded by $-\infty$ and a contradiction occurs. Following along similar lines of the above mentioned analysis, we extend the optimality of distributions with a finite number of mass points to the entire boundary of the rate-equivocation.

Additionally, we investigate the secrecy capacity of the DT--PWC with nonnegativity and average-intensity constraints, and verify that a unique distribution with the following structural properties is secrecy-capacity-achieving: 1) the support set of the optimal solution contains a finitely many mass points in any bounded interval; 2) the support set of the optimal solution is an unbounded set. These two properties imply that the optimal distribution is discrete with countably infinite number of mass points, but with finitely many mass points in any bounded interval. The first property is shown by means of contradiction. We assume, on the contrary, that for some bounded interval, the intersection of the support set of the optimal solution and the bounded interval has an infinite number of mass points. Then, using the KKT conditions and invoking the Bolzano-Weierstrass and Identity Theorems from complex analysis, we find that a nonnegative constant is upper bounded by $-\infty$ which results in a contradiction. The second property is also shown through a contradiction approach. We assume that the optimal support set is bounded and we consider the following cases: 1) when legitimate user's and the eavesdropper's channel gains are not identical, our contradiction hinges on the fact that a linearly increasing function in $x$ must be lower bounded by another function which grows as fast as $x\log x$. This is not possible for large values of $x$ and hence a contradiction occurs; 2) when the channel gains are identical, we find that the Lagrangian multiplier must be lower bounded by a constant and thus, using the Envelope Theorem~\cite{524037}, we observe that the secrecy capacity would at least grow linearly in the average-intensity constraint. However, in Appendix~\ref{App-E} we establish that the secrecy capacity is always upper bounded by a constant for all values of the average-intensity. Therefore, the desired contradiction is reached and the result follows. Moreover, we show that every point on the boundary of the rate-equivocation region is also attained by a unique distribution with countably infinite number of mass points, but finitely many mass points in any bounded interval. This, in turn, implies that the capacity of the discrete-time Poisson channel with average-intensity constraint is also achieved by a discrete distribution with countably infinite number of mass points and settles down Shamai's conjecture in \cite{217161}.  For convenience, we summarize our contributions with respect to the structure of the optimal input distributions achieving the secrecy capacity and exhausting the entire rate-equivocation region of the DT--PWC in Table~\ref{structure}.

\begin{table}
		\caption{Structure of the Optimal Distributions Attaining the Secrecy Capacity and the Boundary of Rate-Equivocation Region.}
		\centering
		\resizebox{1\textwidth}{!}{
			\begin{tabular}{c|c}	\label{structure}
				Active constraints &Structure of Optimal Input Distributions\\
				\hline
				Peak& Discrete distributions with a finite number of mass points\\
				\hline
				Peak and average & Discrete distributions with a finite number of mass points\\
				\hline
				Average & Discrete distributions with countably infinite number of mass points, but with finitely many mass points in a bounded interval\\
	\end{tabular}}
\end{table}

Furthermore, we study the asymptotic behavior of the secrecy capacity in the low- and high-intensity regimes (i.e., the regimes where the peak- and/or the average-intensity constraints tend to zero or infinity, respectively), and fully characterize the secrecy capacity in these regimes. In the low-intensity regime, we find the closed-form expression of the secrecy capacity for the following cases: 1) when only the peak-intensity constraint is active; 2) when both peak- and average-intensity constraints are active with their ratio held fixed; 3) when both peak- and average-intensity constraints are active and the peak-intensity is held fixed while the average-intensity tends to zero; 4) when only the average-intensity constraint is active and the channel gains of the legitimate receiver and the eavesdropper are identical; 5) when only the average-intensity constraint is active and the channel gains are different.

For the first two cases, we find the secrecy capacity and the secrecy-capacity-achieving distribution. We observe that the secrecy capacity scales quadratically in the peak-intensity constraint and the secrecy-capacity-achieving input distribution is binary with mass points located at the origin and the peak-intensity constraint. We establish these results by deriving lower and upper bounds on the secrecy capacity and showing that these bounds coincide. We note that a valid upper bound on the secrecy capacity of the DT--PWC is the secrecy capacity of the CT--PWC across all intensity regimes. This is because in the continuous-time version, input signals are not restricted to have a finite transmission bandwidth and can admit arbitrary waveforms with a very large bandwidth. Thus, under the same constraints, i.e., peak- and/or average-intensity constraints, the secrecy capacity of the CT--PWC is always greater than that of the DT--PWC. Also, a legitimate lower bound on the secrecy capacity of the DT--PWC is the difference between the capacities of the legitimate user's and the eavesdropper's channels. 

For case 3, we also fully characterize the secrecy capacity and the secrecy-capacity-achieving distribution. In this case, we find that the secrecy capacity scales linearly in the average-intensity constraint. We establish these result by showing that the secrecy capacity is a concave function in the \textit{average-intensity} constraint and invoking the secrecy capacity per unit cost argument established by El-Halabi \textit{et al.} \cite{6584947}. However, we note that the secrecy capacity per unit cost argument does not lead to the characterization of the secrecy-capacity-achieving distribution \cite{6584947}. Therefore, by leveraging the fact that the secrecy-capacity-achieving input distribution must have a finite number of mass points (as discussed above), we evaluate the mutual information difference (i.e., the secrecy rate) for a binary input distribution with mass points located at the origin and the peak-intensity constraint with vanishingly small probability mass for the mass point at the peak-intensity constraint. We show that the secrecy rate induced by this specific binary distribution is identical to the secrecy capacity. Thus, we conclude that this specific binary input distribution attains the secrecy capacity. Additionally, we use the secrecy capacity per unit cost argument to find the closed-form expression of the secrecy capacity for case 4. Once again, we see that the secrecy capacity scales linearly in the average-intensity constraint. In this case, despite having a closed-form expression for the secrecy capacity, we do not characterize the secrecy-capacity-achieving distribution. This is because, as mentioned above, the optimal input distribution admits a countably infinite number of mass points and evaluating the secrecy rate for such a distribution is cumbersome. 

Finally, for case 5, we observe that the capacity per unit cost argument will not lead to useful results for finding the asymptotic secrecy capacity. We circumvent this issue by finding lower and upper bounds for the secrecy capacity. Thus, in this case, we analyze the asymptotic behavior of the secrecy capacity through these bounds. The lower bound is derived based on a binary input distribution which gives rise to a secrecy rate that grows like $\frac{(\alpha_B-\alpha_E)}{2}\mathcal{E}\log\log\frac{1}{\mathcal{E}}$, where $\mathcal{E}$ is the average-intensity constraint and $\alpha_B$, and $\alpha_E$ are the legitimate receiver's and the eavesdropper's channel gains, respectively. Also, we can upper bound the secrecy capacity by the capacity of a discrete-time Poisson channel under an average-intensity constraint. To this end, we invoke the results established by Lapidoth \textit{et al.} which shows that the channel capacity scales like $2(\alpha_B-\alpha_E)\mathcal{E}\log\log\frac{1}{\mathcal{E}}$ in the average-intensity constraint for vanishingly small $\mathcal{E}$~\cite[Proposition~2]{5773060}. This constitutes an asymptotic behavior for the upper bound of the secrecy capacity. As a result, the secrecy capacity of the DT--PWC with an average-intensity constraint and different channel gains scales, to within a constant, like $(\alpha_B-\alpha_E)\mathcal{E}\log\log\frac{1}{\mathcal{E}}$ in the low-intensity regime.

In the high-intensity regime, we establish that under peak- and/or average-intensity constraints, the secrecy capacity is always upper bounded by a constant. This implies that the in this regime, the secrecy capacity does \textit{not} scale with the constraints and converges to a constant. To establish this, we consider two cases: 1) when the channel gains of the legitimate receiver and the eavesdropper are identical; 2) when the channel gains are different. For case 1, we upper bound the secrecy capacity using the properties of entropy of a Poisson random variable and prove that across all intensity regimes, the secrecy capacity is upper bounded by a constant. For case 2, we invoke the duality upper bound expression for the conditional mutual information. We note that the duality upper bound expression for the mutual information  was introduced by Lapidoth \textit{et al.} in \cite{1237131,4729780} which provides an upper bound on the channel capacity. Using the duality bound expression, we find an output distribution which results to a constant upper bound on the secrecy capacity across all the intensity regimes. For convenience, we summarize our contributions with respect to the asymptotic analysis of the secrecy capacity in both low- and high-intensity regimes in Table~\ref{tab-single}. 

\begin{table}
	\caption{The Asymptotic Behavior of the Secrecy Capacity $C_{S}$ in Low- and High-intensity Regimes.}
	\centering
	\resizebox{1\textwidth}{!}{
		\begin{tabular}{c|c|c}	\label{tab-single}
			Active constraints & Low-intensity behavior& High-intensity behavior\\
			\hline
			Peak& Scales quadratically in peak& Does not scale in peak\\
			\hline
			Peak and average with fixed ratio&Scales quadratically in peak&Does not scale in peak or average\\
			\hline
			Peak and average with fixed peak and vanishingly small average& Scales linearly in average&Does not scale in peak or average\\
			\hline
			Average with equal channel gains & Scales linearly in average& Does not scale in average\\
			\hline
			Average with different channel gains & Scales like $\mathcal{E}\log\log\frac{1}{\mathcal{E}}$ & Does not scale in average\\
	\end{tabular}}
\end{table}

Finally, through our numerical inspections, we find that when peak-intensity or both peak- and average-intensity constraints are active, in general, the secrecy capacity and the capacity of the DT--PWC are not achieved by the same distribution. Therefore, there is a tradeoff between the rate and its equivocation. This is also true for the CT--PWC when peak-intensity or both peak- and average-intensity constraints are active~\cite{6294444}. It is worth mentioning that since with only an average-intensity constraint, the optimal input distribution admits a countably infinite number of mass points, numerical computation of the secrecy-capacity as well as the boundary of the rate-equivocation region is not feasible. Therefore, for these case, we only resort to providing the asymptotic analysis of the secrecy capacity in the low- and high-intensity regimes.

\subsection{Paper Organization}
The rest of the paper is structured as follows. The degraded DT--PWC is formally defined in Section~\ref{sec-PWC}. The main results of our work regarding the characterization of the optimal distributions attaining the secrecy capacity as well as the entire rate-equivocation region  along with the asymptotic behavior of the secrecy capacity in the low- and high-intensity regimes  are presented in Section~\ref{sec-mainresults}. Proofs of the main results are provided in Section~\ref{sec-proof}. Numerical results are shown in Section \ref{sec-numres}, and finally, conclusions are drawn in Section \ref{sec-conc}.

\section{The Degraded Discrete-Time Poisson Wiretap Channel}\label{sec-PWC}
We consider a practical optical wireless communication system where IM-DD is employed. In this setup, the channel input modulates the emitted light intensity from the light emitting diode (LED) at the transmitter and photodetectors are used for receiving the optical signal at the legitimate user's and the eavesdropper's receivers. 

In the considered wiretap channel, confidential data are transmitted by sending pulse amplitude modulated (PAM) intensity signals which are constant in discrete time slots of $\Delta$ seconds~\cite{217161}. 
This model is referred to as the DT--PWC where a bandwidth constraint is imposed on the input signals by constraining the signals to be rectangular PAM of duration $\Delta$ seconds. We note that in the limiting case where the pulse duration $\Delta$ converges to zero, i.e., $\Delta\rightarrow 0$, the DT--PWC becomes the CT--PWC. In this limiting case, the transmitted pulses are no longer required to be rectangular PAM signals and can admit any arbitrary waveforms. Notice that the results pertaining to the degraded CT--PWC have  been reported by Laourine \textit{et al.} in \cite{6294444}. Therefore, in this work, our mere focus is on addressing the problem of secure communications over the  DT--PWC, i.e., the case where $\Delta$ does \textit{not} approach zero.

\subsection{Channel Model}\label{sub-0}
In the DT--PWC, the receiver is modeled as a photon counter which generates an integer representing the number of received photons. Specifically, in each time slot of $\Delta$ seconds an input intensity $X$ is corrupted by the constant channel gains $\alpha_B$ and $\alpha_E$ and the combined impact of background radiation as well as the photodetectors' dark currents $\lambda_{B}$ and $\lambda_{E}$ at the legitimate user's and the eavesdropper's receivers, respectively. The channel outputs at the legitimate receiver and the eavesdropper are denoted by $Y$ and $Z$, respectively, and are random variables related to the number of received photon in $\Delta$ seconds. These channel outputs conditioned on the input signal obey the Poisson distributions with mean $(\alpha_B X + \lambda_{B})\Delta$ and $(\alpha_E X + \lambda_{E})\Delta$, respectively, i.e.,\cite[equation~16]{217161}
\begin{align}\label{eq-chan-B}
p_{Y\vert X}(y\vert x) &= e^{-\left(\alpha_B x+\lambda_B\right)\Delta}\,\frac{\left[(\alpha_B x+\lambda_B)\Delta\right]^{\,y}}{y!},~ y\in \mathbb{N},\\
p_{Z\vert X}(z\vert x) &= e^{-\left(\alpha_E x+\lambda_E\right)\Delta}\,\frac{\left[(\alpha_E x+\lambda_E)\Delta\right]^{\,z}}{z!},~ z\in \mathbb{N},
\label{eq-chan-E}
\end{align}
where $\mathbb{N}$ is the set of all nonnegative integers. It is worth mentioning that in this work, we assume that the dark currents of the legitimate receiver and the eavesdropper are positive constants, i.e., $\lambda_{B} > 0$ and $\lambda_{E} > 0$.

In the DT--PWC, the channel input $X$ is a nonnegative random variable representing the intensity of the optical signal. Since intensity is constrained due to practical and safety restrictions by peak- and average-intensity constraints, the input must satisfy~\cite{Uysal-Book}
\begin{align}\label{eq-cnts-1}
0 \leq X &\leq \mathcal{A},\\
\mathbb{E}[X] &\leq \mathcal{E}.\label{eq-cnts-2}
\end{align}

In this work, we are interested in the \textit{degraded} DT--PWC. Therefore, we are interested in the case where the following conditions hold
\begin{align}\label{eq-deg-1}
\alpha_B &\geq \alpha_E, \\
\frac{\lambda_{B}}{\alpha_B} &\leq \frac{\lambda_{E}}{\alpha_E},\label{eq-deg-2}
\end{align}
which implies that the random variables $X$, $Y$, and $Z$ form the Markov chain $X\rightarrow Y\rightarrow Z$ and consequently, the DT--PWC becomes stochastically degraded~\cite{Wyn75,1255549,6294444}. In the sequel, without loss of generality, we consider that at least one of the inequalities \eqref{eq-deg-1} or \eqref{eq-deg-2} is strict. This is because if both are tight, then the legitimate receiver's and eavesdropper's channels become identical and the secrecy capacity (defined later in this section) is equal to zero.

\subsection{The Rate-Equivocation Characterization of the DT--PWC}\label{sub-1}
An $(n,2^{nR})$ code for the DT--PWC consists of the random variable $W$ (message set) uniformly distributed over $\mathcal{W} = \{1,2,\cdots,2^{nR}\}$, an encoder at the transmitter $f_n: \mathcal{W}\rightarrow \mathbb{R}_{+}^n$ satisfying the constraints \eqref{eq-cnts-1}--\eqref{eq-cnts-2}, and a decoder at the legitimate user $g_n: \mathbb{N}^{n} \rightarrow \mathcal{W}$. Equivocation of a code is measured by the normalized conditional entropy $\frac{1}{n}\,H(W\rvert Z^n)$. The probability of error for such a code is defined as $P_e^n = \Pr\left[g_n(Y^n)\neq W\right]$. A rate-equivocation pair $(R,R_e)$ is said to be achievable if there exists an $(n,2^{nR})$ code satisfying 
\begin{align}
\lim_{n\rightarrow \infty}P_e^n &= 0, \\
R_e &\leq \lim_{n\rightarrow \infty} \dfrac{1}{n}\,H(W\rvert Z^n),
\end{align}
where $H(W\vert Z^n)$ is the conditional entropy of $W$ given the observations $Z^n$.
The rate-equivocation region consists of all achievable rate-equivocation pairs. A rate $R$ is said to be perfectly secure if we have $R_e = R$, that is, if there exists an $(n,2^{nR})$ code satisfying $\lim_{n\rightarrow \infty}\frac{1}{n}\,I(W;Z^n) = 0$, where $I(W;Z^n)$ is the mutual information between the random variables $W$ and $Z^n$. The supremum of such rates is defined to be the secrecy capacity and is denoted by $C_{S}$. 

Since under the assumptions \eqref{eq-deg-1}--\eqref{eq-deg-2}, the DT--PWC is degraded, its entire rate-equivocation region, denoted by $\mathcal{R}$, can be expressed in a single-letter expression and it is given by the union of all rate-equivocation pairs $(R,R_e)$ such that~\cite{Wyn75}
\begin{equation}
	\begin{cases}
		0 \leq R \leq I(X;Y),\\
		0 \leq R_e \leq I(X;Y) - I(X;Z),
	\end{cases}
\end{equation}
for some input distribution $F_X \in \mathcal{F}^{+}$ where the feasible set $\mathcal{F}^{+}$ is given by one of the following sets
\begin{align}\label{eq-FeasibleSet1}
\Omega^{+}_{\mathcal{A},\,\mathcal{E}} &\stackrel{\triangle}{=}\left\{F_X: \int_{0}^{\mathcal{A}}dF_X(x)=1,\, \int_{0}^\mathcal{A}x\,dF_X(x)\leq \mathcal{E}\right\},\\
\Omega^{+}_{\mathcal{A}} &\stackrel{\triangle}{=}\left\{F_X: \int_{0}^{\mathcal{A}}dF_X(x)=1\right\}, \label{eq-FeasibleSet2}\\
\Omega^{+}_{\mathcal{E}} &\stackrel{\triangle}{=}\left\{F_X: \int_{0}^{\infty}dF_X(x)=1,\,\int_{0}^{
\infty}x\,dF_X(x)\leq \mathcal{E}\right\}. \label{eq-FeasibleSet3}
\end{align}

\section{Main Results}\label{sec-mainresults}
In this section, we present our main results regarding the structure of the optimal input distributions achieving the secrecy capacity and exhausting the entire rate-equivocation region of the degraded DT--PWC. Furthermore, we characterize the behavior of the asymptotic secrecy capacity in the low- and high-intensity regimes.


\subsection{Structure of the Secrecy-Capacity-Achieving Distributions}
For the degraded DT--PWC, the secrecy capacity is given by a single-letter expression as \cite[Chap.~3]{bb_2011,217161}
\begin{equation}\label{eq-SecCap}
C_S = \sup_{F_X\in\mathcal{F}^{+}} f_0(F_X) \stackrel{\triangle}{=}  \sup_{F_X\in\mathcal{F}^{+}}[I(X;Y) - I(X;Z)],
\end{equation}
where the feasible set $\mathcal{F}^{+}$ is given by one of the sets in \eqref{eq-FeasibleSet1}--\eqref{eq-FeasibleSet3}.

We start by characterizing the secrecy-capacity-achieving distribution when $\mathcal{F}^{+} = \Omega^{+}_{\mathcal{A},\,\mathcal{E}}$ in \eqref{eq-SecCap}, i.e., when both peak- and average-intensity constraints are active. In this case, we observe that the solution to the optimization problem in~\eqref{eq-SecCap} exists, is unique and is discrete with finitely many mass points in the interval $[0,\mathcal{A}]$. This is formally stated by the following theorem.
\begin{theorem}\label{theo-1}
	There exists a unique input distribution that attains the secrecy capacity of the DT--PWC with nonnegativity, peak- and average-intensity constraints. Furthermore, the support set of this optimal input distribution is a finite set.
\end{theorem}
\begin{proof}
	For convenience, the proof is presented in Section~\ref{sec-proof}.
\end{proof}
The proof of Theorem~\ref{theo-1} is sketched as follows. Firstly, the set of input distributions $\Omega^{+}_{\mathcal{A},\,\mathcal{E}}$ is shown to be sequentially compact in the L\'evy metric sense and convex. Secondly, it is shown that the objective functional is continuous, weakly differentiable and strictly concave in $F_X$. Thus, a unique solution to \eqref{eq-SecCap} exists. Thirdly, the necessary and sufficient KKT conditions that must be satisfied by an optimal solution $F_X^*$ are derived. Fourthly, it is established that the support set of $F_X^*$ contains finitely many mass points. This is done by providing a contradiction argument. We start by assuming, on the contrary, that the support set contains an infinite number of elements. Next, we invoke the Identity and Bolzano-Weierstrass Theorems from complex analysis and we conclude that: 1) when the legitimate user's and the eavesdropper's channel gains are not identical, the Lagrangian multiplier (which is a nonnegative constant) must be lower bounded by a logarithmically increasing function in $x$ which is a contradiction; 2) when the channel gains are identical, the Lagrangian multiplier is upper bounded by $-\infty$ which again is a contradiction. Following along similar lines of the proof of Theorem~\ref{theo-1}, we extend the optimality of distributions with a finite number of mass points to the entire boundary of the rate-equivocation.

It is worth mentioning that in the CT--PWC studied in \cite{6294444}, the secrecy-capacity-achieving input distribution is always binary with mass points located at the origin and the value of the peak-intensity constraint~\cite[Theorem~1]{6294444}. Furthermore, to achieve the secrecy capacity, input signals must have a very short duty cycle (i.e., $\Delta\rightarrow 0$ or equivalently, a very large transmission bandwidth is required). However, in the DT--PWC the number of mass points of the optimal distribution depends on the value of $\Delta,\,\mathcal{A},\,\mathcal{E}$, and in general, it is greater than two.

Next, we present a corollary which concerns the characterization of the optimal distribution attaining the secrecy capacity of the DT--PWC with nonnegativity and peak-intensity constraints.
\begin{corollary}
	The secrecy capacity of the DT--PWC with nonnegativity and peak-intensity constraints, i.e., the case when $\mathcal{F}^{+} = \Omega^{+}_{\mathcal{A}}$ in \eqref{eq-SecCap}, is achieved by a unique and discrete input distribution with a finite number of mass points.
\end{corollary}
\begin{proof}
	The proof follows along similar lines of those mentioned in the proof of Theorem~\ref{theo-1}.
\end{proof}

Next, we consider the case where $\mathcal{F}^{+} = \Omega^{+}_{\mathcal{E}}$ in \eqref{eq-SecCap} and establish that a discrete distribution with countably infinite number of mass points, but with finitely many mass points in any bounded interval, achieve the secrecy capacity when nonnegativity and average-intensity constraints (\textit{no} peak-intensity constraint) are active. 
\begin{theorem}\label{theo-2}
	There exists a unique input distribution which attains the secrecy capacity of the DT--PWC with nonnegativity and average-intensity constraints. The optimal distribution is discrete with countably infinite number of mass points, but only finitely many mass points in any bounded interval.
\end{theorem}
\begin{proof}
	Theorem~\ref{theo-2} is established in Section~\ref{sec-proof}.
\end{proof}
To prove Theorem~\ref{theo-2}, we first prove that the set of input distributions $\Omega^{+}_{\mathcal{E}}$ is compact and convex. We then invoke similar arguments to those presented in the proof of Theorem~\ref{theo-1} to show that the objective function in~\eqref{eq-SecCap} is continuous, strictly concave and weakly differentiable in the input distribution $F_X$. Therefore, we conclude that the solution to the optimization problem~\eqref{eq-SecCap} exists and is unique. We continue the proof by showing that first, the intersection of the support set of the optimal input distribution denoted by $\mathcal{S}_{F_X^*}$ with any bounded interval $B$ contains a finite number of mass points, i.e., $\lvert \mathcal{S}_{F_X^*}\cap B \rvert < \infty$, where $\lvert B\rvert$ denotes the cardinality of the set $B$. Next, we show that $\mathcal{S}_{F_X^*}$ must be an unbounded set. These structural properties imply that the optimal distribution is discrete with countably infinite number of mass points, but with finitely many mass points in any bounded interval. The first property is shown by means of contradiction. We assume that $\lvert \mathcal{S}_{F_X^*}\cap B \rvert = \infty$. Then, using the KKT conditions and invoking the Bolzano-Weierstrass and Identity Theorems from complex analysis, we find that the Lagrangian multiplier is upper bounded by $-\infty$ which is a contradiction. The second property is also shown through contradiction. Assuming that the optimal support set is bounded, we consider the following cases: 1) if the legitimate user's and the eavesdropper's channel gains are not identical, our contradiction hinges on the fact that a linearly increasing function in $x$ must be lower bounded by another function which grows as fast as $x\log x$ which is a contradiction for large values of $x$; 2) if the channel gains are identical, we find that the Lagrangian multiplier would be lower bounded by a constant and using the Envelope Theorem~\cite{524037}, we observe that the secrecy capacity must at least grow linearly in the average-intensity constraint. However, in Appendix~\ref{App-E} we establish that the secrecy capacity is always upper bounded by a constant for all values of the average-intensity. Therefore, the desired contradiction occurs. 

Finally, we establish the existence of a mass point at $x = 0$ in the support set of the secrecy-capacity-achieving input distributions under all the possible choices for $\mathcal{F}^{+}$ given in \eqref{eq-FeasibleSet1}--\eqref{eq-FeasibleSet3}.
\begin{proposition}\label{prop-0}
	Let $\mathcal{S}_{F_X^*}$ be the support set of the secrecy-capacity-achieving input distribution $F_X^*$ for the DT--PWC under one of the constraints in \eqref{eq-FeasibleSet1}--\eqref{eq-FeasibleSet3}. Then $x=0$ always belong to $\mathcal{S}_{F_X^*}$.
\end{proposition}
\begin{proof}
	The proof is by contradiction and follows along similar lines of \cite[Proposition~1]{8399890} with the difference that the conditional channel laws follow Poisson distribution. For completeness, the proof is relegated to Appendix~\ref{App-A1}.
\end{proof}

Next, we present a corollary which establishes that the support set of the \textit{capacity-achieving} input distribution of the discrete-time Poisson channel (the case without the secrecy constraint) under each of the constraints \eqref{eq-FeasibleSet1}--\eqref{eq-FeasibleSet3} possesses a mass point at the origin.
\begin{corollary}
	The capacity-achieving distribution of the discrete-time Poisson channel, i.e., the case without secrecy constraint, under nonnegativity, peak- and/or average-intensity constraints has a mass point located at the origin. 
\end{corollary} 
\begin{proof}
	The proof is via contradiction and it follows along similar lines of the proof of Proposition~\ref{prop-0} without a secrecy constraint, i.e., disregarding the the eavesdropper's link and its observations.       
\end{proof}
It is worth mentioning that this result provides an alternative proof of the existence of a mass point at the origin which was previously established in \cite[Corollary~2]{6685986}.

\subsection{Structure of the Optimal Distributions Exhausting the Entire Rate-Equivocation Region}
By a time-sharing argument, it can be shown that the rate-equivocation region of the DT--PWC is convex. Therefore, the region can be characterized by finding tangent lines to $\mathcal{R}$ which are given by the solutions of
\begin{equation}
\sup_{F_X\in\mathcal{F}^{+}}f_{\mu}(F_X) \stackrel{\triangle}{=} \sup_{F_X\in\mathcal{F}^{+}}\left[ \mu I(X;Y) +  (1-\mu)[I(X;Y)-I(X;Z)]\right], \quad \forall~ \mu\in [0,1],
\label{eq-equivocregion}
\end{equation}
where the feasible set $\mathcal{F}^{+}$ is one of the sets given by \eqref{eq-FeasibleSet1}--\eqref{eq-FeasibleSet3}.
We start by proving that the entire boundary of the rate-equivocation region of the DT--PWC with nonnegativity, peak- and average-intensity constraints (i.e., $\mathcal{F}^{+} = \Omega^{+}_{\mathcal{A},\,\mathcal{E}}$) is obtained by discrete input distributions with a finite number of mass points. 
\begin{theorem}\label{theo-3}
	Every point on the boundary of the rate-equivocation region of the DT--PWC with nonnegativity, peak- and average-intensity constraints, is achieved by a unique input distribution which is discrete with a finite number of mass points.	
\end{theorem}
\begin{proof}
	For convenience, Theorem~\ref{theo-3} is established in Section~\ref{sec-proof}.
\end{proof}
The proof of Theorem~\ref{theo-3} follows along similar lines as the one in the proof of Theorem~\ref{theo-1} with the difference in the contradiction argument. Here, our contradiction is based on the fact that (regardless of having $\alpha_B = \alpha_E$ or not) the Lagrangian multiplier is lower bounded by a function that grows logarithmically in $x$. 

Next, we present a corollary which states that the entire boundary of the rate-equivocation region of the DT--PWC under nonnegativity and peak-intensity constraints is attained by discrete distributions with finitely many mass points.  
\begin{corollary}
	Every point on the boundary of the rate-equivocation region of the DT--PWC with nonnegativity and peak-intensity constraints is achieved by a unique and discrete input distribution with a finite number of mass points.
\end{corollary}
\begin{proof}
	The proof follows by invoking similar arguments to those in the proof of Theorem~\ref{theo-3}.
\end{proof}

Finally, we consider the case where $\mathcal{F}^{+} = \Omega^{+}_{\mathcal{E}}$ in \eqref{eq-SecCap} and characterize the optimal distributions exhausting the entire rate-equivocation region when nonnegativity and average-intensity constraints are active. 
\begin{theorem}\label{theo-4}
	Every point on the boundary of the rate-equivocation region of the DT--PWC with nonnegativity and average-intensity constraints is achieved by a unique and discrete input distribution with countably infinite number of mass points, but finitely many mass points in any bounded interval.
\end{theorem}
\begin{proof}
	The proof is presented in Section~\ref{sec-proof}.
\end{proof}
The proof of Theorem~\ref{theo-4} follows along similar lines of the proof of Theorem~\ref{theo-2} with a difference in the unboundedness proof of the optimal support set. Here, we do not consider different cases on the channel gains and the desired contradiction occurs by showing that a linearly increasing function in $x$ would be lower bounded by another function growing as fast as $x\log x$.

A direct consequence of Theorem~\ref{theo-4} is that when $\mu = 1$ in \eqref{eq-equivocregion} (the point corresponding to the capacity of the discrete time Poisson channel with nonnegativity and average-intensity constraints), the optimal distribution is discrete with a countably infinite number of mass points, but finitely many mass points in any bounded interval. This result settles down Shamai's conjecture in \cite{217161} using different and simpler arguments than those that appeared in \cite[Theorem~15]{8632953}. 

\begin{remark*}
	In this work, although we have assumed that $\lambda_{B}$ and $\lambda_{E}$ are positive constants, our results pertaining to the structural properties of the optimal input distributions achieving the secrecy capacity and exhausting the entire rate-equivocation region of the DT--PWC can be easily extended to the case where $\lambda_{B} = 0$, and $\lambda_{E}\geq 0$. For completeness, we present the proof regarding this specific case in Appendix~\ref{App-last}.  
\end{remark*}

\subsection{Asymptotic Behavior of the Secrecy Capacity in the Low- and High-Intensity Regimes}\label{sec-asympt-peak}

This section investigates the asymptotic analysis for the secrecy capacity of the DT--PWC in both low- and high-intensity regimes. 

\subsubsection{Low-Intensity Results} 
We begin the asymptotic analysis of the secrecy capacity for asymptotically small values of $\mathcal{A}$ and $\mathcal{E}$. To achieve this, we consider five different cases as follows: 1) both peak- and average-intensity constraints are active, and they tend to zero while their ratio is held fixed at $p \stackrel{\triangle}{=} \frac{\mathcal{E}}{\mathcal{A}}$, where $\,0<p\leq 1$; 2) only the peak-intensity constraint is active and it tends to zero; 3) both peak- and average-intensity constraints are active and the peak-intensity constraint is held fixed while the average-intensity constraint tends to zero, i.e., $p\rightarrow 0$; 4) only the average-intensity constraint is active and it approaches zero along with the fact that the channel gains of the legitimate receiver and the eavesdropper are identical; 5) only the average-intensity constraint is active and it tends to zero, and the channel gains of the legitimate receiver and the eavesdropper are different. The following theorems present a full characterization of the secrecy capacity in the low-intensity regimes for the aforementioned cases. 


\begin{theorem}\label{theo-8}
In the regime where the peak-intensity constraint $\mathcal{A} \rightarrow 0$ or both peak- and average-intensity constraints  $\mathcal{A} \rightarrow 0,~\mathcal{E}\rightarrow 0$ while their ratio is held fixed at $p = \frac{\mathcal{E}}{\mathcal{A}}$, the asymptotic secrecy capacity satisfies 
\begin{equation}\label{eq-CS-LIR-PA}
\lim_{\mathcal{A}\rightarrow 0}\frac{C_S}{\mathcal{A}^2} = 
\begin{cases}
\frac{1}{8}\left(\frac{\alpha_B^2}{\lambda_B} - \frac{\alpha_E^2}{\lambda_E}\right),\quad&\text{if} ~\frac{1}{2} \leq p \leq 1,\\
\frac{1}{2}\,p\,(1-p)\left(\frac{\alpha_B^2}{\lambda_B} - \frac{\alpha_E^2}{\lambda_E}\right),\quad&\text{if}~ 0 < p < \frac{1}{2}.
\end{cases}
\end{equation}
\end{theorem}
\begin{proof}
	The proof is based on deriving lower and upper bounds that asymptotically coincide in the low-intensity regime. The lower bound is based on evaluating the mutual information difference between the legitimate receiver and the eavesdropper for a binary input distribution with mass points at $\{0,\mathcal{A}\}$ and corresponding probability masses $\{\frac{1}{2},\frac{1}{2}\}$ when only the peak-intensity constraint is active, and $\{1-p,p\},\,p\in(0,\frac{1}{2})$ when both peak- and average-intensity constraints are active. The upper bound is given by the secrecy capacity of the CT-PWC with peak- or both peak- and average-intensity constraints. For convenience, the derivation of these bounds are presented in Appendix~\ref{App-LowIntensity}. 
\end{proof}
Theorem~\ref{theo-8} implies that when $\frac{1}{2} \leq p \leq 1$, the average-intensity constraint is inactive and only the peak-intensity constraint is active. Furthermore, this theorem shows that the asymptotic secrecy capacity scales \textit{quadratically} in the peak-intensity constraint when only peak- or both peak- and average-intensity constraints are active. Additionally, we observe that the asymptotic secrecy capacity is independent of the pulse duration $\Delta$ and thus, there is no tradeoff between the secrecy capacity and the transmission bandwidth. In other words, in this case, only the amplitude level of the transmitted pulses affect the secrecy capacity and not the pulse duration.  


\begin{theorem}\label{theo-actv-p-avg}
	When both peak- and average-intensity constraint are active and in the regime where the peak-intensity constraint $\mathcal{A}$ is held fixed while the average-intensity constraint $\mathcal{E}\rightarrow 0$, the asymptotic secrecy capacity satisfies
	\begin{equation}\label{eq-theo-9}
	\lim_{\mathcal{E}\rightarrow 0}\frac{C_S}{\mathcal{E}} = \left[\left(\alpha_B + \frac{\lambda_{B}}{\mathcal{A}}\right)\log\left(1 + \frac{\alpha_B\mathcal{A}}{\lambda_{B}}\right)- \left(\alpha_E + \frac{\lambda_{E}}{\mathcal{A}}\right)\log\left(1 + \frac{\alpha_E\mathcal{A}}{\lambda_{E}}\right) + (\alpha_E - \alpha_B)\right].
	\end{equation}
\end{theorem}
\begin{proof}	
	First, we note that the RHS of \eqref{eq-theo-9} is strictly positive. This is due to the fact that the function
	\begin{equation}
		\Phi(x)\stackrel{\triangle}{=} \left[\left(\alpha_B + \frac{\lambda_{B}}{x}\right)\log\left(1 + \frac{\alpha_Bx}{\lambda_{B}}\right)- \left(\alpha_E + \frac{\lambda_{E}}{x}\right)\log\left(1 + \frac{\alpha_Ex}{\lambda_{E}}\right) + (\alpha_E - \alpha_B)\right]
	\end{equation}
	is strictly increasing over the interval $x\in[0,\mathcal{A}]$. This is established in Appendix~\ref{App-theo-avg-zero}. This implies that $\Phi(\mathcal{A}) > \Phi(0)$. Note that here $\Phi(0)$ is not defined, but $\lim_{x\rightarrow 0}\Phi(x) = 0$.
	 
	We continue the proof by showing that the secrecy capacity is a concave function in the average-intensity constraint (regardless of whether the peak-intensity constraint is active or not). Next, we invoke the secrecy capacity per unit cost argument established by El-Halabi \textit{et al.} \cite{6584947} to find a closed-form expression of the secrecy capacity. However, we note that the secrecy capacity per unit cost argument does not lead to the characterization of the secrecy-capacity-achieving input distribution~\cite{6584947}. Therefore, by leveraging the fact that the secrecy-capacity-achieving input distribution must have a finite number of mass points (as shown in Theorem~\ref{theo-1}), we evaluate the mutual information difference for a binary input distribution with mass points at $\{0,\mathcal{A}\}$ with corresponding probability masses $\{1-p,p\}$, where $p = \frac{\mathcal{E}}{\mathcal{A}}$ and $\mathcal{E}\rightarrow 0$. Finally, we show that this specific binary distribution achieves the asymptotic secrecy capacity. The details of the proof are relegated to Appendix~\ref{App-LowerBoundTight}.
\end{proof}

From Theorem~\ref{theo-actv-p-avg}, we infer that in the low-intensity regime, the asymptotic secrecy capacity scales \textit{linearly} in the average-intensity constraint. Furthermore, similar to the results proved in Theorem~\ref{theo-8}, when the peak-intensity constraint is held fixed while the average-intensity approaches zero, the secrecy capacity is independent of the pulse duration $\Delta$. Consequently, in this regime, there is no tradeoff between the secrecy capacity and the transmission bandwidth and only the amplitude levels of the transmitted signals affect the secrecy capacity. 

\begin{theorem}\label{theo-avg-zero}
		When only an average-intensity constraint is considered and in the regime where the average-intensity $\mathcal{E}\rightarrow 0$ along with identical channel gains $(\alpha_B = \alpha_E)$, the asymptotic secrecy capacity satisfies 
		\begin{equation}\label{eq-theo-10}
			\lim_{\mathcal{E}\rightarrow 0}\frac{C_S}{\mathcal{E}} = \left[\alpha_B  \log\left(\frac{\lambda_{E}}{\lambda_{B}}\right)\right].
		\end{equation}
\end{theorem}
\begin{proof}
	The proof follows along similar lines of the proof of Theorem~\ref{theo-actv-p-avg} with the difference that the peak-intensity constraint is now inactive, i.e., $\mathcal{A} = +\infty$. Moreover, we do not characterize the optimal distribution that attains the secrecy capacity. This is because, as established by Theorem~\ref{theo-2}, the optimal distribution admits a countably infinite number of mass points and evaluating the mutual information difference for such a distribution is onerous. For brevity, the proof is presented in Appendix~\ref{App-theo-avg-zero}.
\end{proof}
Similar to Theorem~\ref{theo-actv-p-avg} results, here too, the asymptotic secrecy capacity scales linearly in the average-intensity constraint. 

\begin{theorem}\label{theo-avg-diff-zero}
	With only an average-intensity constraint and different channel gains, i.e., $\alpha_B > \alpha_E$, the asymptotic secrecy capacity, in the regime $\mathcal{E}\rightarrow 0$, satisfies
	\begin{equation}
		\frac{1}{2} \leq \lim_{\mathcal{E}\rightarrow 0}\frac{C_S}{\,(\alpha_B - \alpha_E)\mathcal{E}\log\log\frac{1}{\mathcal{E}}} \leq 2.
	\end{equation}
\begin{proof}
	We establish Theorem~\ref{theo-avg-diff-zero} by providing lower and upper bounds on the secrecy capacity. The lower bound is based on evaluating the mutual information difference for the binary input distribution with mass points located at $\{0,\zeta\}$ with corresponding probability masses $\{1-p,p\}$, where $\zeta\stackrel{\triangle}{=} \sqrt{\frac{\lambda_{B}}{\alpha_B^2\Delta}\log \frac{1}{\mathcal{E}}}$ and $p = \frac{\mathcal{E}}{\zeta}$. Furthermore, we upper bound the secrecy capacity of the DT--PWC under an average-intensity constraint by the capacity of another discrete-time Poisson channel whose input is $X$ and whose output is $\widetilde{Z}$ with $p_{\widetilde{Z}\vert X}(\widetilde{z}\vert x) = e^{-(\widetilde{\alpha}x + \widetilde{\lambda})\Delta}\frac{[(\widetilde{\alpha}x + \widetilde{\lambda})\Delta]^{\widetilde{z}}}{\widetilde{z}!}$, where $\widetilde{\alpha} \stackrel{\triangle}{=} \alpha_B - \alpha_E$, $\widetilde{\lambda} \stackrel{\triangle}{=} \left(\frac{\alpha_B}{\alpha_E}-1\right)\lambda_E$, and the input is subject to nonnegativity and average-intensity constraint $\mathbb{E}[X]\leq \mathcal{E}$. We derive the upper bound by invoking the results found by Lapidoth \textit{et al.} pertaining to the asymptotic capacity of the discrete-time Poisson channel with an average-intensity constraint and with constant nonzero dark current~\cite[Proposition~2]{5773060}. For convenience, the details of the proof are relegated to Appendix~\ref{App-avg-diff-zero}.
\end{proof}
\end{theorem}

Theorem~\ref{theo-avg-diff-zero} suggests that the asymptotic secrecy capacity scales, to within a constant, like $\mathcal{E}\log\log\frac{1}{\mathcal{E}}$ in the average-intensity constraint $\mathcal{E}$ in the low-intensity regime when the channel gains are different.

Now that we have fully analyzed the asymptotic behavior of the secrecy capacity of the degraded DT--PWC under a variety of constraints in the low-intensity regime, we turn our focus to provide asymptotic analysis in the high-intensity regime. 

\subsubsection{High-Intensity Results}
This section sheds light on the asymptotic behavior of the secrecy capacity of the DT--PWC when the constraints tend to infinity. We start by considering two scenarios based on the degradedness conditions in \eqref{eq-deg-1}--\eqref{eq-deg-2} and for each of these scenarios, we provide an upper bound on the secrecy capacity. The first scenario deals with the case where the inequality \eqref{eq-deg-1} is tight and the inequality \eqref{eq-deg-2} is strict, i.e., $\alpha_B = \alpha_E,\, \frac{\lambda_{E}}{\alpha_E} > \frac{\lambda_{B}}{\alpha_B}$. The second scenario refers to the case where the inequality \eqref{eq-deg-1} is strict and the inequality \eqref{eq-deg-2} is either strict or tight, i.e., $\alpha_B > \alpha_E,\, \frac{\lambda_{E}}{\alpha_E} \geq \frac{\lambda_{B}}{\alpha_B}$. We find that the secrecy capacity of the DT--PWC for both of these scenarios can be upper bounded by a constant across all intensity regimes. This implies that the secrecy capacity does not scale with the constraints in the high-intensity regime, and therefore, it must be a constant value.

Before we present the main results regarding the asymptotic behavior of the secrecy capacity in the high-intensity regime, we state a lemma which we use in our analysis throughout this section.
\begin{lemma}\label{lemma-6}
	For a degraded DT--PWC (i.e., when the conditions in \eqref{eq-deg-1}--\eqref{eq-deg-2} hold true), the mutual information difference $f_0(F_X) = I(X;Y) - I(X;Z)$ can be upper bounded as
	\begin{align}
	f_0(F_X) &= I(X;Y) - I(X;\widetilde{Y}) + I(X;\widetilde{Y}) - I(X;Z)\notag\\ &\leq I(X;Y) - I(X;\widetilde{Y}) + I(X;\widetilde{Z}).
	\end{align}
	\noindent where $\widetilde{Y}\stackrel{\triangle}{=} Y + N_D$, with $N_D$ being a Poisson distributed random variable with mean $\lambda_D\Delta$ independent of $X$ and $Y$, where $\lambda_D \stackrel{\triangle}{=} \frac{\alpha_B}{\alpha_E}\lambda_E - \lambda_B$. Moreover, $\widetilde{Z}\vert X$ is a Poisson random variable with mean $(\widetilde{\alpha} X + \widetilde{\lambda})\Delta$ independent of $Z\vert X$ and such that $\widetilde{Y}\vert X = Z\vert X + \widetilde{Z}\vert X$, where $\widetilde{\alpha} = \alpha_B - \alpha_E$ and $\widetilde{\lambda} = \left(\frac{\alpha_B}{\alpha_E} - 1\right)\lambda_E$.
\end{lemma}
\begin{proof}
	The proof follows along a similar line of~\cite[Lemma~1, Lemma~7]{6294444}.
\end{proof}
Now, we are ready to present the asymptotic results of the secrecy capacity in the high-intensity regime.

\paragraph{Upper Bound on the Secrecy Capacity When $\alpha_B = \alpha_E$ and $ \frac{\lambda_{E}}{\alpha_E} > \frac{\lambda_{B}}{\alpha_B}$}
We start by noting that according to Lemma~\ref{lemma-6}, the random variable $\widetilde{Z}\equiv 0$ and the secrecy capacity of the DT--PWC is upper bounded by 
\begin{equation}\label{eq-f_0}
C_S = f_0(F_X^*) \, = \, H_Y(F_X^*) - H_{\widetilde{Y}}(F_X^*) +  H_{\widetilde{Y}\vert X}(F_X^*) - H_{Y\vert X}(F_X^*),
\end{equation}
where $F_X^* \in \mathcal{F}^{+}$ with $\mathcal{F}^{+}$ being one of the feasible sets defined in \eqref{eq-FeasibleSet1}--\eqref{eq-FeasibleSet3}, and $H_Y(F_X^*)$ and $H_{\widetilde{Y}}(F_X^*)$ are the entropies of the discrete random variables $Y$ and $\widetilde{Y}$, respectively, induced by the optimal input distribution $F_X^*$. Furthermore, $H_{Y\vert X}(F_X^*)$ and $H_{\widetilde{Y}\vert X}(F_X^*)$ are the conditional entropies of $Y\vert X$ and $\widetilde{Y}\vert X$, respectively, induced by $F_X^*$. Next, we present the upper bound on the secrecy capacity of the DT--PWC in the high-intensity regime.
\begin{theorem}\label{theo-9}
	The secrecy capacity of the DT--PWC with either of the considered constraints in \eqref{eq-FeasibleSet1}--\eqref{eq-FeasibleSet3} and under the assumption of $\alpha_B = \alpha_E,\, \frac{\lambda_{E}}{\alpha_E} > \frac{\lambda_B}{\alpha_B}$ is upper bounded by
	\begin{equation}
	C_S \leq \frac{\frac{\lambda_D^2}{2}+\frac{\lambda_D}{\Delta}}{\lambda_B},
	\label{eq-SecCap-Upp}
	\end{equation}
with $\lambda_D = \lambda_E - \lambda_B$.
\end{theorem}
\begin{proof}
	Under the assumption of $\alpha_B = \alpha_E$ and $\frac{\lambda_{E}}{\alpha_E} > \frac{\lambda_B}{\alpha_B}$, we show in Appendix~\ref{App-E} that $H_Y(F_X^*) - H_{\widetilde{Y}}(F_X^*) < 0$. Thus, to upper bound the secrecy capacity, it is sufficient to provide an upper bound for the term $H_{\widetilde{Y}\vert X}(F_X^*) - H_{Y\vert X}(F_X^*)$. For convenience, the details of the proof are presented in Appendix~\ref{App-E}.
\end{proof}
From Theorem~\ref{theo-9}, we notice that the upper bound in \eqref{eq-SecCap-Upp} holds for all values of the peak- and/or average-intensity constraints. This implies that the secrecy capacity of the DT--PWC in the high-intensity regime, where either of the constraints $\mathcal{A}\rightarrow\infty$ or $\mathcal{E}\rightarrow\infty$, does not scale with the constraints and approaches a positive constant, i.e., 
\begin{equation}
C_S = O(1).
\end{equation}

\paragraph{Upper Bound on the Secrecy Capacity When $\alpha_B > \alpha_E$ and $ \frac{\lambda_{E}}{\alpha_E} \geq \frac{\lambda_{B}}{\alpha_B}$}
In this case, we first note that the secrecy capacity can be upper bounded as
\begin{align}
C_S &= \sup_{F_X\in\mathcal{F}^{+}}[I(X;Y) - I(X;Z)]\notag\\
&= \sup_{F_X\in\mathcal{F}^{+}}[I(X;Y) - I(X;\widetilde{Y})+I(X;\widetilde{Y}) - I(X;Z)]\notag\\
&\leq \underbrace{\sup_{F_X\in\mathcal{F}^{+}}[I(X;Y) - I(X;\widetilde{Y})]}_{\stackrel{\triangle}{=}\,C_{S,U1}}+\underbrace{ \sup_{F_X\in\mathcal{F}^{+}} [I(X;\widetilde{Y}) - I(X;Z)]}_{\stackrel{\triangle}{=}\,C_{S,U2}},
\end{align}
where $\widetilde{Y} = Y + N_D$ according to Lemma~\ref{lemma-6}, and the last inequality follows from the property of supremum. 

Now, are ready to upper bound the secrecy capacity. To this end, we provide an upper bound for each of the terms $C_{S,U1}$ and $C_{S,U2}$ and show that these upper bounds are constant values and do not scale with the peak- and/or average-intensity constraints. These results are formally stated by the following theorem.
\begin{theorem}\label{theo-10}
	The secrecy capacity of the DT--PWC with either of the considered constraints in \eqref{eq-FeasibleSet1}--\eqref{eq-FeasibleSet3} and under the assumption of $\alpha_B > \alpha_E$ and $ \frac{\lambda_{E}}{\alpha_E} \geq \frac{\lambda_B}{\alpha_B}$ is upper bounded by
	\begin{equation}
	C_S \leq \frac{\frac{\lambda_D^2}{2}+\frac{\lambda_D}{\Delta}}{\lambda_B} + \frac{1}{\Delta} \log\left(\frac{\alpha_B}{\alpha_E}\right),
	\end{equation}
	where $\lambda_D = \frac{\alpha_B}{\alpha_E}\lambda_E - \lambda_B$.
\end{theorem}
\begin{proof}
	We start the proof by noting that according to Theorem~\ref{theo-9}, $C_{S,U1}$ can be readily upper bounded by a constant value as
	\begin{equation}\label{eq-SecCap-high-up1}
	C_{S,U1} \leq \frac{\frac{\lambda_D^2}{2}+\frac{\lambda_D}{\Delta}}{\lambda_B},
	\end{equation}
	where $\lambda_D = \frac{\alpha_B}{\alpha_E} \lambda_E - \lambda_B$.
	
	 We continue the proof by upper bounding $C_{S,U2}$. To this end, we first note that $C_{S,U2}$ is the secrecy capacity of a degraded DT--PWC whose input is $X$, and whose outputs are $\widetilde{Y}$ and $Z$. Observe that $\widetilde{Y}\vert X$ is a Poisson distributed random variable with mean $(\alpha_B X + \frac{\alpha_B}{\alpha_E} \lambda_E)\Delta$. Also, $Z\vert X$ is another Poisson distributed random variable with mean $(\alpha_E X + \lambda_E)\Delta$. Observe that the observations of the eavesdropper, i.e., $Z$ is obtained from $\widetilde{Y}$ by thinning with erasure probability $1 - \frac{\alpha_E}{\alpha_B}$~\cite{1255549,6294444}. Next, we note that this new DT--PWC is degraded because the conditions in \eqref{eq-deg-1}--\eqref{eq-deg-2} are met since $\alpha_B > \alpha_E$ and $\frac{\lambda_E}{\alpha_E} = \frac{\frac{\alpha_B}{\alpha_E}\lambda_E}{\alpha_B}$. As a result, we have that $I(X;\widetilde{Y}\vert Z) = I(X;\widetilde{Y})-I(X;Z)$ and $C_{S,U2} =  \sup_{F_X\in\mathcal{F}^{+}}I(X;\widetilde{Y}\vert Z)$. By applying the duality upper bound expression found in \cite{4729780,6121996} to the conditional mutual information $I(X;\widetilde{Y}\vert Z)$, we find an upper bound on the secrecy capacity as
	 \begin{equation}\label{eq-SecCap-high-up2}
	 C_{S,U2} \leq \frac{1}{\Delta} \log\left(\frac{\alpha_B}{\alpha_E}\right).
	 \end{equation}
	 For brevity, we present the remainder of the proof details in Appendix~\ref{App-H}.  
\end{proof}
Since the upper bounds \eqref{eq-SecCap-high-up1}--\eqref{eq-SecCap-high-up2} are constant values and do not scale with the peak- and/or average-intensity constraints, the direct consequences of Theorem~\ref{theo-9} and Theorem~\ref{theo-10} are that the secrecy capacity of the DT--PWC in the high-intensity regime converges to a positive constant, i.e., 
\begin{equation}
C_S = O(1).
\end{equation}

\section{Proof of The Main Results}\label{sec-proof}
In this section, we first provide the required preliminaries for the development of the main results. We then give the detailed proofs of theorems and the proposition mentioned in Sec.~\ref{sec-mainresults}.

\subsection{Preliminaries} 
Since both legitimate user's and eavesdropper's channels are discrete-time Poisson channels, the output densities for $Y$ and $Z$ exist for any input distribution $F_X$, and are given by
\begin{align} \label{eq-Outpu-B}
P_Y(y;F_X) &= \int_{0}^{\mathcal{A}}p(y\lvert x)\,dF_X(x), ~y\in\mathbb{N},\\
P_Z(z;F_X) &= \int_{0}^{\mathcal{A}}p(z\lvert x)\,dF_X(x), ~z\in\mathbb{N},
\label{eq-Outpu-E}
\end{align} 
where $p(y\lvert x)$ and $p(z\lvert x)$ are given by \eqref{eq-chan-B}--\eqref{eq-chan-E}.
We define the secrecy rate density $c_s(x;F_X)$ as
\begin{equation}\label{eq-equivocDensityDefinition}
c_S(x;F_X) \stackrel{\triangle}{=} i_B(x;F_X) - i_E(x;F_X),
\end{equation}
where $i_B(x;F_X)$ and $i_E(x;F_X)$ are the mutual information densities for the legitimate user's and the eavesdropper's channels, respectively, and are as follows
\begin{align} \label{eq-mutualinfodensity-Bob}
i_B(x;F_X)=&\frac{1}{\Delta}\sum_{y=0}^{+\infty} p(y\vert x)\,\log\frac{p(y\vert x)}{P_Y(y;F_X)},\\
i_E(x;F_X)=&\frac{1}{\Delta}\sum_{z=0}^{+\infty} p(z\vert x)\,\log\frac{p(z\vert x)}{P_Z(z;F_X)}.
\label{eq-mutualinfodensity-Eve}
\end{align}
Plugging \eqref{eq-chan-B}--\eqref{eq-chan-E} and \eqref{eq-Outpu-B}--\eqref{eq-Outpu-E} into \eqref{eq-mutualinfodensity-Bob}--\eqref{eq-mutualinfodensity-Eve} and after some algebra, we get
\begin{align}\label{eq-ib}
&i_B(x;F_X) = (\alpha_Bx+\lambda_B)\log[(\alpha_Bx+\lambda_B)\Delta] - \alpha_B x -\frac{1}{\Delta}\sum_{y=0}^{+\infty} p(y\vert x)\log g_B(y;F_X),\\
&i_E(x;F_X) = (\alpha_E x+\lambda_E)\log[(\alpha_E x+\lambda_E)\Delta] - \alpha_E x -\frac{1}{\Delta}\sum_{z=0}^{+\infty} p(z\vert x)\log g_E(z;F_X),\label{eq-ie}
\end{align}
where $g_B(y;F_X)$ and $g_E(z;F_X)$ are respectively defined as
\begin{align}
g_B(y;F_X)&\stackrel{\triangle}{=} \int_{0}^{\mathcal{A}}e^{-\alpha_B x\Delta}\,[(\alpha_Bx+\lambda_B)\Delta]^{\,y}\, dF_X(x), \\ 
g_E(z;F_X)&\stackrel{\triangle}{=} \int_{0}^{\mathcal{A}}e^{-\alpha_E x\Delta}\,[(\alpha_E x+\lambda_E)\Delta]^{\,z}\, dF_X(x).
\end{align}
Furthermore, we have the following identities
\begin{align}
I(X;Y) &= \int_{0}^{\mathcal{A}} i_B(x;F_X)\,dF_X(x)\stackrel{\triangle}{=}I_B(F_X),\\
I(X;Z) &= \int_{0}^{\mathcal{A}} i_E(x;F_X)\,dF_X(x)\stackrel{\triangle}{=}I_E(F_X),\\
f_0(F_X) &= \int_{0}^{\mathcal{A}} c_S(x;F_X)\,dF_X(x).
\end{align}
Next, we prove Theorem~\ref{theo-1} using the preliminaries provided in this section.
\subsection{Proof of Theorem~\ref{theo-1}}
 We start by proving that the set of input distributions $\Omega^{+}_{\mathcal{A},\,\mathcal{E}}$ is compact and convex. We then show that the objective functions $f_0(F_X)$ in~\eqref{eq-SecCap} is continuous, strictly concave and weakly differentiable in the input distribution $F_X$ and hence, we conclude that the optimization problems in~\eqref{eq-SecCap} has a unique solutions. We continue the proof by deriving the necessary and sufficient conditions (KKT conditions) for the optimality of the optimal input distribution $F_X^*$. Finally, by means of contradiction we show that the optimal input distributions are discrete with a finite number of mass points. The proof is then streamlined into a few lemmas which we state below.
\begin{lemma}\label{lem-set}
	The feasible set $\Omega^{+}_{\mathcal{A},\,\mathcal{E}}$ is convex and sequentially compact in the Levy metric sense.
\end{lemma}
\begin{proof}
	The proof follows along similar lines as \cite[Lemma~1]{217161}
\end{proof}

\begin{lemma}\label{lem-cont}
	The functional $f_0:\Omega^{+}_{\mathcal{A},\,\mathcal{E}}\rightarrow \mathbb{R},~f_0(F_X) = I_B(F_X)-I_E(F_X)$ is continuous in $F_X$.
\end{lemma}
\begin{proof}
	The proof follows along similar lines as presented in \cite[Lemma~3]{217161}.
\end{proof}
From Lemma~\ref{lem-set} and Lemma~\ref{lem-cont}, $f_0(F_X)$ is continuous in $F_X$ over $\Omega^{+}_{\mathcal{A},\,\mathcal{E}}$ which itself is a compact set, then by the Extreme Value Theorem, $f_0(F_X)$ is bounded above and attains its supremum. That is, the supremum in \eqref{eq-SecCap} is actually a maximum which is achievable by at least one input distribution $F_X$.
\begin{lemma}\label{lemm-conc}
	The functional $f_0(F_X)$ is strictly concave in $F_X$.
\end{lemma}
\begin{proof}
	The proof is by contradiction and follows along similar lines as in \cite[Appendix~A]{8399890} with the difference that the conditional channel laws follow Poisson distribution. For completeness, the proof is relegated to Appendix~\ref{app-strict}.
\end{proof}
Lemma~\ref{lemm-conc} implies that the answer to the optimization problem in \eqref{eq-SecCap} for $\mathcal{F}^{+} = \Omega^{+}_{\mathcal{A},\,\mathcal{E}}$, denoted by $F_X^*$, is unique.
\begin{lemma}\label{lem-diff}
	The functional $f_0(F_X)$ is weakly differentiable in $\Omega^{+}_{\mathcal{A},\,\mathcal{E}}$ and its weak derivative at the point $F_X^o$, denoted by $f_0^{\prime}(F_X^o)$ is given by
	\begin{equation}
	f_0^{\prime}(F_X,F_X^o) \stackrel{\triangle}{=}\lim_{t\rightarrow 0} \frac{f_0((1-t)F_X^o+tF_X)-f_0(F_X^o)}{t}= \int_{0}^{\mathcal{A}} c_S(x;F_X^o)\,dF_X(x) - f(F_X^o),
	\end{equation}
	where $t\in[0,1]$.
\end{lemma}
\begin{proof}
	The proof is based on the definition of the weak derivative and follows along similar lines as the one in \cite{8399890}.
\end{proof}
From Lemma \ref{lem-set}, Lemma~\ref{lemm-conc}, and Lemma \ref{lem-diff}, we have a strictly concave and weak-differentiable function $f_0(F_X)$ over $\Omega^{+}_{\mathcal{A},\,\mathcal{E}}$ which is a convex set, then the necessary and sufficient conditions for an input distribution $F_X^*$ to be optimal is
\begin{equation}
f_0^{\prime}(F_X,F_X^*) \leq 0,\quad \forall~ F_X,\,F_X^*\in\Omega^{+}_{\mathcal{A},\,\mathcal{E}}.
\end{equation}
Now, we define the mapping 
\begin{equation}
g(F_X) = \int_{0}^{\mathcal{A}}x\,dF_X(x) - \mathcal{E},
\end{equation}
from $\Omega^{+}_{\mathcal{A},\,\mathcal{E}}$ to $\mathbb{R}$. This mapping is linear in $F_X$ and hence convex. Furthermore, the weak-derivative of $g(F_X)$ at the point $F_X^o$ is given by
\begin{equation}
g^{\prime}(F_X,F_X^o) = g(F_X) - g(F_X^o).
\end{equation}
Using the Lagrangian Theorem, and noting that $f_0(F_X) - \gamma g(F_X)$ (where $\gamma \geq 0$ is the Lagrangian coefficient) is weakly differentiable and strictly concave in $F_X$, the necessary and sufficient conditions for $F_X^*\in\Omega^{+}_{\mathcal{A},\,\mathcal{E}}$ to be optimal is
\begin{equation}
f_0^{\prime}(F_X,F_X^*) - \gamma g^{\prime}(F_X,F_X^*) \leq 0,\quad \forall~ F_X,\,F_X^*\in\Omega^{+}_{\mathcal{A},\,\mathcal{E}},
\end{equation}
that is
\begin{equation}
\int_{0}^{\mathcal{A}} \left[c_S(x;F_X^*)-\gamma x\right]\,dF_X(x) \leq  C_S - \gamma \mathcal{E},
\end{equation}
where the secrecy capacity is $C_{S} = f_0(F_X^*) = I_B(F_X^*) - I_E(F_X^*)$.
Next, we present a theorem which states the KKT conditions for the optimality of $F_X^*\in\Omega^{+}_{\mathcal{A},\,\mathcal{E}}$.
\begin{theorem}\label{theo-5}
	Let $\mathcal{S}_{F_X^*}\subset[0,\mathcal{A}]$ be the support set of $F_X^*$, then
	\begin{equation}\label{eq-KKT-integral}
	\int_{0}^{\mathcal{A}} \left[c_S(x;F_X^*)-\gamma x\right]\,dF_X(x) \leq  C_S - \gamma \mathcal{E},
	\end{equation}
	for all $F_X\in\Omega^{+}_{\mathcal{A},\,\mathcal{E}}$ if and only if
	\begin{align}
	c_S(x;F_X^*) - \gamma x &\leq C_{S} - \gamma\mathcal{E}, \quad\forall~ x\in[0,\mathcal{A}] \label{eq-KKT-1},\\
	c_S(x;F_X^*) -\gamma x&= C_{S}- \gamma\mathcal{E}, \quad\forall~ x\in \mathcal{S}_{F_X^*}\label{eq-KKT-2}.
	\end{align}
\end{theorem}
\begin{proof}
	The implication from \eqref{eq-KKT-1} to \eqref{eq-KKT-integral} is immediate. For the converse, assume \eqref{eq-KKT-1} is false. Then there exists an $\hat{x}$ such that 
	\begin{equation}
	c_S(\hat{x};F_X^*) > C_S + \gamma(\hat{x}-\mathcal{E}).
	\end{equation}
	If $F_X(x) = u(x-\hat{x})$, where $u(\cdot)$ is the unit step function, then 
	\begin{equation}
	\int_{0}^{\mathcal{A}} \left[c_S(x;F_X^*)-\gamma x\right]\,dF_X(x) = c_S(\hat{x};F_X^*)-\gamma \hat{x} > C_S - \gamma \mathcal{E},
	\end{equation}
	which contradicts \eqref{eq-KKT-integral}. Now, assume that \eqref{eq-KKT-1} is true, but \eqref{eq-KKT-2} is false, i.e., there exists $\hat{x}\in\mathcal{S}_{F_X^*}$ such that 
	\begin{equation}
	c_S(\hat{x};F_X^*) < C_S + \gamma(\hat{x}-\mathcal{E}).
	\end{equation}
	Since all the functions in the above equation are continuous in $x$, the inequality is satisfied strictly on a neighborhood $\mathcal{S}^{\prime}$ of $\hat{x}$. Now, by definition of a support set, the set $\mathcal{S}^{\prime}$ necessarily satisfies $\int_{\mathcal{S}^{\prime}}dF_X^*(x) = \epsilon \in [0,1]$. Hence, 
	\begin{align}
	C_S-\gamma\mathcal{E} = f_0(F_X^*) - \gamma\mathcal{E} &= \int_{0}^{\mathcal{A}}[c_S(x;F_X^*)-\gamma x]\,dF_X^*(x) \notag\\
	&=  \int_{\mathcal{S}^{\prime}}[c_S(x;F_X^*)-\gamma x]dF_X^*(x) + \int_{\mathcal{S}_{F_X^*}-\mathcal{S}^{\prime}}[c_S(x;F_X^*)-\gamma x]dF_X^*(x)\notag\\
	&<\epsilon(C_S - \gamma\mathcal{E}) + (1-\epsilon)(C_S - \gamma\mathcal{E}) < (C_S - \gamma\mathcal{E}),
	\end{align}
	which is a contradiction, and hence the result follows.
\end{proof}

We now prove by contradiction that the secrecy-capacity-achieving input distribution $F_X^*$ has a finite number of mass points. To reach a contradiction, we use the KKT conditions in~\eqref{eq-KKT-1}--\eqref{eq-KKT-2}. To this end, the following lemma establishes that both $i_B(x;F_X)$ and $i_E(x;F_X)$ have analytic extensions over some open connected set in the complex plane $\mathbb{C}$. 
\begin{lemma}\label{lemma-5}
	The secrecy rate density $c_S(x;F_X)-\gamma x$ has an analytic extension to the open connected set $\mathcal{O}\stackrel{\triangle}{=}\lbrace w\in\mathbb{C}: \Re(w) > -\frac{\lambda_B}{\alpha_B}\rbrace$, where $\Re(w)$ is the real part of the complex variable $w$.
\end{lemma}
\begin{proof}
	The mutual information densities $i_B(w;F_X)$ and $i_E(w;F_X)$ have analytic extension to the open connected sets $\mathcal{O}_B\stackrel{\triangle}{=}\lbrace w\in\mathbb{C}: \Re(w) > -\frac{\lambda_B}{\alpha_B}\rbrace$ and $\mathcal{O}_E\stackrel{\triangle}{=}\lbrace w\in\mathbb{C}: \Re(w) > -\frac{\lambda_E}{\alpha_E}\rbrace$, respectively, according to \cite{217161}. Therefore, the secrecy rate density $c_S(w;F_X)-\gamma w$ has an analytic extension to the open connected set $\mathcal{O} = \mathcal{O}_B \cap \mathcal{O}_E$. Since $\frac{\lambda_E}{\alpha_E} \geq \frac{\lambda_B}{\alpha_B}$ (based on \eqref{eq-deg-2}), we have $\mathcal{O} = \mathcal{O}_B$. This completes the proof of Lemma~\ref{lemma-5}.
\end{proof}
Now, we are ready to prove the discreteness and finiteness of the support set of $F_X^*$ using a contradiction argument. We start by assuming that $\mathcal{S}_{F_X^*}$ has an infinite number of elements. In view of the optimality condition \eqref{eq-KKT-2}, the analyticity of $c_S(w;F_X)-\gamma w$ over $\mathcal{O}$ and the Identity Theorem from complex analysis along with Bolzano-Weierstrass Theorem, if $\mathcal{S}_{F_X^*}$ has an infinite number of mass points, we deduce that $r_e(w;F_X^*)-\gamma w = C_S - \gamma\mathcal{E}$ for all $w\in\mathcal{O}$. Since $(-\frac{\lambda_B}{\alpha_B},+\infty) \subset \mathcal{O}$, we conclude that 
\begin{equation}\label{eq-cntrdct}
	c_S(x;F_X^*) - \gamma x = C_S - \gamma\mathcal{E},\quad \forall\,x>-\frac{\lambda_B}{\alpha_B}.
\end{equation}
Next, we show that \eqref{eq-cntrdct} results in a contradiction. Observe that \eqref{eq-cntrdct} implies that $c_S(x;F_X^*)-\gamma x$ is a constant function in $x$ for all $x\in(-\frac{\lambda_B}{\alpha_B},+\infty)$. Therefore, to reach a contradiction, we show that $c_S(x;F_X^*)-\gamma x$ is not a constant function over this interval. To that end, we take the derivative of both sides of \eqref{eq-cntrdct} with respect to $x$ and we find
\begin{equation}
\frac{dc_S(x;F_X^*)}{dx} = \gamma,\quad \forall\, x>-\frac{\lambda_B}{\alpha_B}.\label{eq-deriv}
\end{equation}
Substituting \eqref{eq-ib}--\eqref{eq-ie} into \eqref{eq-equivocDensityDefinition} and taking the derivative with respect to $x$, we can write
\begin{align}
\frac{dc_S(x;F_X^*)}{dx} =&\, \alpha_B\log[(\alpha_Bx + \lambda_{B})\Delta] +\alpha_B \sum_{y=0}^{+\infty}p(y\vert x)\log\frac{g_B(y;F_X^*)}{g_B(y+1;F_X^*)}\notag \\
&\,-\alpha_E\log[(\alpha_Ex + \lambda_{E})\Delta] - \alpha_E\sum_{z=0}^{+\infty}p(z\vert x)\log\frac{g_E(z;F_X^*)}{g_E(z+1;F_X^*)},\,\forall\, x>-\frac{\lambda_B}{\alpha_B}.\label{eq-deriv-main}
\end{align}
It can be easily shown that
\begin{align}\label{eq-gB}
\lambda_{B}\Delta &\leq \frac{g_B(y+1;F_X^*)}{g_B(y;F_X^*)}\leq (\alpha_B\mathcal{A}+\lambda_{B})\Delta,\\
\lambda_{E}\Delta &\leq \frac{g_E(z+1;F_X^*)}{g_E(z;F_X^*)} \leq (\alpha_E\mathcal{A}+\lambda_{E})\Delta,\label{eq-gE}
\end{align}
Using the bounds in \eqref{eq-gB}--\eqref{eq-gE}, one obtains 
\begin{align}
\frac{dc_S(x;F_X^*)}{dx} \geq& \,(\alpha_B - \alpha_E)\log[(\alpha_Bx + \lambda_{B})\Delta] + \alpha_E\log\frac{\alpha_Bx+\lambda_{B}}{\alpha_Ex+\lambda_{E}} -\alpha_B\log[(\alpha_B\mathcal{A}+\lambda_B)\Delta]\notag\\
&+ \alpha_E\log(\lambda_{E}\Delta)\notag\\
 =&\, (\alpha_B-\alpha_E)\log\frac{\alpha_B x +\lambda_B}{\alpha_B \mathcal{A} +\lambda_B}+ \alpha_E\log\frac{\alpha_Bx+\lambda_{B}}{\alpha_Ex+\lambda_{E}} + \alpha_E \log\frac{\lambda_E}{\alpha_B \mathcal{A} +\lambda_B},\forall\, x>-\frac{\lambda_B}{\alpha_B}.\label{eq-final}
\end{align} 
Finally, we consider two cases and for each case we provide a contradiction argument. 
\begin{itemize}
	\item Case 1: $\alpha_B > \alpha_E$\\
	In this case, we note that for sufficiently large values of $x$, the right-hand-side (RHS) of \eqref{eq-final} scales logarithmically in $x$, i.e., $\frac{dc_S(x;F_X^*)}{dx} = \Omega(\log x)$ which means that there exist constants $c>0$ and $x_0>-\frac{\lambda_B}{\alpha_B}$ such that $\frac{dc_S(x;F_X^*)}{dx} \geq c\,\log x$ for all $x > x_0$. However, this results in a contradiction since based on \eqref{eq-deriv}, $\frac{dc_S(x;F_X^*)}{dx}$ must be a constant function in $x$ for all $x>-\frac{\lambda_B}{\alpha_B}$.
	\item Case 2: $\alpha_B = \alpha_E$\\
	For this case, using the bounds in \eqref{eq-gB}--\eqref{eq-gE}, we first upper bound $\frac{dc_S(x;F_X^*)}{dx}$ as follows
		\begin{align}
		\frac{dc_S(x;F_X^*)}{dx} \leq&\, 
		(\alpha_B-\alpha_E)\log\frac{\alpha_B x +\lambda_B}{\alpha_E \mathcal{A} +\lambda_E}+ \alpha_E\log\frac{\alpha_Bx+\lambda_{B}}{\alpha_Ex+\lambda_{E}} + \alpha_B \log\frac{\alpha_E \mathcal{A} +\lambda_E}{\lambda_B}\notag\\		
		 =& \,\alpha_B\log\frac{x+\frac{\lambda_{B}}{\alpha_B}}{x+\frac{\lambda_{E}}{\alpha_E}} + \alpha_B \log\frac{\alpha_E\mathcal{A}+\lambda_{E}}{\lambda_{B}}, \quad\forall~ x>-\frac{\lambda_B}{\alpha_B}.\label{eq-final1}
		\end{align}
	Recall that at least one of the inequalities in \eqref{eq-deg-1}--\eqref{eq-deg-2} is strict (due to the degradedness assumption). Therefore, in this case, \eqref{eq-deg-2} is strict. Now, to reach a contradiction, it suffices to compute the limit of the RHS of \eqref{eq-final1} as $x\rightarrow -\frac{\lambda_{B}}{\alpha_B}^{+}$. For this purpose and in regard of \eqref{eq-deriv}, we have
	\begin{equation}
	\gamma \leq \lim_{x\rightarrow -\frac{\lambda_{B}}{\alpha_B}^{+}} \alpha_B\log\frac{x+\frac{\lambda_{B}}{\alpha_B}}{x+\frac{\lambda_{E}}{\alpha_E}} + \alpha_B\log\frac{\alpha_E\mathcal{A}+\lambda_{E}}{\lambda_{B}}.
	\end{equation}
Observe that since $\frac{\lambda_E}{\alpha_E} > \frac{\lambda_{B}}{\alpha_B}$, the limit $\underset{x\rightarrow -\frac{\lambda_{B}}{\alpha_B}^{+}}{\lim} \log\frac{x+\frac{\lambda_{B}}{\alpha_B}}{x+\frac{\lambda_{E}}{\alpha_E}}= -\infty$ and therefore, we get $\gamma \leq -\infty$ which is a contradiction because $\gamma$ is a nonnegative constant. 
\end{itemize}
Hence, for each case we reach a contradiction which implies that the support set $\mathcal{S}_{F_X^*}$ must have finitely many mass points in the interval $[0,\mathcal{A}]$. This completes the proof of Theorem~\ref{theo-1}. 


\subsection{Proof of Theorem~\ref{theo-2}}
This section presents the proof of Theorem~\ref{theo-2} by extending the analysis in the previous section to the case where only an average-intensity constraint is active. We start the proof by noting that the feasible set $\Omega^{+}_{\mathcal{E}}$ is convex and sequentially compact in the L\'evy metric sense~\cite[Appendix~I.A]{923716}. Furthermore, the functional $f_0:\Omega^{+}_{\mathcal{E}}\rightarrow \mathbb{R}, f_0(F_X) = I_B(F_X)-I_E(F_X)$ is continuous in $F_X$. This is because each one of the mutual information terms $I_B(F_X)$ and $I_E(F_X)$ are continuous in $F_X$ based on \cite[Lemma~17]{8632953}. Therefore, we conclude that the supremum in \eqref{eq-SecCap} for $\mathcal{F}^{+} = \Omega^{+}_{\mathcal{E}}$ is achieved by at least one element $F_X\in\Omega^{+}_{\mathcal{E}}$. Furthermore, the functional $f_0(F_X)$ is strictly concave, and weakly differentiable by following along similar lines of Lemma~\ref{lemm-conc} and Lemma~\ref{lem-diff}. Hence, the maximum is achieved by a unique distribution. Finally, invoking similar arguments that appear in the statement of Theorem~\ref{theo-5}, we find the following necessary and sufficient KKT conditions for the optimality of the input distribution $F_X^*$ as
\begin{align}
c_S(x;F_X^*) - \gamma x &\leq C_{S} - \gamma\mathcal{E}, \quad\forall~ x\in[0,+\infty) \label{eq-KKT-1-avg},\\
c_S(x;F_X^*) -\gamma x&= C_{S}- \gamma\mathcal{E}, \quad\forall~ x\in \mathcal{S}_{F_X^*}\label{eq-KKT-2-avg}.
\end{align}
Next, we prove that the secrecy-capacity-achieving input distribution $F_X^*$ has the following structural properties: 1) the intersection of $\mathcal{S}_{F_X^*}$ with any bounded interval $B$ contains a finite number of mass points, i.e., $\lvert \mathcal{S}_{F_X^*}\cap B\rvert < \infty$; 2) the support set of the optimal distribution is an unbounded set. These two properties imply that $\mathcal{S}_{F_X^*}$ is a countably infinite set. The first property is shown by means of contradiction. We assume, on the contrary, that for some bounded interval $B$, $\mathcal{S}_{F_X^*}\cap B$ contains an infinite number of elements. Then, using the KKT conditions in~\eqref{eq-KKT-1-avg}--\eqref{eq-KKT-2-avg}, the analyticity of the secrecy rate density $c_S(x;F_X^*)$ over $\mathcal{O}$, and invoking the Bolzano-Weierstrass and Identity Theorems, we find that $\gamma \leq -\infty$ which is not possible, and hence results in a contradiction. The second property is also shown through a contradiction approach. We consider two cases for the channel gains $\alpha_B$ and $\alpha_E$ and for each case, we provide a contradiction arguments. These cases are as follows: 1) when $\alpha_B > \alpha_E$, our contradiction hinges on the fact that if $\mathcal{S}_{F_X^*}$ is a bounded set, then the cost function which grows linearly in $x$ must be lower bounded by the secrecy rate density which grows as fast as $x\log x$. This is not possible for large values of $x$ and hence a contradiction occurs; 2) when the channel gains are identical, we find that the Lagrangian multiplier must be lower bounded by a constant and thus, using the Envelope Theorem~\cite{524037} we observe that the secrecy capacity must at least grow linearly in the average-intensity constraint. However, in Appendix~\ref{App-E} we establish that the secrecy capacity is always upper bounded by a constant for all values of the average-intensity. Therefore, the desired contradiction is reached and the result follows.

\subsubsection{The support set of the optimal solution has finitely many mass points in any bounded interval}
Let $B$ be a bounded interval and assume, to the contrary, that $\mathcal{S}_{F_X^*} \cap B$ has an infinite number of elements. Now based on the optimality equation \eqref{eq-KKT-2-avg}, the analyticity of $c_S(x;F_X^*)$ over $\mathcal{O}$, and the Bolzano-Weierstrass and Identity Theorems from complex analysis, one can find
\begin{equation}\label{eq-cntrdct-avg}
c_S(x;F_X^*) - \gamma x = C_S - \gamma\mathcal{E},\quad \forall\,x>-\frac{\lambda_B}{\alpha_B}.
\end{equation}
Next, we show that this results in a contradiction. To this end, we note that 
\begin{align}
g_E(z+1;F_X^*)&= \int_{0}^{+\infty}e^{-\alpha_E x\Delta}\,[(\alpha_E x+\lambda_E)\Delta]^{\,z+1}\, dF_X(x) = e^{\lambda_{E}\Delta}(z+1)!\underbrace{\int_{0}^{+\infty}p(z\vert x)\,dF_X^*(x)}_{\leq\, 1~\text{as}~p(z\vert x)\leq 1} \notag\\
&\leq e^{\lambda_{E}\Delta}(z+1)!.\label{eq-up-bound-gE}
\end{align}
Furthermore, observe that
\begin{equation}\label{eq-low-bound-gE}
g_E(z;F_X^*) \geq (\lambda_{E}\Delta)^{z}\,\mathbb{E}_{F_X^*}[e^{-\alpha_E X \Delta}]\stackrel{(i)}{\geq} (\lambda_{E}\Delta)^{z}\,e^{-\alpha_E\Delta\,\mathbb{E}_{F_X^*}[X]} = (\lambda_{E}\Delta)^{z}\,e^{-\alpha_E\mathcal{E}\Delta},
\end{equation} 
where $(i)$ is due to the Jensen's Inequality as $e^{-\alpha_B x \Delta}$ is a convex function in $x$. Plugging the bounds in \eqref{eq-up-bound-gE}--\eqref{eq-low-bound-gE} into \eqref{eq-deriv-main}, we get
\begin{align}
\frac{dc_S(x;F_X^*)}{dx} \leq&\, (\alpha_B - \alpha_E)\log[(\alpha_Bx + \lambda_{B})\Delta] + \alpha_E\log\frac{\alpha_Bx+\lambda_{B}}{\alpha_Ex+\lambda_{E}}\notag\\
&+\alpha_B\underbrace{\sum_{y=0}^{+\infty}p(y\vert x)\log\frac{g_B(y;F_X^*)}{g_B(y+1;F_X^*)}}_{\stackrel{\triangle}{=}\,\Xi_B(x)} + \alpha_E\underbrace{\sum_{z=0}^{+\infty}p(z\vert x)\log\frac{e^{\lambda_E\Delta}(z+1)!}{e^{-\alpha_E\mathcal{E}\Delta}(\lambda_E\Delta)^z}}_{\stackrel{\triangle}{=}\,\Xi_E(x)}.\label{eq-b-1}
\end{align}
Next, we provide upper bounds on $\Xi_B(x)$ and $\Xi_E(x)$ as follows
\begin{align}
\Xi_B(x) &\stackrel{(ii)}{\leq} \sum_{y=0}^{+\infty} p(y\vert x) \log\frac{1}{\lambda_{B}\Delta} = -\log(\lambda_B\Delta)\\
\Xi_E(x) &= \mathbb{E}_{Z\vert X}[\log(Z+1)!- Z\,\log(\lambda_E\Delta)] + (\alpha_E\mathcal{E}+ \lambda_E)\Delta \notag\\
&= \mathbb{E}_{Z\vert X}[\log\,Z!] + \mathbb{E}_{Z\vert X}[\log(Z+1)] - [(\alpha_Ex+\lambda_E)\Delta]\log(\lambda_E\Delta) + (\alpha_E\mathcal{E}+ \lambda_E)\Delta\notag\\
&\stackrel{(iii)}{\leq} \mathbb{E}_{Z\vert X}[\log\,Z!] + \log(\mathbb{E}_{Z\vert X}[Z]+1)- [(\alpha_Ex+\lambda_E)\Delta]\log(\lambda_E\Delta) + (\alpha_E\mathcal{E}+ \lambda_E)\Delta\notag\\
&\stackrel{(iv)}{\leq} \frac{1}{2}\log[2\pi e (\mathbb{E}_{Z\vert X}[Z]+\frac{1}{12})] - \mathbb{E}_{Z\vert X}[Z] + \mathbb{E}_{Z\vert X}[Z] \log(\mathbb{E}_{Z\vert X}[Z])+ \log(\mathbb{E}_{Z\vert X}[Z]+1)\notag\\
&\quad~ - [(\alpha_Ex+\lambda_E)\Delta]\log(\lambda_E\Delta) + (\alpha_E\mathcal{E}+ \lambda_E)\Delta\notag\\
&\leq [(\alpha_Ex+\lambda_E)\Delta]\log[(\alpha_Ex+\lambda_E)\Delta] - [(\alpha_Ex+\lambda_E)\Delta](1+\log(\lambda_E\Delta)) \notag\\ 
&\quad~ + \frac{3}{2}\log[(\alpha_Ex+\lambda_E)\Delta+1] + (\alpha_E\mathcal{E}+\lambda_E)\Delta + \frac{1}{2}\log(2\pi e),\label{eq-b-2}
\end{align}
where $(ii)$ follows from \eqref{eq-gB}, $(iii)$ is due to the Jensen's Inequality as $\log x$ is a concave function, and $(iv)$ follows from an upper bound on the entropy of the Poisson random variable~\cite[Lemma~10]{4729780}. Combining \eqref{eq-b-1}--\eqref{eq-b-2}, we get
\begin{align}
\frac{dc_S(x;F_X^*)}{dx} \leq&\, (\alpha_B-\alpha_E)\log[(\alpha_Bx + \lambda_{B})\Delta] + \alpha_E\log\frac{x+\frac{\lambda_{B}}{\alpha_B}}{x+\frac{\lambda_{E}}{\alpha_E}} +  \alpha_E\log\frac{\alpha_B}{\alpha_E}
\notag\\
&+\alpha_E\big([(\alpha_Ex+\lambda_E)\Delta]\log[(\alpha_Ex+\lambda_E)\Delta] - [(\alpha_Ex+\lambda_E)\Delta](1+\log(\lambda_E\Delta)) \notag\\ 
& + \frac{3}{2}\log[(\alpha_Ex+\lambda_E)\Delta+1] + (\alpha_E\mathcal{E}+\lambda_E)\Delta + \frac{1}{2}\log(2\pi e)\big)\notag\\&-\alpha_B\log(\lambda_B\Delta),\quad\forall\,x>-\frac{\lambda_B}{\alpha_B}.\label{eq-b-3}
\end{align}
In order to see a contradiction it suffices to compute the limit of the RHS of \eqref{eq-b-3} as $x\rightarrow -\frac{\lambda_{B}}{\alpha_B}^{+}$. For this purpose and in regard of \eqref{eq-deriv} and \eqref{eq-b-3}, we have
\begin{align}
\gamma \leq& \lim_{x\rightarrow -\frac{\lambda_{B}}{\alpha_B}^{+}} (\alpha_B-\alpha_E)\log\frac{\alpha_Bx + \lambda_{B}}{\lambda_{B}} + \lim_{x\rightarrow -\frac{\lambda_{B}}{\alpha_B}^{+}}\alpha_E\log\frac{x+\frac{\lambda_{B}}{\alpha_B}}{x+\frac{\lambda_{E}}{\alpha_E}} + \alpha_E\log\frac{\alpha_B}{\alpha_E}  \notag\\
&~ +\alpha_E \underbrace{\lim_{x\rightarrow -\frac{\lambda_{B}}{\alpha_B}^{+}}\left[[(\alpha_Ex+\lambda_E)\Delta]\log[(\alpha_Ex+\lambda_E)\Delta] - [(\alpha_Ex+\lambda_E)\Delta](1+\log(\lambda_E\Delta))\right]}_{\text{finite value for}\, \frac{\lambda_{E}}{\alpha_E} \geq \frac{\lambda_{B}}{\alpha_B}}\notag\\
&~ +\alpha_E \underbrace{\lim_{x\rightarrow -\frac{\lambda_{B}}{\alpha_B}^{+}}\left[\frac{3}{2}\log[(\alpha_Ex+\lambda_E)\Delta+1] + (\alpha_E\mathcal{E}+\lambda_E)\Delta + \frac{1}{2}\log(2\pi e)\right]}_{\text{finite value for}\, \frac{\lambda_{E}}{\alpha_E} \geq \frac{\lambda_{B}}{\alpha_B}}\notag\\
&~ - \alpha_E\log(\lambda_B\Delta).\label{eq-avg-final}
\end{align}
Thus, we obtain that $\gamma \leq -\infty$ which is a contradiction as $\gamma$ is a nonnegative constant. Therefore, the $\mathcal{S}_{F_X^*}\cap B$ has a finite cardinality. This implies that the optimal input distribution $F_X^*$ possess a countably finite number of mass points in any bounded interval. 


\subsubsection{The support set of the optimal distribution $\mathcal{S}_{F_X^*}$ is unbounded}
To prove this, we again resort to a contradiction approach. Assume, to the contrary, that $\mathcal{S}_{F_X^*}$ is a bounded set, i.e., $\mathcal{S}_{F_X^*}\subseteq [0,h]$, where $h$ is some finite positive constant. In the previous section, we proved that the intersection of $\mathcal{S}_{F_X^*}$ with any bounded interval has a finite cardinality. Since, we are assuming that $\mathcal{S}_{F_X^*}$ is bounded, thus, it has a finite cardinality. This implies that $F_X^*(x) = \sum_{i=1}^{N} p_i u(x-x_i)$, where $N < +\infty$, $0\leq x_1<x_2<\cdots <x_N \leq h$ are the mass points with corresponding probabilities $\{p_1,\ldots,p_N\}$. Furthermore, we can write
\begin{align}
g_E(z;F_X^*) &= \int_{0}^{h}e^{-\alpha_E x\Delta}\,[(\alpha_E x+\lambda_E)\Delta]^{\,z}\,dF_X^*(x)\notag\\
&= \sum_{i=1}^{N} p_i e^{-\alpha_E x_i\Delta}\,[(\alpha_E x_i+\lambda_E)\Delta]^{\,z}\notag\\
&>  p_N e^{-\alpha_E x_N\Delta}\,[(\alpha_E x_N+\lambda_E)\Delta]^{\,z}.\\
g_B(y;F_X^*) &= \int_{0}^{h}e^{-\alpha_B x\Delta}\,[(\alpha_B x+\lambda_B)\Delta]^{\,y}\,dF_X^*(x)\notag\\
&= \sum_{i=1}^{N} p_i e^{-\alpha_B x_i\Delta}\,[(\alpha_B x_i+\lambda_B)\Delta]^{\,y}\notag\\
&\leq [(\alpha_B x_N+\lambda_B)\Delta]^{\,y}
\end{align}
Therefore, $\log g_E(z;F_X^*) > \log p_N - \alpha_E x_N \Delta + z\,\log[(\alpha_E x_N + \lambda_E)\Delta]$ and $\log g_B(y;F_X^*) \leq y\,\log[(\alpha_B x_N + \lambda_B)\Delta]$. In light of the optimality equation \eqref{eq-KKT-1-avg} and using these bounds we obtain
\begin{align}
C_S + \gamma(x-\mathcal{E}) \geq&\, (\alpha_B x + \lambda_B)\log[(\alpha_B x + \lambda_B)\Delta] - (\alpha_E x + \lambda_E)\log[(\alpha_E x + \lambda_E)\Delta] \notag\\ &+ (\alpha_E-\alpha_B)x +\frac{1}{\Delta} \sum_{z=0}^{+\infty}p(z\vert x)\log g_E(z;F_X^*) - \frac{1}{\Delta}\sum_{y=0}^{+\infty}p(y\vert x)\log g_B(y;F_X^*) \notag\\
>&\,(\alpha_B x + \lambda_B)\log[(\alpha_B x + \lambda_B)\Delta] - (\alpha_E x + \lambda_E)\log[(\alpha_E x + \lambda_E)\Delta]\notag\\ &+ (\alpha_E-\alpha_B)x + \frac{\log p_N}{\Delta} - \alpha_E x_N + (\alpha_E x + \lambda_E)\log[(\alpha_E x_N + \lambda_E)\Delta] \notag\\
&-(\alpha_B x + \lambda_B)\log[(\alpha_B x_N + \lambda_B)\Delta]   \notag\\
=&\,(\alpha_B-\alpha_E)x\log[(\alpha_B x +\lambda_B)\Delta] + \alpha_E x \log\frac{\alpha_B x +\lambda_B}{\alpha_E x + \lambda_E} \notag\\
&\,+ x\left[(\alpha_E-\alpha_B) + (\alpha_E-\alpha_B)\log[(\alpha_E x_N+\lambda_E )\Delta] + \alpha_B\log\frac{\alpha_E x_N + \lambda_E}{\alpha_B x_N + \lambda_B}\right]\notag\\
&\,+\lambda_B\log\frac{\alpha_B x + \lambda_B}{\alpha_B x_N + \lambda_B}-\lambda_E\log\frac{\alpha_E x + \lambda_E}{\alpha_E x_N + \lambda_E} + \frac{\log p_N}{\Delta} - \alpha_E x_N,\quad\forall\, x\geq 0.\label{eq-avg-1}
\end{align}
Now, we consider the following cases and for each case we provide a contradiction argument.
\begin{itemize}
	\item Case 1: $\alpha_B > \alpha_E$\\
	Observe that in this case, the RHS of \eqref{eq-avg-1} scales like $x\log x$ for sufficiently large values of $x$, i.e., $C_S + \gamma(x-\mathcal{E}) = \Omega(x\log x)$. However, this is clearly a contradiction because $C_S + \gamma(x-\mathcal{E})$ grows linearly in $x$. Thus, the optimal support set $\mathcal{S}_{F_X^*}$ must be an unbounded set. 
	\item Case 2: $\alpha_B = \alpha_E$\\
	In this case, \eqref{eq-avg-1} can be simplified further as
	\begin{align}
	C_S + \gamma(x-\mathcal{E}) \geq&\, \alpha_B x \left[\log\frac{x_N +\frac{ \lambda_E}{\alpha_E}}{x_N + \frac{\lambda_B}{\alpha_B}} + \log\frac{x +\frac{ \lambda_B}{\alpha_B}}{x + \frac{\lambda_E}{\alpha_E}}\right] + \lambda_B\log\frac{x + \frac{\lambda_B}{\alpha_B}}{x_N + \frac{\lambda_B}{\alpha_B}} \notag\\&- \lambda_E\log\frac{x + \frac{\lambda_E}{\alpha_E}}{x_N +\frac{\lambda_E}{\alpha_E}}+\frac{\log p_N}{\Delta} -\alpha_E x_N,~ \forall \, x\geq 0.\label{eq-avg-2}
	\end{align}
Observe that the RHS of \eqref{eq-avg-2} grows linearly in $x$ for large values of $x$. Thus, dividing the sides of \eqref{eq-avg-2} by $x>0$ and taking the limit as $x\rightarrow \infty$, we find
\begin{equation}
\gamma \geq \alpha_B\log\frac{x_N + \frac{\lambda_E}{\alpha_E}}{x_N + \frac{\lambda_B}{\alpha_B}}.\label{eq-avg-lb}
\end{equation}
We note that since $\alpha_B = \alpha_E$, the inequality in \eqref{eq-deg-2} is strict, i.e., $\frac{\lambda_E}{\alpha_E} > \frac{\lambda_B}{\alpha_B}$ and therefore, $\alpha_B\log\frac{x_N + \frac{\lambda_E}{\alpha_E}}{x_N + \frac{\lambda_B}{\alpha_B}} > 0$. Next, we show that this lower bound on the Lagrangian multiplier $\gamma$ results in a contradiction. To that end, we first note that the Lagrangian multiplier $\gamma$ and the location of the last mass point in the support set of the optimal distribution depend on the value of the average-intensity constraint. Thus, in \eqref{eq-avg-lb} one must replace $\gamma$ by $\gamma(\mathcal{E})$ and $x_N$ by $x_N(\mathcal{E})$. Now, we recall the Envelope Theorem~\cite{524037} which shows that the Lagrangian multiplier $\gamma$ and the secrecy capacity (the optimal value of the objective functional) are related as follows
\begin{equation}
\frac{dC_S(\mathcal{E})}{d\mathcal{E}} = \gamma(\mathcal{E}),~\forall\, \mathcal{E} > 0.
\end{equation}
In light of this relationship and the lower bound in \eqref{eq-avg-lb}, the following lower bound can be found
\begin{equation}
C_S(\mathcal{E}) = \int_{0}^{\mathcal{E}}\gamma(t)\,dt\geq \int_{0}^{\mathcal{E}} \alpha_B\log\frac{x_N(t) + \frac{\lambda_E}{\alpha_E}}{x_N(t) + \frac{\lambda_B}{\alpha_B}}\,dt = \int_{0}^{\mathcal{E}} \alpha_B\log\left[1 + \frac{\frac{\lambda_E}{\alpha_E} - \frac{\lambda_B}{\alpha_B}}{ x_N(t) + \frac{\lambda_B}{\alpha_B}}\right]\,dt.\label{eq-avg-lb1}
\end{equation}
Now, based on the contradiction assumption we have $x_N(t) < h$ with $h$ being a finite positive constant. Therefore, \eqref{eq-avg-lb1} can be further lower bounded as
\begin{equation}
C_S(\mathcal{E}) \geq \int_{0}^{\mathcal{E}}\alpha_B\log\left[1 + \frac{\frac{\lambda_E}{\alpha_E} - \frac{\lambda_B}{\alpha_B}}{h + \frac{\lambda_B}{\alpha_B}}\right]\,dt = \alpha_B\log\left[1 + \frac{\frac{\lambda_E}{\alpha_E} - \frac{\lambda_B}{\alpha_B}}{h + \frac{\lambda_B}{\alpha_B}}\right]\mathcal{E},\label{eq-avg-cntrdct-final}
\end{equation}
which must hold for all $\mathcal{E}> 0$. Since $h>0$ and $\frac{\lambda_E}{\alpha_E} > \frac{\lambda_B}{\alpha_B}$, the logarithm term is always positive implying that $C_S(\mathcal{E})$ must at least grow linearly in $\mathcal{E}$ for all $\mathcal{E} > 0$. However, in Appendix~\ref{App-E}, we establish that the secrecy capacity of the DT--PWC with nonnegativity and average-intensity constraints when $\alpha_B = \alpha_E$ is upper bounded by a constant for all $\mathcal{E} > 0$. Therefore, the implication in \eqref{eq-avg-cntrdct-final} results in a contradiction. This implies that the optimal support set $\mathcal{S}_{F_X^*}$ must be an unbounded set.
\end{itemize}
Showing that $\mathcal{S}_{F_X^*}$ is an unbounded set for these considered cases completes the proof of Theorem~\ref{theo-2}.

\subsection{Proof of Theorem~\ref{theo-3}}
We start the proof by noting that the feasible set $\Omega^{+}_{\mathcal{A},\,\mathcal{E}}$ is compact and convex, and the objective function $f_{\mu}(F_X)$ in \eqref{eq-equivocregion} is continuous in $F_X$, strictly concave, and  weakly differentiable. Therefore, the optimization problem in \eqref{eq-equivocregion} has a \textit{unique} maximizer. We denote the optimal input distribution for \eqref{eq-equivocregion} by $F_X^*$ which depends on the value $\mu$. 

Next, we obtain the KKT conditions for the optimal input distribution of the optimization problem in \eqref{eq-equivocregion}. Following along similar lines of the proof of Theorem~\ref{theo-1} and noting that the objective function $f_{\mu}(F_X)$ is weakly differentiable with a weak derivative given as
\begin{equation}
f_{\mu}^{\prime}(F_X,F_X^*)= \int_{0}^{\mathcal{A}}\left[\mu\,i_{B}(x;F_X^*)+ (1-\mu)\,c_{S}(x;F_X^*)\right]\,dF_X(x) - f_{\mu}(F_X^*),
\end{equation}
the KKT conditions for the optimality of $F_X^*$ are obtained as follows
\begin{align}
\mu i_{B}(x;F_X^*) + (1-\mu)\, c_{S}(x;F_X^*)-\gamma x&\leq \mu I_{B}(F_X^*) +(1-\mu)
\left[I_{B}(F_X^*)-I_{E}(F_X^*)\right]
-\gamma\mathcal{E},\notag\\& \qquad\qquad\qquad\qquad\qquad\forall~ x\in[0,\mathcal{A}], \label{eq-rateequivKKT-Bob}\\ 
\mu i_{B}(x;F_X^*) + (1-\mu)\, c_{S}(x;F_X^*) -\gamma x&= \mu I_{B}(F_X^*)\ + (1-\mu)
\left[I_{B}(F_X^*)-I_{E}(F_X^*)\right]-\gamma\mathcal{E},\notag\\
&\qquad\qquad\qquad\qquad\qquad\forall ~x\in \mathcal{S}_{F_X^*}.
\label{eq-rateequivKKT-Eve}
\end{align}
Next, we show that the optimal input distribution $F_X^*$ has a finite support. To this end, assume to the contrary, that $\mathcal{S}_{F_X^*}$ has an infinite number of elements. Under such an assumption, \eqref{eq-rateequivKKT-Eve}, the analyticity of $i_{B}(w;F_X^*)$ and $i_{E}(w;F_X^*)$ over $\mathcal{O}$ in the complex plane and the Bolzano-Weierstrass and Identity Theorems of complex analysis, one obtains 
\begin{align}
\mu i_{B}(x;F_X^*) + (1-\mu)\, c_{S}(x;F_X^*) -\gamma x &= \mu I_{B}(F_X^*) + (1-\mu)
\left[I_{B}(F_X^*)-I_{E}(F_X^*)\right]-\gamma\mathcal{E},\notag\\&\qquad\qquad\qquad\qquad\forall~x>-\frac{\lambda_B}{\alpha_B}.
\label{eq-rateequivKKT-RealLine}
\end{align} 
We continue the proof by showing that \eqref{eq-rateequivKKT-RealLine} results in a contradiction. To do so, we first observe that RHS of \eqref{eq-rateequivKKT-RealLine} does not depend on $x$ and hence, it is a constant function in $x$. Taking the derivative of both sides of \eqref{eq-rateequivKKT-RealLine} with respect to $x$, we get
\begin{equation}
\mu\frac{di_B(x;F_X^*)}{dx} + (1-\mu) \frac{dc_S(x;F_X^*)}{dx} = \gamma,\quad\forall~x>-\frac{\lambda_B}{\alpha_B}, 
\label{eq-equiv-cntrdct}
\end{equation}
or equivalently
\begin{align}
\gamma =&\, \mu\Big[\alpha_B\log[(\alpha_Bx + \lambda_{B})\Delta] +\alpha_B \sum_{y=0}^{+\infty}p(y\vert x)\log\frac{g_B(y;F_X^*)}{g_B(y+1;F_X^*)}\Big] + (1-\mu)\Big[(\alpha_B - \alpha_E)\notag\\
&\quad~\times\log[(\alpha_Bx + \lambda_{B})\Delta] + \alpha_E\log\frac{\alpha_Bx+\lambda_{B}}{\alpha_Ex+\lambda_{E}}+\alpha_B \sum_{y=0}^{+\infty}p(y\vert x)\log\frac{g_B(y;F_X^*)}{g_B(y+1;F_X^*)}\notag\\
&\quad~
- \alpha_E\sum_{z=0}^{+\infty}p(z\vert x)\log\frac{g_E(z;F_X^*)}{g_E(z+1;F_X^*)}\Big],\quad\forall~x>-\frac{\lambda_B}{\alpha_B}.
\end{align}
Using the bounds in \eqref{eq-gB}--\eqref{eq-gE}, the RHS of \eqref{eq-equiv-cntrdct} can be lower bounded as 
\begin{align}
\gamma \geq&\, \mu\,\alpha_B\log\frac{\alpha_B x + \lambda_{B}}{\alpha_B \mathcal{A} + \lambda_{B}} + (1-\mu)\Big[(\alpha_B-\alpha_E)\log\frac{\alpha_B x +\lambda_B}{\alpha_B \mathcal{A} +\lambda_B}+ \alpha_E\log\frac{\alpha_Bx+\lambda_{B}}{\alpha_Ex+\lambda_{E}}\notag\\& + \alpha_E \log\frac{\lambda_E}{\alpha_B \mathcal{A} +\lambda_B}\Big],\quad\forall~x>-\frac{\lambda_B}{\alpha_B}.
\label{eq-equiv-bound}
\end{align}
Observe that the RHS of \eqref{eq-equiv-bound} scales logarithmically, i.e., $\Omega(\log x)$ for large values of $x$. This is clearly a contradiction because the constant value $\gamma$ cannot be greater than a logarithmically increasing function. This implies that $\mathcal{S}_{F_X^*}$ cannot have infinite elements in the interval $[0,\mathcal{A}]$. Hence, $F_X^*$ is discrete with a finite number of mass points. Additionally, we note that for $\mu = 0$, $F_X^*$ must be discrete with a finite support according to Theorem~\ref{theo-1}, and for $\mu = 1$ (the point corresponding to the capacity of the discrete-time Poisson channel with peak- and average-intensity constraints), $F_X^*$ is also discrete with a finite number of mass points; reproving the results presented in \cite{217161}. Consequently, the entire rate-equivocation region of the DT--PWC with peak- and average-intensity constraints is exhausted by discrete input distributions with finitely many mass points. This completes the proof of Theorem~\ref{theo-3}. 


\subsection{Proof of Theorem~\ref{theo-4}}
We start the proof by noting that the feasible set $\Omega^{+}_{\,\mathcal{E}}$ is compact and convex, and the objective function $f_{\mu}(F_X)$ in \eqref{eq-equivocregion} is continuous in $F_X$, strictly concave, and  weakly differentiable. Therefore, the optimization problem in \eqref{eq-equivocregion} has a \textit{unique} maximizer. We denote the optimal input distribution for \eqref{eq-equivocregion} by $F_X^*$ which depends on $\mu$. 

The KKT conditions for the optimal input distribution $F_X^*$ of the optimization problem in \eqref{eq-equivocregion} is given by
\begin{align}
\mu i_{B}(x;F_X^*) + (1-\mu)\, c_{S}(x;F_X^*)-\gamma x&\leq \mu I_{B}(F_X^*) + (1-\mu)
\left[I_{B}(F_X^*)-I_{E}(F_X^*)\right]-\gamma\mathcal{E},\notag\\&\qquad\qquad\qquad\qquad\qquad\forall~ x\in[0,+\infty), \label{eq-rateequivKKT-avg}\\ 
\mu i_{B}(x;F_X^*) + (1-\mu)\, c_{S}(x;F_X^*) -\gamma x&= \mu I_{B}(F_X^*) + (1-\mu)
\left[I_{B}(F_X^*)-I_{E}(F_X^*)\right]-\gamma\mathcal{E},\notag\\&\qquad\qquad\qquad\qquad\qquad\forall ~x\in \mathcal{S}_{F_X^*}.
\label{eq-rateequivKKT-avg1}
\end{align}
We show that the optimal input distribution $F_X^*$ has the following structural properties: 1) the intersection of the optimal support set with any bounded interval contains finitely many mass points; 2) The optimal support set itself is an unbounded set. Theses properties are proved via similar contradiction approaches that appear in the proof of Theorem~\ref{theo-2}. 
\subsubsection{The intersection of the optimal support set with any bounded interval contains a finite number of elements}
Let $B$ be a bounded interval and assume, to the contrary, that $\mathcal{S}_{F_X^*} \cap B$ has an infinite number of elements. Now based on the optimality equation \eqref{eq-rateequivKKT-avg1}, the analyticity of $i_B(x;F_X^*)$ and $c_S(x;F_X^*)$ over $\mathcal{O}$, the Bolzano-Weierstrass and Identity Theorems from complex analysis, we get 
\begin{align}
\mu i_{B}(x;F_X^*) + (1-\mu)\, c_{S}(x;F_X^*) -\gamma x &= \mu I_{B}(F_X^*) 
+ (1-\mu)
\left[I_{B}(F_X^*)-I_{E}(F_X^*)\right]-\gamma\mathcal{E},\notag\\&\qquad\qquad\qquad\qquad\qquad\forall\,x>-\frac{\lambda_B}{\alpha_B},
\label{eq-rateequivKKT-RealLine-avg}
\end{align} 
and we show that \eqref{eq-rateequivKKT-RealLine-avg} results in a contradiction. By taking the derivative of both sides of \eqref{eq-rateequivKKT-RealLine-avg} with respect to $x$ we find
\begin{equation}
\mu\frac{di_B(x;F_X^*)}{dx} + (1-\mu) \frac{dc_S(x;F_X^*)}{dx} = \gamma,\quad\forall\,x>-\frac{\lambda_B}{\alpha_B}.
\label{eq-equiv-cntrdct-avg}
\end{equation}
Using the bounds in \eqref{eq-b-1}--\eqref{eq-b-2} the RHS of \eqref{eq-equiv-cntrdct-avg} can be upper bounded as 
\begin{align}
\gamma \leq&\,  \mu\,\alpha_B\log\frac{\alpha_Bx + \lambda_{B}}{\lambda_{B}}+(1-\mu)\big[(\alpha_B-\alpha_E)\log[(\alpha_Bx + \lambda_{B})\Delta] + \alpha_E\log\frac{\alpha_B}{\alpha_E}\notag\\
&+ \alpha_E\log\frac{x+\frac{\lambda_{B}}{\alpha_B}}{x+\frac{\lambda_{E}}{\alpha_E}}
+\alpha_E\big([(\alpha_Ex+\lambda_E)\Delta]\log[(\alpha_Ex+\lambda_E)\Delta] - [(\alpha_Ex+\lambda_E)\Delta]\notag\\ 
& \times(1+\log(\lambda_E\Delta)) + \frac{3}{2}\log[(\alpha_Ex+\lambda_E)\Delta+1] + (\alpha_E\mathcal{E}+\lambda_E)\Delta + \frac{1}{2}\log(2\pi e)\big)\notag\\&-\alpha_B\log(\lambda_B\Delta)\big],\quad\forall\,x>-\frac{\lambda_B}{\alpha_B}.
\label{eq-equiv-bound-avg}
\end{align}
Taking the limit from both sides of \eqref{eq-equiv-bound-avg} as $x\rightarrow -\frac{\lambda_{B}}{\alpha_B}^{+}$, we obtain $\gamma \leq -\infty$. This is a contradiction and we conclude that $\mathcal{S}_{F_X^*}\cap B$ must contain finitely many mass points. Notice that this holds true for all $\mu\in[0,1]$ implying that the support set of the capacity-achieving input distribution for the discrete-time Poisson channel with nonnegativity and average-intensity constraints has a finite number of mass points in any bounded interval. Notice that the upper bound in \eqref{eq-equiv-bound-avg} depends on $\Delta$ for all $\mu\in[0,1)$, but it does not depend on $\Delta$ for $\mu = 1$. Therefore, in this case we conclude that the \textit{capacity-achieving} distribution of the continuous-time PWC with nonnegativity and average-intensity constraints admits a finite number of mass points in any bounded interval. Nevertheless, the capacity of the continuous-time version under an average-intensity constraint is infinite~\cite{1056262}.

\subsubsection{The support set of the optimal distribution $\mathcal{S}_{F_X^*}$ for all $\mu\in[0,1]$ is unbounded}
Assume, to the contrary, that $\mathcal{S}_{F_X^*}$ is a bounded set, i.e., $\mathcal{S}_{F_X^*}\subseteq [0,h]$ where $h$ is some finite positive constant. In the previous section, we proved that the intersection of $\mathcal{S}_{F_X^*}$ with any bounded interval has a finite number of elements for all $\mu\in[0,1]$. Since we are assuming that $\mathcal{S}_{F_X^*}$ is bounded, thus, it must contain finitely many mass points. This implies that $F_X^*(x) = \sum_{i=1}^{N} p_i u(x-x_i)$, where $N < +\infty$, $0\leq x_1<x_2<\cdots <x_N \leq h$ are the mass points with corresponding probabilities $\{p_1,\ldots,p_N\}$. Following along similar lines of the proof of Theorem~\ref{theo-2} and in view of the optimality condition \eqref{eq-rateequivKKT-avg}, one can write
\begin{align}
\Psi(\mu,\Delta,F_X^*)+\gamma(x-\mathcal{E}) \geq&\, (1-\mu)\, c_{S}(x;F_X^*)+ \mu i_{B}(x;F_X^*)\notag\\>&\, (1-\mu)\Big[(\alpha_B-\alpha_E)x\log[(\alpha_B x +\lambda_B)\Delta] + \alpha_E x \log\frac{\alpha_B x +\lambda_B}{\alpha_E x + \lambda_E} \notag\\
&+ x\big[(\alpha_E-\alpha_B)(1+\log[(\alpha_E x_N+\lambda_E )\Delta]) + \alpha_B\log\frac{\alpha_E x_N + \lambda_E}{\alpha_B x_N + \lambda_B}\big]\notag\\
&+\lambda_B\log\frac{\alpha_B x + \lambda_B}{\alpha_B x_N + \lambda_B}-\lambda_E\log\frac{\alpha_E x + \lambda_E}{\alpha_E x_N + \lambda_E} + \frac{\log p_N}{\Delta} - \alpha_E x_N\Big]\notag\\
&+\mu\big[(\alpha_B x + \lambda_B)\log\frac{\alpha_B x + \lambda_B}{\alpha_B x_N + \lambda_B} -\alpha_B x \big],~\forall\, x\geq 0,\label{eq-rateequiv-avg}
\end{align}
where $\Psi(\mu,\Delta,F_X^*) \stackrel{\triangle}{=} \mu I_B(F_X^*) + (1-\mu)[I_B(F_X^*) - I_E(F_X^*)]$. Observe that the RHS of \eqref{eq-rateequiv-avg} scales like $x\log x$ for large values of $x$ and for all $\mu\in(0,1]$, but the left hand side (LHS) of \eqref{eq-rateequiv-avg} is a linear function in $x$. Therefore, for all $\mu\in(0,1]$ we reach a contradiction and we have that $\mathcal{S}_{F_X^*}$ must be an unbounded set. Furthermore, we have already established in Theorem~\ref{theo-2} that when $\mu = 0$ (the point corresponding to the secrecy capacity) $\mathcal{S}_{F_X^*}$ is also unbounded. Consequently, we conclude that $\mathcal{S}_{F_X^*}$ is an unbounded set for all $\mu\in[0,1]$. This completes the proof of Theorem~\ref{theo-4}. 

\section{Numerical Results}\label{sec-numres}
In this section, we provide numerical results for the secrecy capacity and the entire rate-equivocation region of the DT--PWC.

\begin{figure}[!t]
	\centering
	\includegraphics[width=0.55\textwidth]{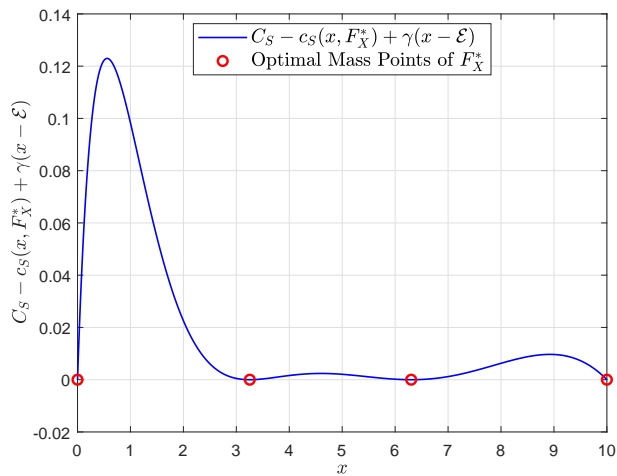}
	\caption{Illustration of $C_S - c_S(x;F_X^*) + \gamma(x-\mathcal{E})$ yielded by the optimal input distribution when $A = 10$, $\mathcal{E}=\frac{\mathcal{A}}{4}$, $\alpha_B = 2$, $\lambda_B=1$, $\alpha_E = 1$, $\lambda_E = 2$, and $\Delta = 0.5$ seconds.}
	\label{fig-equivocdensity}
\end{figure}
Figure~\ref{fig-equivocdensity} provides a plot of the KKT conditions given by \eqref{eq-KKT-1}--\eqref{eq-KKT-2} for an optimal input distribution when $A = 10$, $\mathcal{E}=\frac{\mathcal{A}}{4}$, $\alpha_B = 2$, $\lambda_B=1$, $\alpha_E = 1$, $\lambda_E = 2$, and $\Delta = 0.5$ seconds. We numerically found that for these parameters, the optimal input distribution has four mass points located at $x = 0,\,3.2541,\,6.3032$, and $10$ with probability masses $0.4799,\, 0.3630,\,0.0683$, and $0.0888$, respectively. Furthermore, the corresponding Lagrange multiplier is $\gamma = 0.0513$. We observe that $C_S - c_S(x;F_X^*) + \gamma(x-\mathcal{E})$ is generally nonnegative and is equal to zero at the optimal mass points; verifying the optimality conditions in \eqref{eq-KKT-1}--\eqref{eq-KKT-2}.

\begin{figure}[!t]
	\centering
	\includegraphics[width=0.55\textwidth]{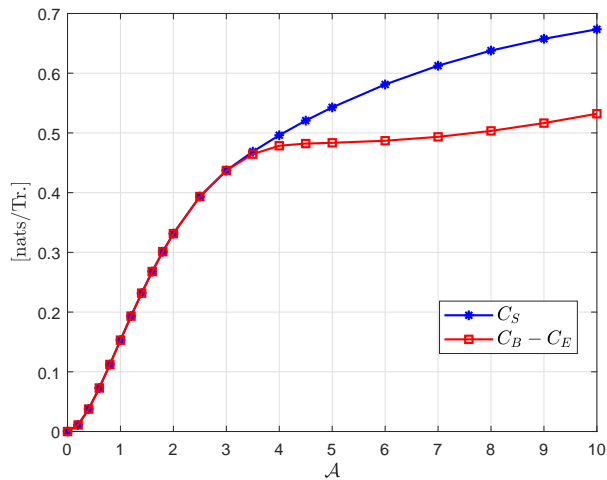}
	\caption{The secrecy capacity when $\mathcal{E}=\frac{\mathcal{A}}{4}$, $\alpha_B = 2$, $\lambda_B=1$, $\alpha_E = 1$, $\lambda_E = 2$, and $\Delta = 0.5$ seconds versus the peak-intensity constraint $\mathcal{A}$.}
	\label{fig-SecCap}
\end{figure}
Figure~\ref{fig-SecCap} illustrates the secrecy capacity $C_S$ and the difference $C_B - C_E$ versus the peak-intensity constraint $\mathcal{A}$, where $C_B$ and $C_E$ are the legitimate user's and the eavesdropper's channel capacities, respectively. First, we observe that the secrecy capacity is an increasing function in $\mathcal{A}$. Furthermore, we see that this difference is a lower bound on the secrecy capacity $C_S$. We also observe that, for small values of $\mathcal{A}$, $C_B - C_E$ and $C_S$ are identical. However, as $\mathcal{A}$ increases $C_B - C_E$ and $C_S$ become different. Similar to the secrecy capacity results of the FSO wiretap channel and optical wiretap channel with input-dependent Gaussian noise under a peak- and average-intensity constraints provided in~\cite{7164335,8399890}, here too, $I(X;Y)$ and $I(X;Z)$ are maximized by the same discrete distribution, however, $I(X;Y) - I(X;Z)$ is maximized by a different distribution. As a specific example, when $\mathcal{A} = 4$, while both $I(X;Y)$ and $I(X;Z)$ are maximized by the same \textit{binary} distribution with mass points at $x = 0$ and $4$ with probability masses $0.75$ and $0.25$, respectively, $I(X;Y) - I(X;Z)$ is maximized by a \textit{ternary} distribution with mass points at $x = 0,\, 2.6848$, and $4$ with probability masses $0.6884,\, 0.1872$, and $0.1244$, respectively. This explains the difference between $C_S$ and $C_B - C_E$ at $\mathcal{A} = 4$ in this figure. 

\begin{figure}[!t]
	\centering
	\includegraphics[width=0.55\textwidth]{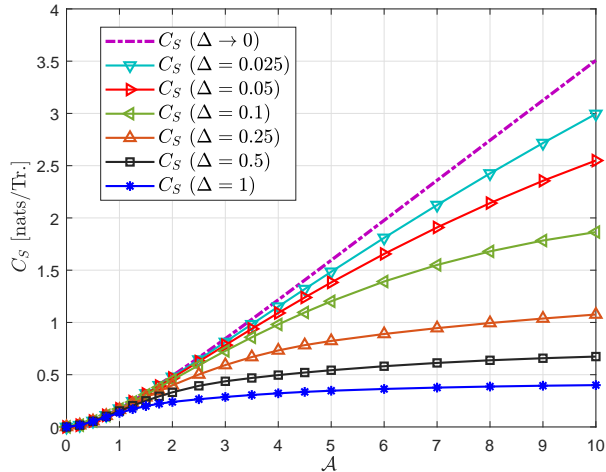}
	\caption{The secrecy capacity of the DT--PWC when $\mathcal{E}=\frac{\mathcal{A}}{4}$, $\alpha_B = 2$, $\lambda_B=1$, $\alpha_E = 1$, and $\lambda_E = 2$ versus the peak-intensity constraint $\mathcal{A}$ for different values of pulse duration $\Delta$.}
	\label{fig-compare}
\end{figure}
In Fig.~\ref{fig-compare}, we plot the effect of pulse duration $\Delta$ on the secrecy capacity of the DT--PWC with nonnegativity, peak- and average-intensity constraints. From the figure, we observe that, in the low-intensity regime, the effect of decreasing $\Delta$ on the secrecy capacity is not significant. However, in the moderate- to high-intensity regime, $\Delta$ becomes significantly influential and the decrease in $\Delta$ results in a higher secrecy capacity. Furthermore, we see that the secrecy capacity of the continuous-time PWC (when $\Delta\rightarrow 0$) is always an upper bound on the secrecy capacity of the DT--PWC.

\begin{figure}[!t]
	\centering
	\includegraphics[width=0.55\textwidth]{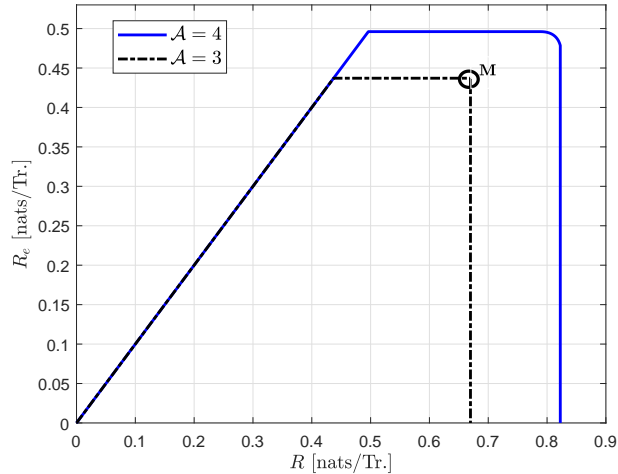}
	\caption{The rate-equivocation region when $\mathcal{E}=\frac{\mathcal{A}}{4}$, $\alpha_B = 2$, $\lambda_B=1$, $\alpha_E = 1$, $\lambda_E = 2$, and $\Delta = 0.5$ for peak-intensity constraints $\mathcal{A}=3$ and $\mathcal{A}=4$. Point $M$ refers to the case when secrecy capacity and capacity are achieved simultaneously.}
	\label{fig-EquiVocRegion}
\end{figure}
Figure~\ref{fig-EquiVocRegion} depicts the entire rate-equivocation region of the DT--PWC with nonnegativity, peak- and average-intensity constraints when $\mathcal{E}=\frac{\mathcal{A}}{4}$, $\alpha_B = 2$, $\lambda_B=1$, $\alpha_E = 1$, $\lambda_E = 2$, and $\Delta = 0.5$ for two different values of $\mathcal{A}$. When $\mathcal{A} = 3$, it is clear from the figure that both the secrecy capacity and the capacity can be attained simultaneously (Point ``M'' in the figure). In particular, for $\mathcal{A} = 3$, the binary input distribution with mass points located at $x = 0$ and $3$ with probabilities $0.75$ and $0.25$, respectively, achieves both the capacity and the secrecy capacity. This implies that, when $\mathcal{A} = 3$, the transmitter can communicate with the legitimate user at the capacity while achieving the maximum equivocation at the eavesdropper. On the other hand, when $\mathcal{A} = 4$ the secrecy capacity and the capacity cannot be achieved simultaneously (notice the curved shape in the figure). More specifically, for $\mathcal{A} = 4$ the binary input distribution with mass points at $x = 0$ and $4$ with probability masses $0.75$ and $0.25$, respectively, achieves the capacity, while a ternary distribution with mass points located at $x = 0,\, 2.6848$, and $4$ with probability masses $0.6884,\, 0.1872$, and $0.1244$, respectively, achieves the secrecy capacity. This implies that the optimal input distributions for the secrecy capacity and the capacity are different. In other words, there is a tradeoff between the rate and its equivocation in the sense that, to increase the communication rate, one must compromise on the equivocation of this communication, and to increase the achieved equivocation, one must compromise on the communication rate.

\begin{figure}[!t]
	\centering
	\includegraphics[width=0.55\textwidth]{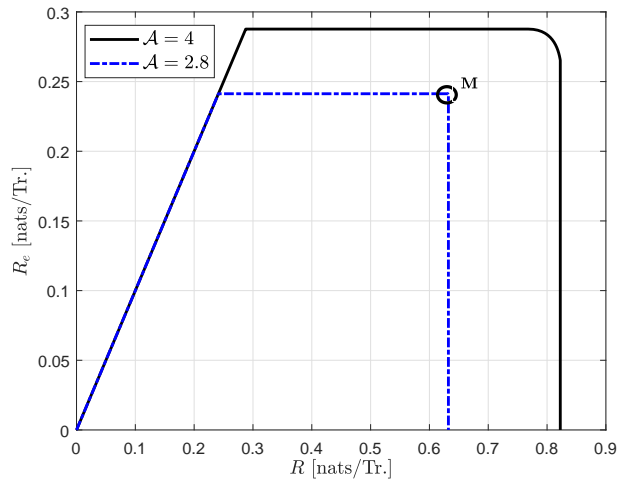}
	\caption{The rate-equivocation region when $\mathcal{E}=\frac{\mathcal{A}}{4}$, $\alpha_B = 2$, $\lambda_B=1$, $\alpha_E = 1$, $\lambda_E = 0.5$, and $\Delta = 0.5$ for peak-intensity constraints $\mathcal{A}=2.8$ and $\mathcal{A}=4$. Point $M$ refers to the case when secrecy capacity and capacity are achieved simultaneously.}
	\label{fig-EquiVocRegion-Thinned}
\end{figure}
Figure~\ref{fig-EquiVocRegion-Thinned} illustrates the entire rate-equivocation region of the DT--PWC with nonnegativity, peak- and average-intensity constraints for the case when $\alpha_B > \alpha_E$ and $\frac{\lambda_{B}}{\alpha_B} = \frac{\lambda_{E}}{\alpha_E}$. In this case, the eavesdropper's observations are just the thinned version of those of the legitimate receiver's and \cite{6294444} shows that for the continuous-time PWC, $C_S = C_B - C_E$, i.e., there is no tradeoff between the rate and its equivocation. This is in contrast to the case of the DT--PWC as shown in this figure. We observe that even in this extreme case, in general, there is a tradeoff between the rate and its equivocation. 

\begin{figure}[!t]
	\centering
	\includegraphics[width=0.55\textwidth]{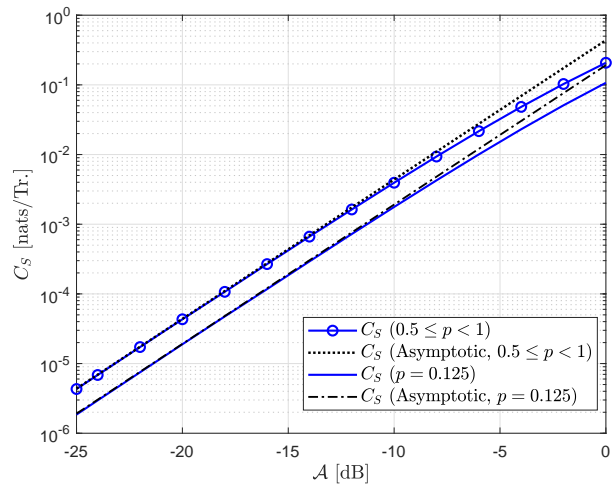}
	\caption{The asymptotic and exact secrecy capacity for $\alpha_B = 2$, $\lambda_{B}=1$, $\alpha_E=1$, $\lambda_E = 2$, and $\Delta = 0.5$ versus $\mathcal{A}$ for both peak- and average-intensity constraints.}
	\label{fig-asymptot}
\end{figure}
In Fig.~\ref{fig-asymptot}, we plot the exact and asymptotic secrecy capacity results in the low-intensity regime versus the peak-intensity constraint $\mathcal{A}$ when peak-intensity or both the peak- and average-intensity constraints are active. From the figure, we observe that the asymptotic results for the secrecy capacity given in \eqref{eq-CS-LIR-PA} are in precise agreement with the numerical result.

\section{Conclusions}\label{sec-conc}
We studied the DT--PWC where a combination of peak- and average-intensity constraints were considered. We formally characterized the secrecy-capacity-achieving input distribution to be unique and discrete with a finite number of mass points when peak-intensity or both peak- and average-intensity constraints were active. Also, we established that the entire rate-equivocation region of the DT--PWC under peak-intensity or both peak- and average-intensity constraints is exhausted by discrete distributions with finitely many mass points. However, when only an average-intensity constraint is imposed we showed that the secrecy capacity as well as the entire boundary of the rate-equivocation region are attained by discrete distributions with countably infinite number of mass points, but finitely many mass points in any bounded interval. 

Besides, we characterized the behavior of the secrecy capacity in both the low- and high-intensity regimes. In the low-intensity regime, we fully characterized the secrecy capacity and the secrecy-capacity-achieving input distribution when peak-intensity or both peak- and average-intensity constraints are active. We proved that in this regime the secrecy capacity scales quadratically in the peak-intensity constraint and the optimal input distribution is binary. Also, when both peak- and average-intensity constraints were active and the peak-intensity was held fixed while the average-intensity tended to zero, we established that the secrecy capacity scales linearly in the average-intensity constraint and the optimal input distribution is binary. Moreover, when only the average-intensity constraint was active and the channel gains of the legitimate receiver and the eavesdropper were identical, the secrecy capacity scaled linearly in the average-intensity. Finally, we observed that with only the average-intensity constraint and different channel gains, the secrecy capacity scaled, to within a constant, like $\mathcal{E}\log\log\frac{1}{\mathcal{E}}$. In the high-intensity regime, we established that under either of the peak- or average-intensity constraints, the secrecy capacity must be a constant.

Towards the ending part of our work, we provided numerical experiments. Our numerical results indicated that under both the peak- and average-intensity constraints, the secrecy capacity and the capacity of the DT--PWC channel cannot be obtained simultaneously in general, i.e., there is a tradeoff between the rate and its equivocation. 

\begin{appendices}
	\section{The Strict Concavity of $f(F_X)$ in $F_X$}\label{app-strict}
	We start the proof by noting that for random variables $X$, $Y$ and $Z$ that form the Markov chain $X\rightarrow Y \rightarrow Z$, $I(X;Y\lvert Z) = I(X;Y) - I(X;Z)$ is a concave functional in $F_X$~\cite[Appendix~A]{4529277}. Now, let $X_1$ and $X_2$ be two channel inputs generated by $F_{X_1}$ and $F_{X_2}$, respectively, and $Q$ be a binary-valued random variable such that 
	\begin{equation}\label{eq-Qdefine}
	p(y,z,x\lvert q) = 
	\begin{cases}
	p(y,z\lvert x)\,p_{X_1}(x), \quad q = 1, \\
	p(y,z\lvert x)\,p_{X_2}(x), \quad q = 2,
	\end{cases}
	\end{equation} 
	where $p_{X_1}(x)$ and $p_{X_2}(x)$ are the probability density functions of the random variables $X_1$ and $X_2$. Based on~\eqref{eq-Qdefine}, we have the following Markov chain
	\begin{equation}\label{eq-MC-1}
	Q\rightarrow X \rightarrow Y\rightarrow Z.
	\end{equation}
	Following along the same lines as~\cite[Appendix~A]{4529277}, one can show that 
	\begin{equation}
	I(X;Y\lvert Z,Q) - I(X;Y\lvert Z) = - I(Q;Y\lvert Z).
	\end{equation}
	Since $I(Q;Y\lvert Z) \geq 0$, $I(X;Y\lvert Z,Q) \leq I(X;Y\lvert Z)$. This implies that $I(X;Y\lvert Z)$ is a concave function in $F_X$. 
	Now, we prove that with the Markov chain $Q\rightarrow X\rightarrow Y\rightarrow Z$, $I(X;Y\lvert Z)$ is strictly concave in $F_X$, i.e., $I(Q;Y\lvert Z) > 0$. Assume, to the contrary, that there exists an $F_X$ such that $I(Q;Y\lvert Z) = 0$. This implies that random variables $Q$, $Y$ and $Z$ also form the Markov chain
	\begin{equation}\label{eq-markovC1}
	Q\rightarrow Z \rightarrow Y.
	\end{equation}
	Furthermore, from the Markov chain~\eqref{eq-MC-1}, we have
	\begin{equation}
	Q\rightarrow X \rightarrow Z. \label{eq-markovC2}
	\end{equation}
	Combining Markov chains~\eqref{eq-markovC1} and~\eqref{eq-markovC2} results in a new Markov chain given by
	\begin{equation}\label{eq-MC-2}
	Q\rightarrow X\rightarrow Z\rightarrow Y.
	\end{equation}
	Now, based on~\eqref{eq-MC-1} and \eqref{eq-MC-2}, we obtain the following
	\begin{align}
	p(y,z,x)\big\vert_{\text{Markov chain~\eqref{eq-MC-1}}} &= p(y,z,x)\big\vert_{\text{Markov chain~\eqref{eq-MC-2}}} \nonumber\\
	p_X(x)\,p(y\vert x)\,p(z\vert y) &= p_X(x)\,p(z\vert x)\,p(y\lvert z) \nonumber\\
	\dfrac{p(y\vert x)}{p(z\vert x)} &=\dfrac{p(y\vert z)}{p(z\vert y)}. \label{eq-cntrdct-1}
	\end{align}
	We note that~\eqref{eq-cntrdct-1} holds for any $y, z\in \mathbb{N}$ and $x\in\mathcal{S}_{F_X}$, where $\mathcal{S}_{F_X}$ is the support set of $F_X$. As a result, for fixed values of $y$ and $z$ the RHS of~\eqref{eq-cntrdct-1} is fixed, while the LHS is a function of $x$. 
	Since $Y\lvert X$ and $Z\lvert X$ are Poisson distributed with mean $(\alpha_Bx+\lambda_B)\Delta$ and $(\alpha_Ex+\lambda_E)\Delta$, respectively, \eqref{eq-cntrdct-1} reduces to
	\begin{equation}\label{eq-cntrdct-2}
	\frac{e^{-(\alpha_Bx+\lambda_B)\Delta}[(\alpha_Bx+\lambda_B)\Delta]^{\,y}/y!}{e^{-(\alpha_Ex+\lambda_E)\Delta}[(\alpha_Ex+\lambda_E)\Delta]^{\,z}/z!} = \frac{p(y\lvert z)}{p(z\lvert y)}.
	\end{equation}
	To reach a contradiction, let us choose $y=z=1$. Now, it is sufficient to show that the LHS of~\eqref{eq-cntrdct-2} is not a constant function in $x$. To this end, let $h(x)$ denote the LHS of~\eqref{eq-cntrdct-2} for $y=z=1$. In this case, we have $h(x) = e^{[(\alpha_E-\alpha_B)x + (\lambda_{E}-\lambda_{B})]\Delta}\frac{\alpha_Bx+\lambda_{B}}{\alpha_Ex+\lambda_E}$. It is clear that $h(x)$ is not a constant function in $x$, for $x\in\mathcal{S}_{F_X}$. This is because at leas one of the inequalities in \eqref{eq-deg-1} or \eqref{eq-deg-2} is strict. Therefore, we reach a contradiction. This, in turn, implies that $I(Q;Y\lvert Z) > 0$ and as a result, $I(X;Y\lvert Z)$ is strictly concave in $F_X$. Furthermore, the output distributions are unique, i.e., if $F_{X_1}$ and $F_{X_2}$ are both secrecy-capacity-achieving, then $p_Y(y;F_{X_1}) = p_Y(y;F_{X_2})$ and $p_Z(z;F_{X_1}) = p_Z(z;F_{X_2})$. 
	
	\section{The Existence of a Mass Point at The Origin}\label{App-A1}
	\begin{figure}[!t]
		\centering
		\includegraphics[width=0.48\textwidth]{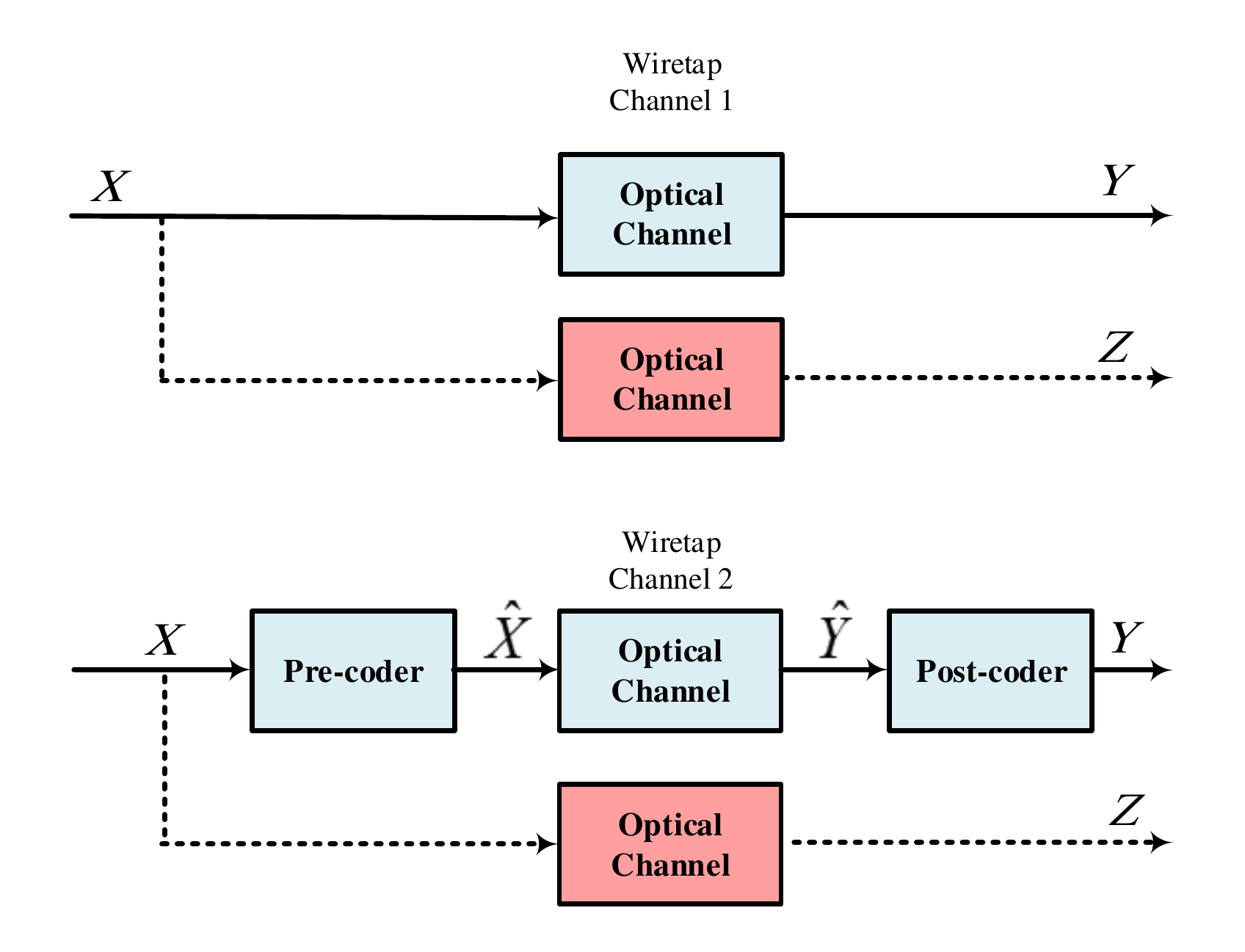}
		\caption{Two discrete-time Poisson wiretap channels.}
		\label{fig_twoWiretapChannels}
	\end{figure}
	Suppose, to the contrary, that $x=0$ does not belong to the support set of the optimal input distribution $\mathcal{S}_{F_X^*}$. Let $0 < x_1\leq x_2 \leq \ldots \leq x_N\leq \infty$ be the mass points in the set $\mathcal{S}_{F_X^*}$. Consider two DT--PWC depicted in Fig.~\ref{fig_twoWiretapChannels}. Wiretap channel 1 is the original optical wiretap channel, and wiretap channel 2 is obtained from wiretap channel 1 by appending a pre-coder and a post-coder before and after the inner optical channel in the legitimate user's link. Specifically, $\hat{X} = X - x_1$ and $Y = \hat{Y} + \hat{N}_{B}$, where $\hat{N}_{B}$ is a Poisson random variable with mean $\alpha_B x_1\Delta$ and is independent from $(X,\hat{Y})$. For any $x \geq x_1$, the conditional probability density functions $p(y\lvert x)$ and $p(z\lvert x)$ are the same in both wiretap channels. Thus, the joint probability density functions of $p(y,x)$ and $p(z,x)$ in the two wiretap channels are also the same, if the input distribution is $F_X^*$. As a result, $C_{S}$ is identical in both wiretap channels. 
	
	In the second wiretap channel, as $X,\hat{X},\hat{Y},Y$ and $Z$ form the Markov chain $X\rightarrow \hat{X} \rightarrow \hat{Y} \rightarrow Y \rightarrow Z$, we have $I(\hat{X};\hat{Y}\lvert Z) \geq I(X;Y\lvert Z)$ by the data processing inequality. This indicates that $I(\hat{X};\hat{Y})- I(\hat{X};Z) \geq I(X;Y) - I(X;Z)$. Now, let $F_{\hat{X}}^*$ be the distribution function of $\hat{X}$ when the distribution function of $X$ is $F_{X}^*$. Clearly, $F_{\hat{X}}^*$ satisfies either of the constraints in \eqref{eq-FeasibleSet1}--\eqref{eq-FeasibleSet3} that are active. Hence, $F_{\hat{X}}^*$ is also secrecy-capacity-achieving for wiretap channel~1. Based on Appendix~\ref{app-strict}, the secrecy-capacity-achieving output distribution is unique, as a result, $p_{Y}(y;F_X^*) = p_{Y}(y;F_{\hat{X}}^*)$. Therefore, for wiretap channel 2, given the input distribution function of $X$ is $F_{X}^*$, the probability density functions for $Y$ and $\hat{Y}$ are the same, which is not possible since $\mathbb{E}[Y] = \mathbb{E}[\hat{Y}] + \alpha_B x_1\Delta$. Hence, we reach a contradiction and the proposition follows. 

\section{Proof of Theorem~\ref{theo-8}}\label{App-LowIntensity}
	To derive \eqref{eq-CS-LIR-PA}, we provide lower and upper bounds on the secrecy capacity and show that these bounds coincide in the low-intensity regime. To that end, we consider the secrecy capacity of the continuous-time PWC and we note that it is a valid upper bound on the secrecy capacity of the DT--PWC across all the intensity regimes. This is because in the continuous-time version, input signals are not restricted to be PAM and can have arbitrary waveforms with an infinite transmission bandwidth. Furthermore, it can be easily shown that the difference between the capacities of the legitimate user's and the eavesdropper's channels is a valid lower bound on the secrecy capacity. 
	
	Based on these arguments, we present two lemmas that provide closed-form expressions for the lower and the upper bounds on the secrecy capacity in the low-intensity regime.
\subsection{Lower Bound}
	\begin{lemma}\label{lem-low}
		The secrecy capacity of the DT--PWC in the low-intensity regime when peak-intensity or both peak- and average-intensity constraints are active is lower bounded by
		\begin{equation}
		C_S \geq C_B - C_E \geq 
		\begin{cases}
		\frac{\mathcal{A}^2}{8}\left(\frac{\alpha_B^2}{\lambda_B} - \frac{\alpha_E^2}{\lambda_E}\right)+o(\mathcal{A}^2),\quad&\text{if} ~\frac{1}{2} \leq p \leq 1,\\
		\frac{\mathcal{A}^2}{2}\,p\,(1-p)\left(\frac{\alpha_B^2}{\lambda_B} - \frac{\alpha_E^2}{\lambda_E}\right)+o(\mathcal{A}^2),\quad&\text{if}~ 0 < p < \frac{1}{2}.
		\end{cases}
		\end{equation}	
		where $C_B$ is the capacity of the legitimate receiver's channel and $C_E$ is the capacity of the eavesdropper's channel, and $o(\mathcal{A}^2)$ refers to the terms that tend to zero faster than $\mathcal{A}^2$, i.e., $\lim_{\mathcal{A}\rightarrow 0} \frac{o(\mathcal{A}^2)}{\mathcal{A}^2} = 0$.
	\end{lemma}
	\begin{proof}
		We start the proof by noting that $C_B \geq I(X^b;Y)$ where $X^b$ is the channel input with a binary distribution. We choose the input distribution to be either $F_X(x) = \frac{1}{2} u(x) + \frac{1}{2} u(x-\mathcal{A})$ when only the peak-intensity constraint is active or to be $F_X(x) = (1-p) u(x) + p u(x-\mathcal{A}),~ 0< p <\frac{1}{2}$, when both peak- and average-intensity constraints are active and they both go to zero with their ratio held fixed at $p$. Now, we follow along similar lines of \cite[Proposition~2]{5773060} to find the closed-form expression of the mutual information $I(X^b;Y)$ in the low-intensity regime when both peak- and average-intensity constraints are active, i.e., $0<p<\frac{1}{2}$. We note that 
		\begin{align}
		I(X^b;Y) =&\, -\frac{1}{\Delta}\sum_{y=0}^{+\infty} \left[(1-p) p(y\vert 0) + p\,p(y\vert \mathcal{A})\right]\!\log\left[(1-p) p(y\vert 0) + p\, p(y\vert \mathcal{A})\right] \notag\\
		& + \frac{(1-p)}{\Delta}\sum_{y=0}^{+\infty} p(y\vert 0)\log(p(y\vert 0)) + p  \sum_{y=0}^{+\infty} p(y\vert \mathcal{A})\log(p(y\vert \mathcal{A}))\notag\\
		=& -\frac{p}{\Delta} \sum_{y=0}^{+\infty} p(y\vert\mathcal{A})\left(\log\frac{p(y\vert 0)}{p(y\vert \mathcal{A})} + \log\left((1-p) + p \frac{p(y\vert \mathcal{A})}{p(y\vert 0)}\right)\right) -\frac{(1-p)}{\Delta} \sum_{y=0}^{+\infty} p(y\vert 0) \notag\\
		&\times \log\left((1-p) + p \frac{p(y\vert \mathcal{A})}{p(y\vert 0)}\right)\notag\\
		=&\frac{1}{\Delta} \underbrace{p\sum_{y=0}^{+\infty}p(y\vert \mathcal{A}) \log\frac{p(y\vert \mathcal{A})}{p(y\vert 0)}}_{\stackrel{\triangle}{=}\, T_1(\mathcal{A})}-\frac{1}{\Delta} \sum_{y=0}^{+\infty}\! \left((1-p) p(y\vert 0) + p p(y\vert \mathcal{A})\right)\notag\\
		&~\times\underbrace{\log\left((1-p) + p \frac{p(y\vert \mathcal{A})}{p(y\vert 0)}\right)}_{\stackrel{\triangle}{=}\, T_2(\mathcal{A},y)}.
		\label{eq-LIR-1}
		\end{align}
		Note that $T_1(\mathcal{A}) = -p\,\alpha_B\mathcal{A}\Delta + p\,(\lambda_B+\alpha_B\mathcal{A})\Delta\log\left(1+\frac{\alpha_B\mathcal{A}}{\lambda_B}\right)$. Now, consider the Taylor expansion of $T_1(\mathcal{A})$ around $\mathcal{A} = 0$ to get 
		\begin{equation}
		T_1(\mathcal{A}) = -p\Delta\,\alpha_B\mathcal{A} + p\Delta\left(\alpha_B\mathcal{A} + \frac{\alpha_B^2}{2\lambda_B}\mathcal{A}^2 + o(\mathcal{A}^2)\right) = p\Delta\left(\frac{\mathcal{A}^2}{2\lambda_B} \alpha_B^2+ o(\mathcal{A}^2)\right),
		\end{equation}
		Now, observe that $T_2(\mathcal{A},y) = \log\left( (1-p) + p\, e^{-\alpha_B\mathcal{A}\Delta} \left(1 + \frac{\alpha_B\mathcal{A}}{\lambda_B}\right)^y\right)$, and the Taylor expansion of $T_2(\mathcal{A},y)$ around $\mathcal{A} = 0$ gives
		\begin{align}
		T_2(\mathcal{A},y) = p\,\alpha_B\left(\frac{y}{\lambda_B} - \Delta \right)\mathcal{A} + \Big(&p \left(\frac{\alpha_B^2 \Delta^2}{2} - \frac{\alpha_B^2 y}{2\lambda_B^2} + \frac{\alpha_B^2 y^2}{2\lambda_B^2} - \frac{\alpha_B^2 y \Delta}{\lambda_B}\right)\notag\\ &- p^2 \alpha_B^2 \frac{(\Delta - \frac{y}{\lambda_B})^2}{2}\Big) \mathcal{A}^2 + \Delta o(\mathcal{A}^2y).\label{eq-LIR-2}
		\end{align}
		Plugging \eqref{eq-LIR-2} into \eqref{eq-LIR-1}, the second term in \eqref{eq-LIR-1} denoted by $T_3(\mathcal{A})$ becomes
		\begin{align}
		T_3(\mathcal{A}) \stackrel{\triangle}{=}\,& - \sum_{y=0}^{+\infty} \left((1-p) p(y\vert 0) + p p(y\vert \mathcal{A})\right)T_2(\mathcal{A},y) = -(1-p) p \left(\frac{\lambda_B\Delta}{\lambda_B} - \Delta\right)\mathcal{A}-(1-p)p\notag\\&\times\left(\frac{\alpha_B^2\Delta^2}{2} -\frac{\alpha_B^2\Delta}{2\lambda_B} + \frac{\alpha_B^2\Delta^2}{2} + \frac{\alpha_B^2\Delta}{2\lambda_B} - \alpha_B^2\Delta^2\right)\mathcal{A}^2
		+(1-p)p^2 \alpha_B^2 \frac{\mathcal{A}^2}{2\lambda_B}\Delta\notag\\
		& -p^2 \alpha_B^2 \frac{\mathcal{A}^2}{\lambda_B}\Delta - p^2 \mathcal{A}^2\Big(\frac{\alpha_B^2\Delta^2}{2} - \frac{\alpha_B^3\mathcal{A}\Delta}{2\lambda_B^2} - \frac{\alpha_B^2\Delta}{2\lambda_B} + \frac{\alpha_B^4\mathcal{A}^2\Delta^2}{2\lambda_B^2} + \frac{\alpha_B^3\mathcal{A}\Delta^2}{\lambda_B^2}\notag\\
		&+\frac{\alpha_B^2 \Delta^2}{2} + \frac{\alpha_B^3\mathcal{A}\Delta}{2\lambda_B^2} + \frac{\alpha_B^2\Delta}{2\lambda_B} - \frac{\alpha_B^3\mathcal{A}\Delta^2}{\lambda_B} - \alpha_B^2\Delta^2 \Big) + p^3\alpha_B^2\frac{\mathcal{A}^2}{2}\Big(\Delta^2 - \frac{2\alpha_B\mathcal{A}\Delta}{\lambda_B} - 2\Delta^2 \notag\\
		& + \frac{\alpha_B^2\mathcal{A}^2\Delta^2}{\lambda_B^2} + \Delta^2 +\frac{2\alpha_B\mathcal{A}\Delta^2}{\lambda_B} + \frac{\alpha_B\mathcal{A}\Delta}{\lambda_B} + \frac{\Delta}{\lambda_B}\Big)+\Delta\,o(\mathcal{A}^2) \notag\\
		=\,& -p^2\frac{\mathcal{A}^2}{2\lambda_B}\alpha_B^2\Delta + \Delta\, o(\mathcal{A}^2).
		\end{align}
		Combining this with $T_1(\mathcal{A})$, one obtains
		\begin{equation}
		I(X^b;Y) =  \frac{\mathcal{A}^2}{2\lambda_B}\alpha_B^2\,p\,(1-p) + o(\mathcal{A}^2\Delta).
		\end{equation}
		Hence, in the regime where $\mathcal{A}\rightarrow 0$, $C_B \geq I(X^b;Y)\geq \frac{\mathcal{A}^2}{2\lambda_B}\alpha_B^2 p(1-p)$. Note that when only the peak-intensity constraint is active, we choose $p = \frac{1}{2}$. Thus, we have
		\begin{equation}\label{eq-LB1}
		C_B \geq
		\begin{cases}
		\frac{\mathcal{A}^2}{8}\frac{\alpha_B^2}{\lambda_B}+o(\mathcal{A}^2),\quad&\text{if} ~\frac{1}{2} \leq p \leq 1,\\
		\frac{\mathcal{A}^2}{2}\,p\,(1-p)\frac{\alpha_B^2}{\lambda_B}+o(\mathcal{A}^2),\quad&\text{if}~ 0 < p < \frac{1}{2}.
		\end{cases}
		\end{equation}
		Next, we observe that $C_E$ can be upper bounded by the capacity of the continuous-time Poisson channel since in this case, the channel input admits infinite bandwidth and is not restricted to be a PAM signal. Therefore, in the low intensity regime $C_E$ can be upper bounded by \cite[Theorem~2]{1056262}
		\begin{equation}\label{eq-LB2}
		C_E \leq  
		\begin{cases}
		\frac{\mathcal{A}^2}{8}\frac{\alpha_E^2}{\lambda_E}+o(\mathcal{A}^2),\quad&\text{if} ~\frac{1}{2} \leq p \leq 1,\\
		\frac{\mathcal{A}^2}{2}\,p\,(1-p)\frac{\alpha_E^2}{\lambda_E}+o(\mathcal{A}^2),\quad&\text{if}~ 0 < p < \frac{1}{2}.
		\end{cases}
		\end{equation}
		Finally, from \eqref{eq-LB1} and \eqref{eq-LB2}, we find that 
		\begin{equation}
		C_S \geq 
		\begin{cases}
		\frac{\mathcal{A}^2}{8}\left(\frac{\alpha_B^2}{\lambda_B} - \frac{\alpha_E^2}{\lambda_E}\right)+o(\mathcal{A}^2),\quad&\text{if} ~\frac{1}{2} \leq p \leq 1,\\
		\frac{\mathcal{A}^2}{2}\,p\,(1-p)\left(\frac{\alpha_B^2}{\lambda_B} -\frac{\alpha_E^2}{\lambda_E}\right)+o(\mathcal{A}^2),\quad&\text{if}~ 0 < p < \frac{1}{2}.
		\end{cases}
		\end{equation}
		This completes the proof of the lemma.
	\end{proof}
\subsection{Upper Bound}	
	\begin{lemma}\label{lem-upp}
		The secrecy capacity of the DT--PWC in the low-intensity regime when peak-intensity or both peak- and average-intensity constraints are active is upper bounded by
		\begin{equation}
		C_S \leq C_S^{CT,LI} = 
		\begin{cases}
		\frac{\mathcal{A}^2}{8}\left(\frac{\alpha_B^2}{\lambda_B} - \frac{\alpha_E^2}{\lambda_E}\right)+o(\mathcal{A}^2),\quad&\text{if} ~\frac{1}{2} \leq p \leq 1,\\
		\frac{\mathcal{A}^2}{2}\,p\,(1-p)\left(\frac{\alpha_B^2}{\lambda_B} - \frac{\alpha_E^2}{\lambda_E}\right)+o(\mathcal{A}^2),\quad&\text{if}~ 0 < p < \frac{1}{2}.
		\end{cases}
		\end{equation}	
		where $C_S^{CT,LI}$ is the secrecy capacity of the degraded continuous-time PWC in the low-intensity regime with either peak-intensity or both peak- and average-intensity constraints.
	\end{lemma}
	\begin{proof}
			We start the proof by noting that the secrecy capacity of the DT--PWC with peak- and average-intensity constraints is upper bounded by the secrecy capacity of the continuous-time PWC with peak- and average-intensity constraints. This is because in the continuous-time version, the input signals are not restricted to be PAM signals and can admit any waveform with very large transmission bandwidth. Now, we recall the secrecy capacity of the degraded continuous-time PWC with a peak-intensity constraint from \cite[Theorem~3]{6294444}
		\begin{equation}
		C_S^{CT} =
		p K(\mathcal{A}) + (1-p) K(0) - K(p\mathcal{A}), \quad  0 \leq p \leq 1,
		\end{equation}
		where $K(x) = (\alpha_B x + \lambda_B)\log(\alpha_B x + \lambda_B) - (\alpha_E x + \lambda_E)\log(\alpha_E x + \lambda_E)$ and where $p$ is the solution of the equation 
		\begin{equation}\label{eq-prob-find}
		K(\mathcal{A}) - K(0) = \mathcal{A} K^{\prime}(p\mathcal{A}),
		\end{equation}
		where $K'(x)$ denotes the derivative of $K(x)$ with respect to $x$.
		The secrecy capacity $C_S^{CT}$ is achieved by a binary input distributions with mass points at $\{0,\mathcal{A}\}$ and respective probabilities $\{1-p,p\}$. We now find the closed-form expression of $C_S^{CT}$ in the regime where $\mathcal{A}\rightarrow 0$. To this end, we expand $K(\mathcal{A})$ around $\mathcal{A} = 0$ and we get 
		\begin{equation}
		K(\mathcal{A}) = \log \frac{\lambda_B^{\lambda_B}}{\lambda_E^{\lambda_E}} + \left((\alpha_B - \alpha_E) + \log\frac{\lambda_B}{\lambda_E}\right)\mathcal{A} + \left(\frac{\alpha_B^2}{\lambda_B} - \frac{\alpha_E^2}{\lambda_E}\right)\frac{\mathcal{A}^2}{2} + o(\mathcal{A}^2).
		\end{equation}
		Therefore, plugging this expansion into \eqref{eq-prob-find}, the optimal $p$ in the regime where $\mathcal{A}\rightarrow 0$ is given by 
		\begin{equation}
		\left(\frac{\alpha_B^2}{\lambda_B} - \frac{\alpha_E^2}{\lambda_E}\right)\frac{\mathcal{A}^2}{2} = \left(\frac{\alpha_B^2}{\lambda_B} - \frac{\alpha_E^2}{\lambda_E}\right)p\mathcal{A}^2 \Rightarrow p = \frac{1}{2}.
		\end{equation}
		Thus, $C_S^{CT}$ in the regime where $\mathcal{A}\rightarrow 0$ is denoted by $C_S^{CT,LI}$ and is given by 
		\begin{equation}
		C_S^{CT,LI} = \frac{\mathcal{A}^2}{8}\left(\frac{\alpha_B^2}{\lambda_B} - \frac{\alpha_E^2}{\lambda_E}\right)+o(\mathcal{A}^2).
		\end{equation}
		Furthermore, we observe that when both peak- and average-intensity constraints are active, the optimal input distribution is also binary with mass points $\{0,\mathcal{A}\}$ and respective probabilities $\{1-p, p\}$ where $p = \frac{\mathcal{E}}{\mathcal{A}}$. Therefore, in the regime where $\mathcal{A}\rightarrow 0 $ and $\mathcal{E}\rightarrow 0$ with their ratio held fixed at $p$, $C_S^{CT,LI}$ is 
		\begin{equation}
		C_S^{CT,LI} = \frac{\mathcal{A}^2}{2}p(1-p)\left(\frac{\alpha_B^2}{\lambda_B} - \frac{\alpha_E^2}{\lambda_E}\right)+o(\mathcal{A}^2).
		\end{equation}
		Finally, note that since $p = \frac{1}{2}$ is the optimal value of \eqref{eq-prob-find}, we conclude that when $\frac{\mathcal{E}}{\mathcal{A}} \geq \frac{1}{2}$, $p = \frac{1}{2}$ and the average-intensity constraint is not active. This completes the proof of the lemma.
	\end{proof}
	We observe that the lower and upper bounds on the secrecy capacity of the DT--PWC asymptotically coincide, when peak-intensity or both peak- and average-intensity constraints are active. Thus, we can fully characterize the asymptotic secrecy capacity of the DT--PWC in this regime. 

\section{Proof of Theorem~\ref{theo-actv-p-avg}}\label{App-LowerBoundTight}	
In this appendix, we prove that when the peak-intensity constraint $\mathcal{A}$ is held fixed, while the average-intensity constraint $\mathcal{E}$ tends to zero, the asymptotic secrecy capacity satisfies \eqref{eq-theo-9}, and it scales linearly in $\mathcal{E}$. Additionally, we characterize the optimal input distribution that achieves the asymptotic secrecy capacity in this regime.

\subsection{Asymptotic Secrecy Capacity Expression}
We start the proof by first making an important observation. We observe that when the channel input of the degraded DT--PWC is constrained by an average-intensity constraint (regardless of the peak-intensity constraint being active or not), i.e., $\mathbb{E}[X]\leq \mathcal{E}$, the input alphabet contains a \textit{zero-cost} symbol. 

Here, by \textit{cost}, we mean the mapping $b:[0,\mathcal{A}]\rightarrow [0,\mathcal{A}], b(x) = x$. Note that the input alphabet contains a zero-cost symbol because $b(0) = 0$. Thus, the average-intensity constraint can be expressed as $\mathbb{E}[b(X)]\leq \mathcal{E}$. Intuitively speaking, ''0`` will contribute nothing to the average constraint while it belong to the set $[0,\mathcal{A}]$.

Next, we recall the secrecy capacity per unit cost argument established by El-Halabi \textit{et al.} in~\cite[Theorem~1, Theorem~2]{6584947} for a degraded wiretap channel which is stated by the following lemma.
\begin{lemma}\label{lemma-SecCapUnitCost}
	The secrecy capacity per unit cost, denoted by $C_{S,PUC}$, of the stochastically degraded DT--PWC with a zero-cost input letter and when the channel input is constrained by both peak- and average-intensity constraints, is given by
	\begin{equation}\label{eq-SecCapUC-ZC}
	C_{S,PUC} = \sup_{\mathcal{E}>0}\frac{C_S(\mathcal{A},\mathcal{E})}{\mathcal{E}} = \frac{1}{\Delta} \sup_{x\in[0,\mathcal{A}]} \frac{N(x)}{b(x)},
	\end{equation}
	where $C_S(\mathcal{A},\mathcal{E})$ is the secrecy capacity of the DT--PWC under peak- and average-intensity constraints, and $N(x)$ is defined as  
	\begin{equation}\label{eq-N_x}
	N(x) \stackrel{\triangle}{=} D\left(p_{Y\vert X=x}(y\vert x)\parallel p_{Y\vert X=0}(y\vert 0)\right) - D\left(p_{Z\vert X=x}(z\vert x)\parallel p_{Z\vert X=0}(z\vert 0)\right), 
	\end{equation}
	with $D(p\parallel q)$ denoting the Kullback-Liebler divergence between two probability distributions $p$ and $q$.
\end{lemma}
\begin{proof}
	The proof follows along similar lines of~\cite[Theorem~1, Theorem~2]{6584947}.
\end{proof}
Lemma~\ref{lemma-SecCapUnitCost} is instrumental to establish the behavior of the secrecy capacity in the low-intensity regime due to the fact that one can link the secrecy capacity per unit cost of the degraded DT--PWC to the secrecy capacity with an average-intensity constraint in the regime where $\mathcal{E}\rightarrow 0$. To establish this link, we need to prove that the secrecy capacity of the degraded DT--PWC with an average-intensity constraint (regardless of the presence of a peak-intensity constraint) is a concave function in the average-intensity constraint. This is formally presented by the following lemma.
\begin{lemma}
	The secrecy capacity of the degraded DT--PWC with an average-intensity constraint and regardless of the existence of a peak-intensity constraint is a concave function in the average-intensity constraint.
\end{lemma} 
\begin{proof}
	Without loss of generality, we assume that both the peak- and average-intensity constraints are active. To prove the concavity of $C_S(\mathcal{A},\mathcal{E})$ in $\mathcal{E}$, we first observe that due to the degradedness (i.e., $X\rightarrow Y\rightarrow Z$), we have 
	\begin{align}
	C_S &= I(X^*;Y) - I(X^*;Z)\notag\\
	 &= I(X^*;Y\vert Z),  
	\end{align}
	where $X^*\in[0,\mathcal{A}],\,\mathbb{E}[X^*] = \mathcal{E}$ is the channel input distributed according to the secrecy-capacity-input distribution.
	Next, let $X_1^*\in[0,\mathcal{A}]$ and $X_2^*\in[0,\mathcal{A}]$ be the optimal channel inputs of the DT--PWC with corresponding 
	$\mathbb{E}[X_1^*] = \mathcal{E}_1$, and $\mathbb{E}[X_2^*] = \mathcal{E}_2$, respectively, such that $\mathcal{E} = \delta\mathcal{E}_1 + (1-\delta)\mathcal{E}_2$, where $\delta\in[0,1]$. Furthermore, let $Q$ be a binary-valued random variable in such a way that 
	\begin{equation}\label{eq-Qdefine-1}
	p(y,z,x\lvert q) = 
	\begin{cases}
	p(y,z\lvert x)\,p_{X_1}^*(x), \quad q = 1, \\
	p(y,z\lvert x)\,p_{X_2}^*(x), \quad q = 2,
	\end{cases}
	\end{equation} 
	where $\Pr\lbrace Q = 1\rbrace = \delta$, $p_{X_1}^*(x)$ and $p_{X_2}^*(x)$ are the optimal probability mass functions of the random variables $X_1^*$ and $X_2^*$, respectively. Based on~\eqref{eq-Qdefine-1}, we have the following Markov chain
	\begin{equation}\label{eq-SecCapUC-1}
	Q\rightarrow X^* \rightarrow Y\rightarrow Z.
	\end{equation}
	Following along similar lines of~\cite[Appendix~A]{4529277}, one can show that 
	\begin{equation}
	I(X^*;Y\lvert Z,Q) \leq I(X^*;Y\lvert Z),
	\end{equation}
	or equivalently 
	\begin{equation}\label{eq-SecCaUC-2}
	\delta\, I(X_1^*;Y\lvert Z) + (1 - \delta)\, I(X_2^*;Y\lvert Z) \leq I(X^*;Y\vert Z).
	\end{equation}
	Observe that \eqref{eq-SecCaUC-2} is equivalent to
	\begin{equation}
		\delta\, C_S(\mathcal{A},\mathcal{E}_1) + (1 - \delta)\, C_S(\mathcal{A},\mathcal{E}_2) \leq C_S\left(\mathcal{A},\delta \mathcal{E}_1 + (1-\delta) \mathcal{E}_2\right),
	\end{equation}
	i.e., the secrecy capacity $C_S(\mathcal{A},\mathcal{E})$ is a concave function in $\mathcal{E}$. Since $\mathcal{E}>0$, the function $C_S(\mathcal{A},\mathcal{E})$ is concave on $(0,+\infty)$. This completes the proof of the lemma.
\end{proof}
Now that we showed $C_S(\mathcal{A},\mathcal{E})$ is a concave function in $\mathcal{E}$, we are ready to establish the link between the secrecy capacity per unit cost and the secrecy capacity in the regime where $\mathcal{E}\rightarrow 0$. Notice that due to the concavity of $C_S(\mathcal{A},\mathcal{E})$ on $(0,+\infty)$, the function $\frac{C_S(\mathcal{A},\mathcal{E})}{\mathcal{E}}$ is monotone and nonincreasing over $(0,+\infty)$. Hence, we have
\begin{equation}\label{eq-SecCapUC-link}
	\sup_{\mathcal{E}>0}\frac{C_S(\mathcal{A},\mathcal{E})}{\mathcal{E}} = \lim_{\mathcal{E}\rightarrow 0} \frac{C_S(\mathcal{A},\mathcal{E})}{\mathcal{E}}.
\end{equation}
From \eqref{eq-SecCapUC-ZC} and \eqref{eq-SecCapUC-link}, we deduce that 
\begin{equation}
	\lim_{\mathcal{E}\rightarrow 0} \frac{C_S(\mathcal{A},\mathcal{E}) }{\mathcal{E}}=  \frac{1}{\Delta}  \sup_{x\in[0,\mathcal{A}]}\frac{N(x)}{x},
\end{equation}
where $N(x)$ is given by \eqref{eq-N_x}. We continue the proof by expanding $N(x)$ and  plugging in $p_{Y\vert X=x}(y\vert x)$, $p_{Y\vert X=0}(y\vert 0)$, $p_{Z\vert X}(z\vert x)$, and $p_{Z\vert X=0}(z\vert 0)$ into \eqref{eq-N_x}. After some algebraic manipulations, we find that
\begin{equation}
	\lim_{\mathcal{E}\rightarrow 0} \frac{C_S(\mathcal{A},\mathcal{E}) }{\mathcal{E}} = \sup_{x\in[0,\mathcal{A}]} \Phi(x) ,\label{eq-LowIntensity-peakavg}
\end{equation} 
where $\Phi(x)$ is defined as
\begin{equation}
\Phi(x) \stackrel{\triangle}{=}\left[(\alpha_E-\alpha_B) + \left(\alpha_B + \frac{\lambda_B}{x}\right)\log\left(1 + \frac{\alpha_B x}{\lambda_B}\right)
-\left(\alpha_E + \frac{\lambda_E}{x}\right)\log\left(1 + \frac{\alpha_E x}{\lambda_E}\right)\right].\label{eq-Phi}
\end{equation}
Next, we observe that $\Phi(x)$ is a strictly increasing function over $[0,\mathcal{A}]$. This is formally established below.
\begin{proposition}\label{prop-app}
The function $\Phi(x)$ is a strictly increasing function over $[0,\mathcal{A}]$ whenever at least one of the following inequalities
\begin{equation}\label{eq-app-ineq}
\begin{cases}
	\alpha_B \geq \alpha_E \\
	\frac{\lambda_E}{\alpha_E} \geq \frac{\lambda_B}{\alpha_B},
\end{cases}
\end{equation}
is strict.
\end{proposition}
\begin{proof}
	To prove that $\Phi(x)$ is strictly increasing on $[0,\mathcal{A}]$, we take the derivative of $\Phi(x)$ and find that 
	\begin{equation}\label{eq-app-strctincr}
		\frac{d\Phi(x)}{dx} = \frac{(\alpha_B - \alpha_E)x + \lambda_E\,\log\left(1 + \frac{\alpha_E x}{\lambda_E}\right) - \lambda_B\,\log\left(1 + \frac{\alpha_B x}{\lambda_B}\right)}{x^2}.
	\end{equation}
	Now, in order to establish strictly increasing, we need to show that $\frac{d\Phi(x)}{dx} > 0$. To this end, we note that the numerator of \eqref{eq-app-strctincr} is positive for all $x>0$ whenever at least one of the inequalities in \eqref{eq-app-ineq} is strict. This is because the numerator of \eqref{eq-app-strctincr} is a strictly increasing function for $x>0$ provided that at least one of the inequalities in \eqref{eq-app-ineq} is strict. To be more specific, let $\Pi(x)$ denote the numerator of \eqref{eq-app-strctincr}, i.e.,
	\begin{equation}
		\Pi(x) \stackrel{\triangle}{=} (\alpha_B - \alpha_E)x + \lambda_E\,\log\left(1 + \frac{\alpha_E x}{\lambda_E}\right) - \lambda_B\,\log\left(1 + \frac{\alpha_B x}{\lambda_B}\right).
	\end{equation}  
	We have that 
	\begin{align}
		\frac{d\Pi(x)}{dx} &= \frac{\alpha_B^2 x}{\alpha_Bx + \lambda_B} - \frac{\alpha_E^2 x}{\alpha_Ex + \lambda_E}\notag\\ 
		&=\frac{\alpha_B\alpha_E(\alpha_B-\alpha_E)x^2 + x\alpha_B\alpha_E\left(\alpha_B\frac{\lambda_E}{\alpha_E}-\alpha_E\frac{\lambda_B}{\alpha_B}\right)}{(\alpha_Bx+\lambda_{B})(\alpha_Ex+\lambda_E)}.
	\end{align}
	Notice that $\frac{d\Pi(x)}{dx} > 0$ for all $x>0$ provided that at least one of the inequalities in \eqref{eq-app-ineq} is strict. This implies that $\Pi(x)$ is strictly increasing for all $x>0$. In other words, for all $x>0$, we have $\Pi(x) > \Pi(0) = 0$. This, in turn, implies that for all $x>0$, we have $\frac{d\Phi(x)}{dx} > 0$. This completes the proof of the proposition. 
\end{proof}
From Proposition~\ref{prop-app}, we infer that 
$\sup_{x\in[0,\mathcal{A}]} \Phi(x) = \Phi(\mathcal{A})$. Substituting this into \eqref{eq-LowIntensity-peakavg}, we get
\begin{equation}
\lim_{\mathcal{E}\rightarrow 0} \frac{C_S(\mathcal{A},\mathcal{E}) }{\mathcal{E}}= \Phi(\mathcal{A}).
\end{equation}
Therefore, we find that the asymptotic secrecy capacity satisfies \eqref{eq-theo-9}.

\subsection{Optimal Input Distribution}
Now that we have found the closed-form expression of the asymptotic secrecy capacity in the regime $\mathcal{E}\rightarrow 0$ while $\mathcal{A}$ is held fixed, we will strive to find an input distribution that attains the secrecy capacity. To this end, we invoke similar arguments of the proof of Lemma~\ref{lem-low} in Appendix~\ref{App-LowIntensity}. We again resort to a binary input distribution with mass points at $\{0,\mathcal{A}\}$ and corresponding probability masses $\{1-p,p\}$ where $p = \frac{\mathcal{E}}{\mathcal{A}}$ and $ \mathcal{E}\rightarrow 0$. This choice will lead us to find that
\begin{align}
I(X^b;Y) - I(X^b;Z) =&\, \frac{p}{\Delta}\sum_{y=0}^{+\infty}p(y\vert \mathcal{A}) \log\frac{p(y\vert \mathcal{A})}{p(y\vert 0)} -\! \sum_{y=0}^{+\infty}\! \frac{\left((1-p) p(y\vert 0) + p p(y\vert \mathcal{A})\right)}{\Delta}\notag\\&\times\log\left((1-p) + p \frac{p(y\vert \mathcal{A})}{p(y\vert 0)}\right)\notag\\&-\frac{p}{\Delta}\sum_{z=0}^{+\infty}p(z\vert \mathcal{A}) \log\frac{p(z\vert \mathcal{A})}{p(z\vert 0)} +\! \sum_{z=0}^{+\infty}\! \frac{\left((1-p) p(z\vert 0) + p p(z\vert \mathcal{A})\right)}{\Delta}\notag\\&\times\log\left((1-p) + p \frac{p(z\vert \mathcal{A})}{p(z\vert 0)}\right).
\end{align}
After some algebraic manipulations, we obtain 
\begin{align}
&I(X^b;Y) - I(X^b;Z) \notag\\
=\, &\left[(\alpha_E - \alpha_B) + \left(\alpha_B + \frac{\lambda_B}{\mathcal{A}}\right)\log\left(1 + \frac{\alpha_B\mathcal{A}}{\lambda_B}\right)-\left(\alpha_E + \frac{\lambda_E}{\mathcal{A}}\right)\log\left(1 + \frac{\alpha_E\mathcal{A}}{\lambda_E}\right)\right]\mathcal{E}\notag\\
& - \frac{1-p}{\Delta}\mathbb{E}_{Y\vert X=0}\left[\log\left(1 + r e^{-\alpha_B\mathcal{A}\Delta}\xi_B^Y\right)\right]-\frac{p}{\Delta} \mathbb{E}_{Y\vert X=\mathcal{A}}\left[\log\left(1 + r e^{-\alpha_B\mathcal{A}\Delta}\xi_B^Y\right)\right]\notag\\
&+ \frac{1-p}{\Delta}\mathbb{E}_{Z\vert X=0}\left[\log\left(1 + r e^{-\alpha_E\mathcal{A}\Delta}\xi_E^Z\right)\right]+\frac{p}{\Delta} \mathbb{E}_{Z\vert X=\mathcal{A}}\left[\log\left(1 + r e^{-\alpha_E\mathcal{A}\Delta}\xi_E^Z\right)\right],\label{eq-AppD-0}
\end{align}
where $r\stackrel{\triangle}{=}\frac{p}{1-p}$, $\xi_B \stackrel{\triangle}{=} 1 + \frac{\alpha_B\mathcal{A}}{\lambda_B}$, and  $\xi_E \stackrel{\triangle}{=} 1 + \frac{\alpha_E\mathcal{A}}{\lambda_E}$. Now, observe that since $\mathcal{E}\rightarrow 0$, $p\rightarrow 0$ and as a result $r\rightarrow 0$. Hence, one can approximate $\log\left(1 + re^{-\alpha_B\mathcal{A}\Delta}\xi_B^Y\right)\sim re^{-\alpha_B\mathcal{A}\Delta}\xi_B^Y$. Plugging this approximation into \eqref{eq-AppD-0}, we get
\begin{align}
&I(X^b;Y) - I(X^b;Z) \notag\\
=\,&\left[(\alpha_E - \alpha_B) + \left(\alpha_B + \frac{\lambda_B}{\mathcal{A}}\right)\log\left(1 + \frac{\alpha_B\mathcal{A}}{\lambda_B}\right)-\left(\alpha_E + \frac{\lambda_E}{\mathcal{A}}\right)\log\left(1 + \frac{\alpha_E\mathcal{A}}{\lambda_E}\right)\right]\mathcal{E}\notag\\
& - \frac{1-p}{\Delta}r e^{-\alpha_B\mathcal{A}\Delta}\mathbb{E}_{Y\vert X=0}\left[\xi_B^Y\right]-\frac{p}{\Delta}r e^{-\alpha_B\mathcal{A}\Delta} \mathbb{E}_{Y\vert X=\mathcal{A}}\left[\xi_B^Y\right]\notag\\
&+ \frac{1-p}{\Delta}r e^{-\alpha_E\mathcal{A}\Delta}\mathbb{E}_{Z\vert X=0}\left[\xi_E^Z\right]+\frac{p}{\Delta}r e^{-\alpha_E\mathcal{A}\Delta} \mathbb{E}_{Z\vert X=\mathcal{A}}\left[\xi_E^Z\right].\label{eq-AppD-1}
\end{align} 
Since $Y\vert X=x$ and $Z\vert X=x$ are Poisson distributed random variables with means $(\alpha_Bx + \lambda_B)\Delta$ and $(\alpha_Ex + \lambda_E)\Delta$, respectively, we have that $\mathbb{E}_{Y\vert X = x}\left[\xi_B^Y \right] = e^{[(\alpha_Bx+\lambda_B)\Delta]\left(\xi_B-1\right)}$ and $\mathbb{E}_{Z\vert X = x}\left[\xi_E^Z \right] = e^{[(\alpha_Ex+\lambda_E)\Delta]\left(\xi_E-1\right)}$. Therefore, $\eqref{eq-AppD-1}$ becomes
\begin{align}
&I(X^b;Y) - I(X^b;Z) \notag\\
=\,&\left[(\alpha_E - \alpha_B) + \left(\alpha_B + \frac{\lambda_B}{\mathcal{A}}\right)\log\left(1 + \frac{\alpha_B\mathcal{A}}{\lambda_B}\right)-\left(\alpha_E + \frac{\lambda_E}{\mathcal{A}}\right)\log\left(1 + \frac{\alpha_E\mathcal{A}}{\lambda_E}\right)\right]\mathcal{E}\notag\\
& - \frac{1-p}{\Delta}r -\frac{p}{\Delta}r e^{\frac{\left(\alpha_B\mathcal{A}\right)^2\Delta}{\lambda_B}}+\frac{1-p}{\Delta}r +\frac{p}{\Delta}r e^{\frac{\left(\alpha_E\mathcal{A}\right)^2\Delta}{\lambda_E}}\notag\\
=&\,\left[(\alpha_E - \alpha_B) + \left(\alpha_B + \frac{\lambda_B}{\mathcal{A}}\right)\log\left(1 + \frac{\alpha_B\mathcal{A}}{\lambda_B}\right)-\left(\alpha_E + \frac{\lambda_E}{\mathcal{A}}\right)\log\left(1 + \frac{\alpha_E\mathcal{A}}{\lambda_E}\right)\right]\mathcal{E}\notag\\
& + \frac{\mathcal{E}^2}{(\mathcal{A}^2 - \mathcal{A}\mathcal{E})\Delta}\left(e^{\frac{\left(\alpha_E\mathcal{A}\right)^2\Delta}{\lambda_E}}-e^{\frac{\left(\alpha_B\mathcal{A}\right)^2\Delta}{\lambda_B}}\right)\notag\\
=&\, \left[(\alpha_E - \alpha_B) + \left(\alpha_B + \frac{\lambda_B}{\mathcal{A}}\right)\log\left(1 + \frac{\alpha_B\mathcal{A}}{\lambda_B}\right)-\left(\alpha_E + \frac{\lambda_E}{\mathcal{A}}\right)\log\left(1 + \frac{\alpha_E\mathcal{A}}{\lambda_E}\right)\right]\mathcal{E}+o(\mathcal{E}).\label{eq-AppD-2}
\end{align} 
Thus, we observe that the binary input distribution with mass points at $\{0,\mathcal{A}\}$ with corresponding probability masses $\{1-\frac{\mathcal{E}}{\mathcal{A}},\frac{\mathcal{E}}{\mathcal{A}}\}$ where $\mathcal{E}\rightarrow 0$, asymptotically achieves the asymptotic secrecy capacity. This completes the proof of the theorem.

\section{Proof of Theorem~\ref{theo-avg-zero}}\label{App-theo-avg-zero}	
In this appendix, we prove that when the peak-intensity constraint $\mathcal{A}$ is inactive, i.e., $\mathcal{A} = +\infty$, and the average-intensity constraint $\mathcal{E}$ tends to zero, the asymptotic secrecy capacity satisfies \eqref{eq-theo-10}, and it scales linearly in $\mathcal{E}$. We follow along similar lines of Appendix~\ref{App-LowerBoundTight} and establish the behavior of the asymptotic secrecy capacity. We note that based on \eqref{eq-LowIntensity-peakavg}, and in the absence of the peak-intensity constraint with $\alpha_B = \alpha_E$, the asymptotic secrecy capacity satisfies
\begin{align}
	\lim_{\mathcal{E}\rightarrow 0} \frac{C_S(\mathcal{A},\mathcal{E}) }{\mathcal{E}} &= \sup_{x\in[0,+\infty)} \left[\left(\alpha_B + \frac{\lambda_B}{x}\right)\log\left(1 + \frac{\alpha_B x}{\lambda_B}\right)
	-\left(\alpha_B + \frac{\lambda_E}{x}\right)\log\left(1 + \frac{\alpha_B x}{\lambda_E}\right)\right]\notag\\
	&\stackrel{(a)}{=} \lim_{x\rightarrow +\infty}\left[\left(\alpha_B + \frac{\lambda_B}{x}\right)\log\left(1 + \frac{\alpha_B x}{\lambda_B}\right)
	-\left(\alpha_B + \frac{\lambda_E}{x}\right)\log\left(1 + \frac{\alpha_B x}{\lambda_E}\right)\right]\notag\\
	&= \alpha_B\log\left(\frac{\lambda_E}{\lambda_B}\right),
\end{align}
where $(a)$ is justified since $\Phi(x)$ is a strictly increasing function for all $x>0$ (as shown in Proposition~\ref{prop-app} in Appendix~\ref{App-LowerBoundTight}). This completes the proof of the theorem.

\section{Proof of Theorem~\ref{theo-avg-diff-zero}}\label{App-avg-diff-zero}
Before starting the proof, we need to state that invoking the secrecy capacity per unit cost argument, which we used in proving Theorem~\ref{theo-actv-p-avg} and Theorem~\ref{theo-avg-zero}, does not lead to a sensible results to establish the Theorem at hand, i.e., Theorem~\ref{theo-avg-diff-zero}. This is because in the absence of the peak-intensity constraint, the asymptotic secrecy capacity, which is identical to the secrecy capacity per unit cost, is given by 
\begin{equation}
	\lim_{\mathcal{E}\rightarrow 0} \frac{C_S(\mathcal{E})}{\mathcal{E}} = \sup_{x\in[0,+\infty)} \Phi(x),
\end{equation}
where $\Phi(x)$ is defined in \eqref{eq-Phi}. We note that when $\alpha_B > \alpha_E$ and $\frac{\lambda_E}{\alpha_E}\geq \frac{\lambda_B}{\alpha_B}$, $\Phi(x)$ is a strictly increasing function for all $x>0$ as shown in Proposition~\ref{prop-app}. Thus, $\sup_{x\in[0,+\infty)}\Phi(x) = \lim_{x\rightarrow +\infty}\Phi(x)$. Now, we observe that because $\alpha_B > \alpha_E$ and $\frac{\lambda_E}{\alpha_E}\geq \frac{\lambda_B}{\alpha_B}$, the limit
\begin{equation}
	\lim_{x\rightarrow +\infty}\Phi(x) = +\infty,
\end{equation}
which does not lead into a closed-form expression for the asymptotic secrecy capacity in this case. This implies that in this case, the asymptotic secrecy capacity must grow faster than the linear growth established in Theorem~\ref{theo-actv-p-avg} and Theorem~\ref{theo-avg-zero}.  

To circumvent this issue, we resort to providing lower and upper bounds on the secrecy capacity and we will strive to characterize the asymptotic secrecy capacity using the provided bounds. 

\subsection{Lower Bound}
To find a lower bound on the secrecy capacity, we evaluate the mutual information difference $I(X;Y)-I(X;Z)$ for the binary input distribution with mass points located at $\{0,\zeta\}$ with corresponding probability masses $\{1-p,p\}$, where $\zeta\stackrel{\triangle}{=} \sqrt{\frac{\lambda_{B}}{\alpha_B^2\Delta}\log \frac{1}{\mathcal{E}}}$ and $p = \frac{\mathcal{E}}{\zeta}$. We note that in this case, as $\mathcal{E}\rightarrow 0$, $\zeta\rightarrow +\infty$ and $p\rightarrow 0$. Next, we try to lower bound the secrecy capacity by $I(X^b;Y) - I(X^b;Z)$, where $X^b$ is the input random variable distributed according to the aforementioned binary distribution. To this end, we lower bound $I(X^b;Y)$ and upper bound $I(X^b;Z)$. We start by lower bounding $I(X^b;Y)$ as follows
\begin{align}
& I(X^b;Y)\notag\\
=& -\alpha_B\mathcal{E} + \left(\alpha_B\mathcal{E} + \frac{\lambda_B\mathcal{E}}{\zeta}\right)\log\left(1 + \frac{\alpha_B\zeta}{\lambda_B}\right)
-\frac{1-p}{\Delta} \mathbb{E}_{Y\vert X=0}\big[\underbrace{\log\left(1-p + p e^{-\alpha_B\zeta\Delta}\xi_{B,1}^Y\right)}_{\leq\, p\, e^{-\alpha_B\zeta\Delta}\,\xi_{B,1}^Y}\big]\notag\\
&- \frac{p}{\Delta}\mathbb{E}_{Y\vert X=\zeta}\big[\underbrace{\log\left(1-p + p e^{-\alpha_B\zeta\Delta}\xi_{B,1}^Y\right)}_{\leq\, p\, e^{-\alpha_B\zeta\Delta}\,\xi_{B,1}^Y}\big]\notag\\
\geq&-\alpha_B\mathcal{E} + \left(\alpha_B\mathcal{E} + \frac{\lambda_B\mathcal{E}}{\zeta}\right)\log\left(1 + \frac{\alpha_B\zeta}{\lambda_B}\right)-\frac{(1-p)p}{\Delta}\notag\\
&-\frac{p^2}{\Delta}e^{\frac{(\alpha_B\zeta)^2}{\lambda_B}\Delta}\underbrace{\sum_{y=0}^{+\infty}e^{-\left(\frac{(\alpha_B\zeta)^2}{\lambda_B}+2\alpha_B\zeta+\lambda_B\right)\Delta}\frac{\left(\left[\frac{(\alpha_B\zeta)^2}{\lambda_B}+2\alpha_B\zeta+\lambda_B\right]\Delta\right)^y}{y!}}_{=\,1}\notag\\
=& -\alpha_B\mathcal{E} + \left(\alpha_B\mathcal{E} + \frac{\lambda_B\mathcal{E}}{\zeta}\right)\log\left(1 + \frac{\alpha_B\zeta}{\lambda_B}\right)-\frac{p}{\Delta}-\frac{p^2}{\Delta}\left(e^{\frac{(\alpha_B\zeta)^2}{\lambda_B}\Delta}-1\right),\label{eq-AppF-Bob-0}
\end{align}
where $\xi_{B,1} \stackrel{\triangle}{=} 1 + \frac{\alpha_B\zeta}{\lambda_B}$.

Next, we upper bound $I(X^b;Z)$ as follows
\begin{align}
& I(X^b;Z)\notag\\
=& -\alpha_E\mathcal{E} + \left(\alpha_E\mathcal{E} + \frac{\lambda_E\mathcal{E}}{\zeta}\right)\log\left(1 + \frac{\alpha_E\zeta}{\lambda_E}\right)
-\frac{1-p}{\Delta} \mathbb{E}_{Z\vert X=0}\big[\underbrace{\log\left(1-p + p e^{-\alpha_E\zeta\Delta}\xi_{E,1}^Z\right)}_{\stackrel{(a)}{\geq}\, (1-p)\log(1) + p\log\left(e^{-\alpha_E\zeta\Delta}\xi_{E,1}^Z\right)}\big]\notag\\
&- \frac{p}{\Delta}\mathbb{E}_{Z\vert X=\zeta}\big[\underbrace{\log\left(1-p + p e^{-\alpha_E\zeta\Delta}\xi_{E,1}^Y\right)}_{\geq\,  (1-p)\log(1) + p\log\left(e^{-\alpha_E\zeta\Delta}\xi_{E,1}^Z\right)}\big]\notag\\
\geq& -\alpha_E\mathcal{E} + \left(\alpha_E\mathcal{E} + \frac{\lambda_E\mathcal{E}}{\zeta}\right)\log\left(1 + \frac{\alpha_E\zeta}{\lambda_E}\right) + \frac{(1-p)p}{\Delta}\alpha_E\zeta\Delta-\frac{(1-p)p}{\Delta}\lambda_E\Delta\log(\xi_{E,1})\notag\\
&+\frac{p^2}{\Delta}\alpha_E\zeta\Delta-\frac{p^2}{\Delta}(\alpha_E\zeta+\lambda_E)\Delta\log(\xi_{E,1})\notag\\
=&\, (1-p)\alpha_E\mathcal{E}\log(\xi_{E,1}),\label{eq-AppF-Eve0}
\end{align}
where $\xi_{E,1} \stackrel{\triangle}{=} 1 + \frac{\alpha_E\zeta}{\lambda_E}$. In \eqref{eq-AppF-Eve0}, the inequality $(a)$ is justified because $\log(\cdot)$ is a concave function. To establish the inequality, we consider the argument of the logarithm to be the expected value of a binary random variable, say $T$, with mass points located at $\{1,e^{-\alpha_E\zeta\Delta}\xi_{E,1}^Z\}$ with corresponding probability masses $\{1-p,p\}$. Therefore, from the concavity of the logarithm function, we have that $\mathbb{E}[\log(T)]\leq \log(\mathbb{E}[T])$.

Now, that we found a lower bound for $I(X^b;Y)$ and an upper bound for $I(X^b;Z)$, we can lower bound the secrecy capacity by combining \eqref{eq-AppF-Bob-0}--\eqref{eq-AppF-Eve0} as
\begin{align}
C_S \geq&\, I(X^b;Y)-I(X^b;Z)\notag\\
\geq&\, \left[\alpha_B\log(\xi_{B,1})-\alpha_E\log(\xi_{E,1})\right]\mathcal{E}-\alpha_B\mathcal{E}+\frac{\alpha_E\log(\xi_{E,1})}{\zeta}\mathcal{E}^2 + \frac{\lambda_B\log(\xi_{B,1})}{\zeta}\mathcal{E} - \frac{\mathcal{E}}{\Delta\zeta}\notag\\
&\,-\frac{\mathcal{E}^2}{\Delta\zeta^2}\left(e^{\frac{(\alpha_B\zeta)^2}{\lambda_B}\Delta}-1\right).\label{eq-AppF-1}
\end{align}
By plugging the value of $\zeta = \sqrt{\frac{\lambda_{B}}{\alpha_B^2\Delta}\log \frac{1}{\mathcal{E}}}$ into \eqref{eq-AppF-1}, we obtain
\begin{align}
C_S \geq\,& \left[\alpha_B\log\left(1 + \frac{\alpha_B}{\lambda_B}\sqrt{\frac{\lambda_{B}}{\alpha_B^2\Delta}\log \frac{1}{\mathcal{E}}}\right) - \alpha_E\log\left(1 + \frac{\alpha_E}{\lambda_E}\sqrt{\frac{\lambda_{B}}{\alpha_B^2\Delta}\log \frac{1}{\mathcal{E}}}\right)\right]\mathcal{E}-\alpha_B\mathcal{E} \notag\\
&+\alpha_E\mathcal{E}^2\frac{\log\left(1+\frac{\alpha_E}{\lambda_E}\sqrt{\frac{\lambda_{B}}{\alpha_B^2\Delta}\log \frac{1}{\mathcal{E}}}\right)}{\sqrt{\frac{\lambda_{B}}{\alpha_B^2\Delta}\log \frac{1}{\mathcal{E}}}}+\lambda_B\mathcal{E}\frac{\log\left(1+\frac{\alpha_E}{\lambda_E}\sqrt{\frac{\lambda_{B}}{\alpha_B^2\Delta}\log \frac{1}{\mathcal{E}}}\right)}{\sqrt{\frac{\lambda_{B}}{\alpha_B^2\Delta}\log \frac{1}{\mathcal{E}}}}\notag\\
&-\frac{\mathcal{E}}{\Delta\sqrt{\frac{\lambda_{B}}{\alpha_B^2\Delta}\log \frac{1}{\mathcal{E}}}}-\frac{\mathcal{E}^2}{\Delta\frac{\lambda_{B}}{\alpha_B^2\Delta}\log \frac{1}{\mathcal{E}}}\left(\frac{1}{\mathcal{E}}-1\right). \label{eq-AppF-2} 	
\end{align}
Now, from \eqref{eq-AppF-2}, we can write
\begin{align}
	&\lim_{\mathcal{E}\rightarrow 0} \frac{C_S}{\mathcal{E}\log\log\frac{1}{\mathcal{E}}}\notag\\
	\geq\,& \lim_{\mathcal{E}\rightarrow 0}\frac{\left[\alpha_B\log\left(1 + \frac{\alpha_B}{\lambda_B}\sqrt{\frac{\lambda_{B}}{\alpha_B^2\Delta}\log \frac{1}{\mathcal{E}}}\right) - \alpha_E\log\left(1 + \frac{\alpha_E}{\lambda_E}\sqrt{\frac{\lambda_{B}}{\alpha_B^2\Delta}\log \frac{1}{\mathcal{E}}}\right)\right]\mathcal{E}-\alpha_B\mathcal{E}}{\mathcal{E}\log\log\frac{1}{\mathcal{E}}}\notag\\
	&+\lim_{\mathcal{E}\rightarrow 0}\frac{\alpha_E\mathcal{E}^2\frac{\log\left(1+\frac{\alpha_E}{\lambda_E}\sqrt{\frac{\lambda_{B}}{\alpha_B^2\Delta}\log \frac{1}{\mathcal{E}}}\right)}{\sqrt{\frac{\lambda_{B}}{\alpha_B^2\Delta}\log \frac{1}{\mathcal{E}}}}}{\mathcal{E}\log\log\frac{1}{\mathcal{E}}}+\lim_{\mathcal{E}\rightarrow 0}\frac{\lambda_B\mathcal{E}\frac{\log\left(1+\frac{\alpha_E}{\lambda_E}\sqrt{\frac{\lambda_{B}}{\alpha_B^2\Delta}\log \frac{1}{\mathcal{E}}}\right)}{\sqrt{\frac{\lambda_{B}}{\alpha_B^2\Delta}\log \frac{1}{\mathcal{E}}}}}{\mathcal{E}\log\log\frac{1}{\mathcal{E}}}\notag\\
	&-\lim_{\mathcal{E}\rightarrow 0}\frac{\frac{\mathcal{E}}{\Delta\sqrt{\frac{\lambda_{B}}{\alpha_B^2\Delta}\log \frac{1}{\mathcal{E}}}}}{\mathcal{E}\log\log\frac{1}{\mathcal{E}}}-\lim_{\mathcal{E}\rightarrow 0}\frac{\frac{\mathcal{E}^2}{\Delta\frac{\lambda_{B}}{\alpha_B^2\Delta}\log \frac{1}{\mathcal{E}}}\left(\frac{1}{\mathcal{E}}-1\right)}{\mathcal{E}\log\log\frac{1}{\mathcal{E}}}\notag\\
	=\,& \lim_{\mathcal{E}\rightarrow 0}\frac{\frac{(\alpha_B-\alpha_E)}{2}\left(\log\left(\frac{1}{\Delta}\log\frac{1}{\mathcal{E}}\right)\right)}{\log\log\frac{1}{\mathcal{E}}}-0+0+0-0-0\notag\\
	=\,& \frac{(\alpha_B-\alpha_E)}{2}.\label{eq-AppF-2-1}
\end{align} 
This completes the analysis of the lower bound for the asymptotic secrecy capacity. Next, we provide an upper bound for the asymptotic secrecy capacity.

\subsection{Upper Bound} 
To find an upper bound on the secrecy capacity we start by noting that due to Lemma~\ref{lemma-6}, we have that
\begin{align}
C_S &\leq \sup_{F_X\in\Omega^{+}_{\mathcal{E}}} I(X;Y) - I(X;\widetilde{Y})+I(X;\widetilde{Z})\notag\\
&\leq \underbrace{\sup_{F_X\in\Omega^{+}_{\mathcal{E}}} I(X;Y)-I(X;\widetilde{Y})}_{\stackrel{\triangle}{=}\,C_{S,1}}+\underbrace{\sup_{F_X\in\Omega^{+}_{\mathcal{E}}} I(X;\widetilde{Z})}_{\stackrel{\triangle}{=}\,C_1},
\end{align}
where $\widetilde{Y}\vert X$ is a Poisson random variable with mean $(\alpha_BX+\frac{\alpha_B}{\alpha_E}\lambda_E)$, and $\widetilde{Z}\vert X$ is a Poisson random variable with mean $(\widetilde{\alpha}x + \widetilde{\lambda})\Delta$, where $\widetilde{\alpha} \stackrel{\triangle}{=} \alpha_B - \alpha_E$, $\widetilde{\lambda} \stackrel{\triangle}{=} \left(\frac{\alpha_B}{\alpha_E}-1\right)\lambda_E$. Observe that $C_{S,1}$ is the secrecy capacity of a degraded DT--PWC whose input is $X$ such that $X\geq 0$ and $\mathbb{E}[X]\leq \mathcal{E}$ and whose outputs are $Y$ and $\widetilde{Y}$. Furthermore, notice that $C_1$ is the channel capacity of a discrete-time Poisson channel whose input is $X$ subject to nonnegativity and average-intensity constraint and whose output is $\widetilde{Z}$. Consequently, from Theorem~\ref{theo-avg-zero}, we know that 
\begin{equation}
\lim_{\mathcal{E}\rightarrow 0}\frac{C_{S,1}}{\mathcal{E}} = \alpha_B\log\left(\frac{\lambda_E\alpha_B}{\lambda_B\alpha_E}\right).\label{eq-AppF-3}
\end{equation}
Furthermore, since $\alpha_B>\alpha_E$, we have that $\widetilde{\lambda} > 0$. Thus, we can invoke the asymptotic channel capacity results by Lapidoth \textit{et al.} in~\cite[Proposition~2]{5773060}. From~\cite[Proposition~2]{5773060}, we have that in the absence of the peak-intensity constraint and with nonzero constant dark current along with an average-intensity constraint, the asymptotic channel capacity satisfies
\begin{equation}
\lim_{\mathcal{E}\rightarrow 0}\frac{C_1}{\mathcal{E}\log\log\frac{1}{\mathcal{E}}}\leq 2(\alpha_B-\alpha_E).\label{eq-AppF-4}
\end{equation}
Finally, from \eqref{eq-AppF-3}--\eqref{eq-AppF-4}, we can conclude that the asymptotic secrecy capacity satisfies
\begin{align}
\lim_{\mathcal{E}\rightarrow 0}\frac{C}{\mathcal{E}\log\log\frac{1}{\mathcal{E}}}&\leq \lim_{\mathcal{E}\rightarrow 0}\frac{\alpha_B\log\left(\frac{\lambda_E\alpha_B}{\lambda_B\alpha_E}\right)}{\log\log\frac{1}{\mathcal{E}}} + 2(\alpha_B-\alpha_E)\notag\\
&= 2(\alpha_B-\alpha_E).\label{eq-AppF-5}
\end{align}
By combining \eqref{eq-AppF-2-1} and \eqref{eq-AppF-5}, we find the lower and upper bounds on the asymptotic secrecy capacity. This completes the proof of the theorem.

\section{Upper Bound on the Secrecy Capacity in the High-Intensity Regime For Equal Channel Gains}\label{App-E}
	We start the proof by noting that the output of the eavesdropper's channel $Z$ can be written as $Z = \widetilde{Y} = Y + N_{D}$, where $N_D$ is defined in the statement of Lemma~\ref{lemma-6}. Therefore, $H_Z(F_X^*) > H_{Z\vert N_D}(F_X^*) = H_Y(F_X^*)$, and consequently $H_Y(F_X^*) - H_Z(F_X^*) < 0$ for any nontrivial input distribution $F_X^*\in\mathcal{F}^{+}$. Furthermore, we can expand $H_{Z\vert X}(F_X^*) - H_{Y\vert X}(F_X^*)$ as follows
	\begin{align}
		H_{Z\vert X}(F_X^*) - H_{Y\vert X}(F_X^*) &= \frac{1}{\Delta}\mathbb{E}_{X,Z}\left[-\log p_{Z\vert X}(z\vert x)\right] - \frac{1}{\Delta}\mathbb{E}_{X,Y}\left[-\log p_{Y\vert X}(y\vert x)\right] \notag
		\\ &\stackrel{(a)}{=} \frac{1}{\Delta} \mathbb{E}_{Z\vert X,Y}\left[\mathbb{E}_{X,Y}\left[\log p_{Y\vert X}(y\vert x)\right]\right] - \frac{1}{\Delta} \mathbb{E}_{Y\vert X,Z}\left[\mathbb{E}_{X,Z}\left[\log p_{Z\vert X}(z\vert x)\right]\right] \notag\\
		&=\frac{1}{\Delta}\mathbb{E}_{X,Y,Z}\left[\log p_{Y\vert X}(y\vert x)\right] -\frac{1}{\Delta} \mathbb{E}_{X,Y,Z}\left[\log p_{Z\vert X}(z\vert x)\right] \notag \\
		&=\frac{1}{\Delta} \mathbb{E}_{X,Y,Z}\left[\log\frac{p_{Y\vert X}(y\vert x)}{p_{Z\vert X}(z\vert x)} \right],\label{eq-high}
	\end{align}
where $(a)$ follows as $\log p_{Y\vert X}(y\vert x)$ and $\log p_{Z\vert X}(z\vert x)$ do no depend on $Z$ and $Y$, respectively. Plugging \eqref{eq-chan-B} and \eqref{eq-chan-E} into \eqref{eq-high}, we get
\begin{align}
	H_{Z\vert X}(F_X^*) - H_{Y\vert X}(F_X^*) =&\,\frac{1}{\Delta} \mathbb{E}_{X,Y,Z}\left[\log\frac{e^{-(\alpha_Bx+\lambda_B)\Delta}[(\alpha_Bx+\lambda_B)\Delta]^{\,y}/y!}{e^{-(\alpha_Ex+\lambda_E)\Delta}[(\alpha_Ex+\lambda_E)\Delta]^{\,z}/z!} \right] \notag \\
	=&\, \lambda_D + \mathbb{E}_{X}\big[(\alpha_Bx+\lambda_B)\log[(\alpha_Bx+\lambda_{B})\Delta]- (\alpha_Ex+\lambda_{E})\notag\\&\times\log[(\alpha_Ex+\lambda_{E})\Delta]\big]  + \frac{1}{\Delta}\mathbb{E}_{X,Y,Z}\left[\log\frac{Z!}{Y!}\right],\label{eq-upper}
\end{align}
Next, we consider the last term in~\eqref{eq-upper} and try to find an upper bound on it. To this end, we first note that 
\begin{equation}
\mathbb{E}_{X,Y,Z}\left[\log\frac{Z!}{Y!}\right] = \mathbb{E}_X\left[\mathbb{E}_{Y\vert X}\left[\mathbb{E}_{Z\vert Y}\left[\log\frac{Z!}{Y!}\right]\right]\right]
\end{equation}
as $X\rightarrow Y\rightarrow Z$ is a Markov chain. Now, we have to find the conditional PDF of $Z\vert Y$. We proceed by observing that $Z = Y + N_D$, hence, one can show that 
\begin{equation}\label{eq-PDF-yz}
p_{Z\vert Y} (z\vert y)= 
	\begin{cases}
		0,~ &\text{if} ~ z < y, \\
		e^{-\lambda_{D}\Delta} \frac{(\lambda_{D}\Delta)^{(z-y)}}{(z-y)!}, ~ &\text{if} ~ z \geq y.
	\end{cases}
\end{equation}
In what follows, we present chain of inequalities based on \eqref{eq-PDF-yz} which leads to the upper bound in \eqref{eq-SecCap-Upp},
\begin{align}
	\mathbb{E}_{X,Y,Z}\left[\log\frac{Z!}{Y!}\right] &= \mathbb{E}_X\left[\mathbb{E}_{Y\vert X}\left[\sum_{z=0}^{+\infty}p_{Z\vert Y}(z\vert y)\log\frac{z!}{y!}\right]\right] \notag\\
	&= \mathbb{E}_X\left[\mathbb{E}_{Y\vert X}\left[\sum_{z=y}^{+\infty}e^{-\lambda_{D}\Delta}\frac{(\lambda_{D}\Delta)^{(z-y)}}{(z-y)!}\log\frac{z!}{y!}\right]\right] \notag\\
	&=\mathbb{E}_X\left[\mathbb{E}_{Y\vert X}\left[ \sum_{t=0}^{+\infty}e^{-\lambda_{D}\Delta}\frac{(\lambda_{D}\Delta)^t}{t!}\log\frac{(t+y)!}{y!}\right]\right]\notag\\
	&= \mathbb{E}_X\left[\mathbb{E}_{Y\vert X}\left[\sum_{t=0}^{+\infty}e^{-\lambda_{D}\Delta}\frac{(\lambda_{D}\Delta)^t}{t!}\sum_{i=1}^{t}\log(y+i)\right]\right] \notag \\
	&\stackrel{(b)}{\leq}\mathbb{E}_{X}\left[ \sum_{t=0}^{+\infty}e^{-\lambda_{D}\Delta}\frac{(\lambda_{D}\Delta)^t}{t!}\sum_{i=1}^{t}\log[(\alpha_B x+\lambda_{B})\Delta+i]\right] \notag \\
	&=\mathbb{E}_{X}\Big[ \sum_{t=0}^{+\infty}e^{-\lambda_D\Delta}\frac{(\lambda_D\Delta)^t}{t!}\Big[t\log[(\alpha_Bx+\lambda_{B})\Delta]\notag\\
	&\quad\qquad+\sum_{i=1}^{t}\log\left[1+\frac{i}{(\alpha_Bx+\lambda_{B})\Delta}\right]\Big]\Big] \notag\\
	&\stackrel{(c)}{\leq}\mathbb{E}_{X}\left[ \sum_{t=0}^{+\infty}e^{-\lambda_D\Delta}\frac{(\lambda_D\Delta)^t}{t!}\left[t\log[(\alpha_Bx+\lambda_{B})\Delta]+\sum_{i=1}^{t}\frac{i}{(\alpha_Bx+\lambda_{B})\Delta}\right]\right] \notag\\
	&=\mathbb{E}_{X}\Big[\log[(\alpha_Bx+\lambda_{B})\Delta]\sum_{t=0}^{+\infty}e^{-\lambda_D\Delta}\frac{(\lambda_D\Delta)^t}{t!}t + \frac{1}{(\alpha_Bx+\lambda_{B})\Delta}\notag\\
	&\quad\qquad\times\sum_{t=0}^{+\infty}e^{-\lambda_D\Delta}\frac{(\lambda_D\Delta)^t}{t!}\frac{t(t+1)}{2}\Big]\notag\\
	&= \mathbb{E}_{X}\left[(\lambda_D\Delta)\log[(\alpha_Bx+\lambda_{B})\Delta] + \frac{1}{(\alpha_Bx+\lambda_{B})\Delta}\left[\frac{(\lambda_D\Delta)^2}{2}+\lambda_D\Delta\right]\right],
	\label{eq-upper-final}
\end{align} 
where $(b)$ follows from sliding the expectation $\mathbb{E}_{Y\vert X}$ through the summations and then applying the Jensen's Inequality (as $\log(y+i)$ is a concave function in $y$), and $(c)$ follows from the fact that $\log(1+x) \leq x,~\forall x\geq 0$. Now, using the upper bound in \eqref{eq-upper-final}, $H_{Z\vert X}(F_X^*) - H_{Y\vert X}(F_X^*)$ can be upper bounded as
\begin{align}
	H_{Z\vert X}(F_X^*) - H_{Y\vert X}(F_X^*) &\leq 
	\lambda_D + \mathbb{E}_{X}\big[(\alpha_Bx+\lambda_B)\log[(\alpha_Bx+\lambda_{B})\Delta]- (\alpha_Ex+\lambda_{E})\notag\\&\quad
	\times\log[(\alpha_Ex+\lambda_{E})\Delta]\big] + \mathbb{E}_{X}\left[\lambda_D\log[(\alpha_Bx+\lambda_{B})\Delta]\right]\notag\\
	 &\quad+ \mathbb{E}_{X}\left[\frac{1}{(\alpha_Bx+\lambda_{B})\Delta^2}\left[\frac{(\lambda_D\Delta)^2}{2}+\lambda_D\Delta\right]\right]\notag\\
	&= \lambda_{D} + \mathbb{E}_{X}\left[(\alpha_Ex+\lambda_{E})\log\frac{\alpha_Bx+\lambda_{B}}{\alpha_Ex+\lambda_{E}}\right] \notag\\
	&\quad +\left[\frac{\lambda_D^2}{2}+\frac{\lambda_D}{\Delta}\right]\mathbb{E}_{X}\left[\frac{1}{\alpha_Bx+\lambda_{B}}\right].\label{eq-Upp-final}
\end{align}
Now, we note that since $x \geq 0$, $\mathbb{E}_X\left[\frac{1}{\alpha_Bx+\lambda_{B}}\right] \leq \frac{1}{\lambda_{B}}$. Furthermore, denoting $\psi(x) \stackrel{\triangle}{=} (\alpha_Ex+\lambda_{E})\log\frac{\alpha_Bx+\lambda_{B}}{\alpha_Ex+\lambda_{E}}$, we observe that $\psi(x)$ is strictly negative when $\alpha_B = \alpha_E$ and $\frac{\lambda_{E}}{\alpha_E} > \frac{\lambda_B}{\alpha_B}$. Furthermore, $\psi(x)$ is a strictly increasing function in $x$ due to the fact that 
\begin{equation}
\frac{d\psi(x)}{dx} = \alpha_B\left[ -\log\left[1+\frac{\lambda_{D}}{\alpha_Bx+\lambda_{B}}\right]+\frac{\lambda_{D}}{\alpha_Bx+\lambda_{B}}\right] > 0,\quad \forall ~x\geq 0.
\end{equation}
This implies that the maximum value of $\psi(x)$ is located at the end point of the interval $[0,\mathcal{A}]$, if the peak-intensity is active, and is located at $x = +\infty$, if the average-intensity is the only active constraint. In either of these cases, we can write
\begin{equation}
\psi(x) \leq \lim_{x\rightarrow +\infty}(\alpha_Ex+\lambda_E)\log\frac{\alpha_B x+\lambda_B}{\alpha_Ex+\lambda_E} = -\lambda_{D}.
\label{eq-finalValue}
\end{equation}
From the upper bound on $\mathbb{E}_X\left[\frac{1}{\alpha_Bx+\lambda_{B}}\right]$ and \eqref{eq-finalValue}, one can upper bound \eqref{eq-upper-final} as
\begin{equation}
H_{Z\vert X}(F_X^*) - H_{Y\vert X}(F_X^*) \leq \frac{\frac{\lambda_D^2}{2}+\frac{\lambda_D}{\Delta}}{\lambda_{B}}.
\end{equation}
We note that this \textit{constant} upper bound is valid for all values of the peak- and/or average-intensity constraints. This completes the proof of the proposition.

\section{Upper Bound on the Secrecy Capacity in the High-Intensity Regime For Different Channel Gains}\label{App-H}
We start the proof by making the following important observation for the conditional mutual information $I(X;\widetilde{Y}\vert Z)$~\cite{4729780,6121996,KornerBook,8421280}
\begin{equation}
	I(X;\widetilde{Y}\vert Z) + \frac{1}{\Delta}\mathbb{E}_{X,Z}\left[D\left(p_{\widetilde{Y}\vert Z}(\widetilde{y}\vert z)\parallel q_{\widetilde{Y}\vert Z}(\widetilde{y}\vert z)\right)\right] = \frac{1}{\Delta} \mathbb{E}_{X,Z}\left[D\left( p_{\widetilde{Y}\vert X,Z}(\widetilde{y}\vert x,z)\parallel q_{\widetilde{Y}\vert Z}(\widetilde{y}\vert z)\right)\right],
\end{equation}
where $q_{\widetilde{Y}\vert Z}(\widetilde{y}\vert z)$ is an arbitrarily conditional probability mass function. Since the relative entropy is nonnegative, we have the following upper bound on $I(X;\widetilde{Y}\vert Z)$ as
\begin{equation}
I(X;\widetilde{Y}\vert Z)\leq \frac{1}{\Delta} \mathbb{E}_{X,Z}\left[D\left( p_{\widetilde{Y}\vert X,Z}(\widetilde{y}\vert x,z)\parallel q_{\widetilde{Y}\vert Z}(\widetilde{y}\vert z)\right)\right].\label{eq-App-H-0}
\end{equation} 
Observe that the inequality \eqref{eq-App-H-0} holds for all the admissible input distributions and any arbitrary conditional probability mass function $q_{\widetilde{Y}\vert Z}(\widetilde{y}\vert z)$. Hence, we can upper bound the secrecy capacity as follows
\begin{equation}
C_S = I(X^*;\widetilde{Y}\vert Z)\leq \frac{1}{\Delta}\mathbb{E}_{X^*,Z}\left[D\left( p_{\widetilde{Y}\vert X^*,Z}(\widetilde{y}\vert x^*,z)\parallel q_{\widetilde{Y}\vert Z}(\widetilde{y}\vert z)\right)\right],\label{eq-App-H-1}
\end{equation}
where $X^*$ is the input random variable distributed according to the secrecy-capacity-achieving distribution. Therefore, the problem of finding a constant upper bound on the secrecy capacity which does not scale with the constraints boils down to finding a clever choice for $q_{\widetilde{Y}\vert Z}(\widetilde{y}\vert z)$. 

Next, we will expand the RHS of \eqref{eq-App-H-1} as follows
\begin{align}
&\mathbb{E}_{X^*,Z}\left[D\left( p_{\widetilde{Y}\vert X^*,Z}(\widetilde{y}\vert x^*,z)\parallel q_{\widetilde{Y}\vert Z}(\widetilde{y}\vert z)\right)\right]\notag\\
=& \sum_{x^* \in\mathcal{S}_{F_X^*} }\sum_{\widetilde{y}=0}^{+\infty}\sum_{z=0}^{+\infty}p_{X^*,\widetilde{Y},Z}(x^*,\widetilde{y},z)\log\left(\frac{p_{\widetilde{Y}\vert X^*,Z}(\widetilde{y}\vert x^*,z)}{q_{\widetilde{Y}\vert Z}(\widetilde{y}\vert z)}\right)\notag\\
=&\sum_{x^* \in\mathcal{S}_{F_X^*} }\sum_{\widetilde{y}=0}^{+\infty}\sum_{z=0}^{+\infty}p_{X^*,\widetilde{Y},Z}(x^*,\widetilde{y},z)\log\left(\frac{p_{Z\vert  X^*,\widetilde{Y}}(z\vert x^*,\widetilde{y})p_{\widetilde{Y},X^*}(\widetilde{y},x^*)}{p_{X^*,Z}(x^*,z)q_{\widetilde{Y}\vert Z}(\widetilde{y}\vert z)}\right)\notag\\
\stackrel{(a)}{=}&\sum_{x^* \in\mathcal{S}_{F_X^*} }\sum_{\widetilde{y}=0}^{+\infty}\sum_{z=0}^{+\infty}p_{X^*,\widetilde{Y},Z}(x^*,\widetilde{y},z)\log\left(\frac{p_{Z\vert\widetilde{Y}}(z\vert\widetilde{y})p_{\widetilde{Y}\vert X^*}(\widetilde{y}\vert x^*)}{p_{Z\vert X^*}(z\vert x^*)q_{\widetilde{Y}\vert Z}(\widetilde{y}\vert z)}\right)\notag\\
=& \, \Delta \left[H(Z\vert X^*) - H(\widetilde{Y}\vert X^*)\right] + \mathbb{E}_{X^*,\widetilde{Y},Z}\left[\log\left(\frac{p_{Z\vert \widetilde{Y}}(z\vert \widetilde{y})}{q_{\widetilde{Y}\vert Z}(\widetilde{y}\vert z)}\right)\right]
,\label{eq-App-H-2}
\end{align}
where $\mathcal{S}_{F_X^*}$ is the support set of the secrecy-capacity-achieving input distribution, and $(a)$ follows because we have the Markov chain $X\rightarrow \widetilde{Y}\rightarrow Z$. Now, we will further upper bound \eqref{eq-App-H-2}. Towards achieving this goal, we note that based on Lemma~\ref{lemma-6}, we have that $H(\widetilde{Y}\vert X^*)\geq H(\widetilde{Y}\vert X^*,\widetilde{Z}) = H(Z\vert X^*)$. As such, $H(Z\vert X^*) - H(\widetilde{Y}\vert X^*)\leq 0$. Next, we have to identify $p_{Z\vert \widetilde{Y}}(z\vert\widetilde{y})$ and find a clever choice for $q_{\widetilde{Y}\vert Z}(\widetilde{y}\vert z)$. Observe that since $Z$ is obtained by thinning $\widetilde{Y}$ with erasure probability $1-\frac{\alpha_E}{\alpha_B}$, we can write~\cite{5550280}
\begin{equation}
	p_{Z\vert\widetilde{Y}}(z\vert\widetilde{y}) = 
	\begin{cases}
		0, \quad&\text{if} ~z>\widetilde{y} \\
		\dbinom{\widetilde{y}}{z}\left(\frac{\alpha_E}{\alpha_B}\right)^z\left(1-\frac{\alpha_E}{\alpha_B}\right)^{\widetilde{y}-z},\quad&\text{if}~z\leq \widetilde{y} 
	\end{cases}\label{eq-App-H-3}
\end{equation}
The conditional probability mass function $p_{\widetilde{Y}\vert Z}(\widetilde{y}\vert z)$ hints us towards choosing a clever $q_{Z\vert\widetilde{Y}}(z\vert \widetilde{y})$. Amongst all the possible conditional probability mass functions and in light of the nature of $p_{\widetilde{Y}\vert Z}(\widetilde{y}\vert z)$, if we choose $q_{\widetilde{Y}\vert Z}(\widetilde{y}\vert z)$ to be a \textit{negative Binomial} distribution with $\widetilde{y}-z\geq 0$ failures and $z+1$ successes with success probability $\frac{\alpha_E}{\alpha_B}$, we then can have
\begin{equation}
	q_{\widetilde{Y}\vert Z}(\widetilde{y}\vert z)=
	\begin{cases}
		0,\quad&\text{if}~\widetilde{y}<z \\
		\dbinom{\widetilde{y}}{z}\left(\frac{\alpha_E}{\alpha_B}\right)^{z+1}\left(1-\frac{\alpha_E}{\alpha_B}\right)^{\widetilde{y}-z},\quad&\text{if}~\widetilde{y}\geq z
	\end{cases}\label{eq-App-H-4}
\end{equation}   
By substituting \eqref{eq-App-H-3}--\eqref{eq-App-H-4} into \eqref{eq-App-H-2} and noting that $H(Z\vert X^*) - H(\widetilde{Y}\vert X^*)\leq 0$, we find that the secrecy capacity can be upper bounded as follows
\begin{equation}
C_S \leq \frac{1}{\Delta}\log\left(\frac{\alpha_B}{\alpha_E}\right).
\end{equation}
Therefore, the secrecy capacity is upper bounded by a constant value across all intensity regimes. This completes the proof of the theorem.

\section{Structure of the Optimal Input Distributions When $\lambda_B=0$}\label{App-last}
Without loss of generality, we will provide the proof for the structure of the optimal input distribution which attains the secrecy capacity of the DT--PWC with nonnegativity and average-intensity constraints, i.e., $\mu = 0$ in \eqref{eq-rateequivKKT-avg}--\eqref{eq-rateequivKKT-avg1}. For $\mu \in [0,1]$ along with the existence of both peak- and average-intensity constraints, the proof follows along similar lines as below. 

We establish that for $\lambda_B=0$, the optimal input distribution $F_{X}^{*}$ has the following structural properties: 1) the intersection of the optimal support set with any bounded interval contains finitely many mass points; 2) The optimal support set itself is an unbounded set.
\begin{enumerate}
	\item The intersection of the optimal support set with any bounded interval contains a finite number of elements:
\end{enumerate}

Let $B$ be a bounded interval and assume, to the contrary, that $\mathcal{S}_{F_X^*} \cap B$ has an infinite number of elements. Now based on the optimality equation \eqref{eq-rateequivKKT-avg1}, the analyticity of $c_S(x;F_X^*)$ over $\mathcal{O}$ and the Identity Theorem from complex analysis, if $\mathcal{S}_{F_X^*} \cap B$ has an accumulation point in $\mathcal{O}$, then \eqref{eq-rateequivKKT-avg1} applies everywhere in $\mathcal{O}$. That $\mathcal{S}_{F_X^*} \cap B$ has an accumulation point is guaranteed by the Bolzano-Weierstrass Theorem since $\mathcal{S}_{F_X^*} \cap B \subseteq B$ and $B$ is bounded. However, the accumulation point might be equal to 0 and $0 \notin \mathcal{O}$. Next, we show that 0 cannot be an accumulation point of $\mathcal{S}_{F_X^*} \cap B $ so that (88) actually applies over $\mathcal{O}$ and in particular over $(0,+\infty)$. Assume to the contrary that 0 is an accumulation point of $\mathcal{S}_{F_X^*} \cap B $. Then, there exists a sequence $(x)_{i} $ defined on $\mathcal{S}_{F_X^*} \cap B $ such that $x_i \neq 0$ and $\underset{i \to +\infty}{\lim}x_i=0$.  

For convenience, let $\alpha_{B}^{\Delta}= \alpha_{B} \Delta$, $\alpha_{E}^{\Delta}=\alpha_{E} \Delta$, $\lambda_{B}^{\Delta}= \lambda_{B} \Delta$ and $\lambda_{E}^{\Delta}= \lambda_{E} \Delta$.  Then, by expanding $p_{Y\vert X}(y\vert x)$ around 0 we find that
\begin{align}\label{E1}
p_{Y\vert X}(y\vert x) = \frac{ e^{-\lambda_{B}^{\Delta}}(\lambda_{B}^{\Delta})^{y}}{y!} + 
\frac{ \alpha_{B}^{\Delta} e^{-\lambda_{B}^{\Delta}} (\lambda_{B}^{\Delta})^{y-1}  (y-\lambda_{B}^{\Delta})}   {y!} x + 
\text{o}(x),
\end{align}
implying that
\begin{align}\label{E2}
c_S(x;F_{X}^{*}) =\,& (\alpha_B x+\lambda_B) \log(\alpha_{B}^{\Delta} x+\lambda_{B}^{\Delta}) -  (\alpha_E x+\lambda_E)\log(\alpha_{E}^{\Delta} x+\lambda_{E}^{\Delta})  \notag \\
& - \frac{1}{\Delta} \sum_{y=0}^{+\infty}  \frac{e^{-\lambda_{B}^{\Delta}}(\lambda_{B}^{\Delta})^{y}}{y!} \log( g_B(y;F_{X}^{*})) +  \frac{1}{\Delta} \sum_{z=0}^{+\infty}  \frac{e^{-\lambda_{E}^{\Delta}}(\lambda_{E}^{\Delta})^{z}}{z!} \log( g_E(z;F_{X}^{*})) \notag \\
&  + \kappa\left(\alpha_B,\alpha_E,\lambda_B,\lambda_E, \Delta\right) x +\text{o}(x),
\end{align}
where 
\begin{align}\label{E3}
\kappa\left(\alpha_B,\alpha_E,\lambda_B,\lambda_E, \Delta \right) =&  - \frac{1}{\Delta} \sum_{y=0}^{+\infty} \frac{ \alpha_{B}^{\Delta} e^{-\lambda_{B}^{\Delta}} (\lambda_{B}^{\Delta})^{y-1}  (y-\lambda_{B}^{\Delta})}   {y!}  \log( g_B(y;F_{X}^{*}))  -\alpha_B \notag \\
& +  \frac{1}{\Delta}  \sum_{y=0}^{+\infty} \frac{ \alpha_{E}^{\Delta} e^{-\lambda_{E}^{\Delta}} (\lambda_{E}^{\Delta})^{z-1}  (z-\lambda_{E}^{\Delta})}   {z!} \log( g_E(z;F_{X}^{*})) + \alpha_E,
\end{align}
In particular, we can write
\begin{align}\label{E4}
c_S(0;F_{X}^{*}) =& \lambda_B \log(\lambda_{B}^{\Delta}) -  \lambda_E \log(\lambda_{E}^{\Delta})  \notag \\
& - \frac{1}{\Delta} \sum_{y=0}^{+\infty}  \frac{e^{-\lambda_{B}^{\Delta}}(\lambda_{B}^{\Delta})^{y}}{y!} \log( g_B(y;F_{X}^{*})) +  \frac{1}{\Delta} \sum_{z=0}^{+\infty}  \frac{e^{-\lambda_{E}^{\Delta}}(\lambda_{E}^{\Delta})^{z}}{z!} \log( g_E(z;F_{X}^{*})). 
\end{align}
Substituting \eqref{E4} in \eqref{E2} yields
\begin{align}\label{E5}
c_S(x;F_{X}^{*}) =\,& c_S(0;F_{X}^{*}) +  (\alpha_B x+\lambda_B)   \log ( \alpha_{B}^{\Delta}x+\lambda_{B}^{\Delta} ) -  (\alpha_E x+\lambda_E) \log (\alpha_{E}^{\Delta}x+\lambda_{E}^{\Delta})  \notag \\
& - \lambda_B \log(\lambda_{B}^{\Delta}) + \lambda_E \log(\lambda_{E}^{\Delta})   + \kappa\left(\alpha_B,\alpha_E,\lambda_B,\lambda_E, \Delta \right) x +\text{o}(x),
\end{align}
Considering the KKT condition \eqref{eq-rateequivKKT-avg1} for $\mu = 0$ and the fact that $(x)_i \in \mathcal{S}_{F_X^*} \cap B $, one obtains
\begin{align}\label{E6}
\frac{C_S -\gamma \mathcal{E} -c_s(0;F_{X}^{*})}{x_i} =\,& \frac{ (\alpha_B x_i+\lambda_B)\log(\alpha_{B}^{\Delta} x_i+\lambda_{B}^{\Delta})}{x_i} -  \frac{(\alpha_E x_i+\lambda_E) \log(    \alpha_{E}^{\Delta}    x_i    +   \lambda_{E}^{\Delta}   )   }{x_i} \notag \\
& - \frac{\lambda_B \log ( \lambda_{B}^{ \Delta}  )} {x_i} + \frac{\lambda_E \log ( \lambda_{E}^{\Delta } )} {x_i}   + \kappa\left(\alpha_B,\alpha_E,\lambda_B,\lambda_E, \Delta \right)   -\gamma +\text{o}(1).
\end{align} 
Also, in regard of \eqref{eq-rateequivKKT-avg} for $\mu=0$, $C_S -\gamma \mathcal{E} -c_s (0;F_{X}^{*}) \geq 0$. Hence,
\begin{align}\label{E7}
\lim_{i \to \infty} \frac{C_S -\gamma \mathcal{E} -c_S(0;F_{X}^{*})}{x_i} =
\begin{cases}
+ \infty, & \text{if}                      \quad            C_S -\gamma \mathcal{E} -c_S(0;F_{X}^{*}) > 0 \\
0,          & \text{otherwise}        
\end{cases}
\end{align}
Let us compute the limit of the right hand side (RHS) of \eqref{E6} as $i \to \infty$. For this purpose, we distinguish two cases.
\begin{itemize}
	\item  Case 1: $\lambda_B=0$ and $\lambda_E >0$ 
\end{itemize}
In this case, the RHS of \eqref{E6} becomes:
\begin{align}
\text{RHS of \eqref{E6}}  = & \,\alpha_B \log(\alpha_{B}^{\Delta} x_i)  -  \alpha_E \log(    \alpha_{E}^{\Delta}    x_i   +   \lambda_{E}^{\Delta}   )  -   \frac{\lambda_E}{x_i} \log(    \alpha_{E}^{\Delta}    x_i    +   \lambda_{E}^{\Delta}   )  \notag \\
& + \frac{\lambda_E}{x_i} \log ( \lambda_{E}^{\Delta})  + \kappa\left(\alpha_B,\alpha_E,\lambda_B,\lambda_E, \Delta \right)   -\gamma +\text{o}(1) \\
= &\,   \alpha_B \log(\alpha_{B}^{\Delta} x_i) -  \alpha_E \log(    \alpha_{E}^{\Delta}    x_i   +   \lambda_{E}^{\Delta}   ) -   \frac{\lambda_E} {x_i}   \log\left(    1+ \frac{\alpha_{E} }{\lambda_E}   x_i  \right) \notag \\
& + \kappa\left(\alpha_B,\alpha_E,\lambda_B,\lambda_E, \Delta \right)   -\gamma +\text{o}(1)  \\
= & \, \alpha_B \log(\alpha_{B}^{\Delta} x_i) -  \alpha_E \log(    \alpha_{E}^{\Delta}    x_i   +   \lambda_{E}^{\Delta}   ) - \alpha_E + \kappa\left(\alpha_B,\alpha_E,\lambda_B,\lambda_E, \Delta \right)   -\gamma +\text{o} (1), \label{E61}  
\end{align} 
where \eqref{E61} follows since $\log(1+x)=x+\text{o}(x)$. Hence, $\underset{i \to + \infty}{\lim} \text{RHS of \eqref{E6}}=- \infty$, thus reaching a contradiction.
\begin{itemize}
	\item Case 2: $\lambda_B=0$ and $\lambda_E =0$ 
\end{itemize}
In this case,  the RHS of \eqref{E6} becomes:
\begin{align}\label{E71}
\text{RHS of \eqref{E6}}  = &\,   \alpha_B \log(\alpha_{B}^{\Delta} x_i)  -  \alpha_E \log(    \alpha_{E}^{\Delta}    x_i        )  + \kappa\left(\alpha_B,\alpha_E,\lambda_B,\lambda_E, \Delta \right)   -\gamma +\text{o}(1).
\end{align} 
Note that in regard of the degradedness assumptions (5) and (6), we must have $\alpha_B > \alpha_E$ (otherwise the secrecy capacity is equal to 0). That is $\alpha_B=\alpha_E + \alpha_{BE}$, with $\alpha_{BE}=\alpha_B -\alpha_E >0$. Substituting the value of $\alpha_B$ in \eqref{E71}, we get:
\begin{align}\label{E8}
\text{RHS of \eqref{E6} } =   \alpha_E \log\left(\frac{\alpha_B}{\alpha_E} \right) +  \alpha_{BE} \log(\alpha_{B}^{\Delta} x_i) +\kappa\left(\alpha_B,\alpha_E,\lambda_B,\lambda_E, \Delta \right)   -\gamma +\text{o}(1),
\end{align}   
which clearly converges to $-\infty$ as $i \to +\infty$, thus reaching a contradiction again. 

Summarizing both case 1 and case 2, we conclude that when $\lambda_B=0$, 0 cannot be an accumulation point of the set  $\mathcal{S}_{F_X^*} \cap B $ and hence its accumulation point is necessarily in $(0,+\infty)$ and therefore (88) holds for all $x \in (0,+\infty)$. This implication is itself not possible since   
\begin{align}
\underset{x \to 0}{\lim} \frac {c_s(x;F_{X}^{*}) -c_s(0;F_{X}^{*})} {x} = - \infty,
\end{align}
whereas $c_S(x;F_{X}^{*}) -c_S(0;F_{X}^{*}) = C_S- \gamma \mathcal{E} + \gamma x - c_S(0;F_{X}^{*}) \geq 0$ due to \eqref{eq-rateequivKKT-avg} for $\mu=0$, thus reaching a contradiction. Therefore, $\mathcal{S}_{F_X^*} \cap B$ cannot have an infinite number of elements and must be necessarily a finite set as claimed.

\begin{enumerate}
	\item [2)] 
	The support set of the optimal distribution $S_{F_{X}^{*}}$ is unbounded:
\end{enumerate}
This part of the proof is similar to the one of Theorem \ref{theo-4} and does not require a special treatment for the case $\lambda_B=0$.

\end{appendices}

\bibliographystyle{IEEEtran}
\bibliography{IEEEabrv,IEEEref}

\begin{thebibliography}{10}
\providecommand{\url}[1]{#1}
\csname url@samestyle\endcsname
\providecommand{\newblock}{\relax}
\providecommand{\bibinfo}[2]{#2}
\providecommand{\BIBentrySTDinterwordspacing}{\spaceskip=0pt\relax}
\providecommand{\BIBentryALTinterwordstretchfactor}{4}
\providecommand{\BIBentryALTinterwordspacing}{\spaceskip=\fontdimen2\font plus
\BIBentryALTinterwordstretchfactor\fontdimen3\font minus
  \fontdimen4\font\relax}
\providecommand{\BIBforeignlanguage}[2]{{%
\expandafter\ifx\csname l@#1\endcsname\relax
\typeout{** WARNING: IEEEtran.bst: No hyphenation pattern has been}%
\typeout{** loaded for the language `#1'. Using the pattern for}%
\typeout{** the default language instead.}%
\else
\language=\csname l@#1\endcsname
\fi
#2}}
\providecommand{\BIBdecl}{\relax}
\BIBdecl

\bibitem{5238736}
A.~Lapidoth, S.~M. Moser, and M.~A. Wigger, ``On the capacity of free-space
  optical intensity channels,'' \emph{{IEEE} Trans. Inf. Theory}, vol.~55,
  no.~10, pp. 4449--4461, Oct. 2009.

\bibitem{Uysal-Book}
S.~Arnon, J.~Barry, G.~Karagiannidis, R.~Schober, and M.~Uysal, \emph{Advanced
  Optical Wireless Communication Systems}, 1st~ed.\hskip 1em plus 0.5em minus
  0.4em\relax New York, NY, USA: Cambridge University Press, 2012.

\bibitem{6121996}
S.~M. Moser, ``Capacity results of an optical intensity channel with
  input-dependent {Gaussian} noise,'' \emph{IEEE Transactions on Information
  Theory}, vol.~58, no.~1, pp. 207--223, Jan. 2012.

\bibitem{21284}
A.~D. Wyner, ``Capacity and error exponent for the direct detection photon
  channel. i,'' \emph{{IEEE} Trans. Inf. Theory}, vol.~34, no.~6, pp.
  1449--1461, Nov. 1988.

\bibitem{217161}
S.~Shamai, ``Capacity of a pulse amplitude modulated direct detection photon
  channel,'' \emph{IEE Proceedings I - Communications, Speech and Vision}, vol.
  137, no.~6, pp. 424--430, Dec. 1990.

\bibitem{4729780}
A.~Lapidoth and S.~M. Moser, ``On the capacity of the discrete-time {Poisson}
  channel,'' \emph{{IEEE} Trans. Inf. Theory}, vol.~55, no.~1, pp. 303--322,
  Jan. 2009.

\bibitem{1435651}
T.~H. Chan, S.~Hranilovic, and F.~R. Kschischang, ``Capacity-achieving
  probability measure for conditionally {Gaussian} channels with bounded
  inputs,'' \emph{{IEEE} Trans. Inf. Theory}, vol.~51, no.~6, pp. 2073--2088,
  Jun. 2005.

\bibitem{8632953}
M.~{Cheraghchi} and J.~{Ribeiro}, ``Improved upper bounds and structural
  results on the capacity of the discrete-time {Poisson} channel,''
  \emph{{IEEE} Trans. Inf. Theory}, vol.~65, no.~7, Jul. 2019.

\bibitem{5773060}
A.~Lapidoth, J.~H. Shapiro, V.~Venkatesan, and L.~Wang, ``The discrete-time
  {Poisson} channel at low input powers,'' \emph{{IEEE} Trans. Inf. Theory},
  vol.~57, no.~6, pp. 3260--3272, Jun. 2011.

\bibitem{Martinez}
A.~Martinez, ``Spectral efficiency of optical direct detection,'' \emph{J. Opt.
  Soc. Am. B}, vol.~24, no.~4, pp. 739--749, Apr. 2007.

\bibitem{1056262}
M.~{Davis}, ``Capacity and cutoff rate for {Poisson}-type channels,''
  \emph{{IEEE} Trans. Inf. Theory}, vol.~26, no.~6, pp. 710--715, Nov. 1980.

\bibitem{Shannon1949}
C.~E. Shannon, ``Communication theory of secrecy systems,'' \emph{Bell Syst.
  Tech. J.}, vol.~28, no.~4, pp. 656--715, 1949.

\bibitem{Wyn75}
A.~D. Wyner, ``{The Wire-tap Channel},'' \emph{Bell Syst. Tech. J.}, vol.~54,
  no.~8, pp. 1355--1387, Jan. 1975.

\bibitem{1055892}
I.~Csiszar and J.~Korner, ``Broadcast channels with confidential messages,''
  \emph{{IEEE} Trans. Inf. Theory}, vol.~24, no.~3, pp. 339--348, May 1978.

\bibitem{bb_2011}
M.~Bloch and J.~Barros, \emph{Physical-Layer Security: From Information Theory
  to Security Engineering}.\hskip 1em plus 0.5em minus 0.4em\relax Cambridge
  University Press, 2011.

\bibitem{7164335}
O.~Ozel, E.~Ekrem, and S.~Ulukus, ``Gaussian wiretap channel with amplitude and
  variance constraints,'' \emph{{IEEE} Trans. Inf. Theory}, vol.~61, no.~10,
  pp. 5553--5563, Oct 2015.

\bibitem{8613368}
A.~{Dytso}, M.~{Egan}, S.~M. {Perlaza}, H.~V. {Poor}, and S.~S. {Shitz},
  ``Optimal inputs for some classes of degraded wiretap channels,'' in
  \emph{Proc. {IEEE} Information Theory Workshop}, Nov. 2018.

\bibitem{8399890}
M.~Soltani and Z.~Rezki, ``Optical wiretap channel with input-dependent
  {Gaussian} noise under peak- and average-intensity constraints,''
  \emph{{IEEE} Trans. Inf. Theory}, vol.~64, no.~10, pp. 6878--6893, Oct 2018.

\bibitem{6294444}
A.~Laourine and A.~B. Wagner, ``The degraded {Poisson} wiretap channel,''
  \emph{{IEEE} Trans. Inf. Theory}, vol.~58, no.~12, pp. 7073--7085, Dec 2012.

\bibitem{Smith71a}
J.~G. Smith, ``The information capacity of amplitude- and variance-constrained
  scalar {Gaussian} channels,'' \emph{Information and Control}, vol.~18, no.~3,
  pp. 203--219, April 1971.

\bibitem{524037}
D.~G. Luenberger, \emph{Optimization by Vector Space Methods}, 1st~ed.\hskip
  1em plus 0.5em minus 0.4em\relax USA: John Wiley and Sons, Inc., 1997.

\bibitem{6584947}
M.~{El-Halabi}, T.~{Liu}, and C.~N. {Georghiades}, ``Secrecy capacity per unit
  cost,'' \emph{{IEEE} J. Sel. Areas Commun.}, vol.~31, no.~9, pp. 1909--1920,
  2013.

\bibitem{1237131}
A.~{Lapidoth} and S.~M. {Moser}, ``Capacity bounds via duality with
  applications to multiple-antenna systems on flat-fading channels,''
  \emph{{IEEE} Trans. Inf. Theory}, vol.~49, no.~10, pp. 2426--2467, 2003.

\bibitem{1255549}
A.~Lapidoth, I.~E. Telatar, and R.~Urbanke, ``On wide-band broadcast
  channels,'' \emph{{IEEE} Trans. Inf. Theory}, vol.~49, no.~12, pp.
  3250--3258, Dec. 2003.

\bibitem{6685986}
J.~{Cao}, S.~{Hranilovic}, and J.~{Chen}, ``Capacity-achieving distributions
  for the discrete-time {Poisson} channel—part i: General properties and
  numerical techniques,'' \emph{{IEEE} Trans. Commun.}, vol.~62, no.~1, pp.
  194--202, Jan. 2014.

\bibitem{923716}
I.~C. {Abou-Faycal}, M.~D. {Trott}, and S.~{Shamai}, ``The capacity of
  discrete-time memoryless {Rayleigh}-fading channels,'' \emph{{IEEE} Trans.
  Inf. Theory}, vol.~47, no.~4, pp. 1290--1301, May 2001.

\bibitem{4529277}
A.~Khisti, A.~Tchamkerten, and G.~W. Wornell, ``Secure broadcasting over fading
  channels,'' \emph{{IEEE} Trans. Inf. Theory}, vol.~54, no.~6, pp. 2453--2469,
  Jun. 2008.

\bibitem{KornerBook}
I.~Csiszar and J.~Korner, \emph{Information theory : coding theorems for
  discrete memoryless systems}.\hskip 1em plus 0.5em minus 0.4em\relax Academic
  Press Akademiai Kiado New York : Budapest, 1981.

\bibitem{8421280}
J.~{Wang}, C.~{Liu}, J.~{Wang}, Y.~{Wu}, M.~{Lin}, and J.~{Cheng},
  ``Physical-layer security for indoor visible light communications: Secrecy
  capacity analysis,'' \emph{{IEEE} Trans. Commun.}, vol.~66, no.~12, pp.
  6423--6436, 2018.

\bibitem{5550280}
P.~{Harremoës}, O.~{Johnson}, and I.~{Kontoyiannis}, ``Thinning, entropy, and
  the law of thin numbers,'' \emph{{IEEE} Trans. Inf. Theory}, vol.~56, no.~9,
  pp. 4228--4244, 2010.

\end{thebibliography}

\end{document}